\documentclass[11pt]{article}
\usepackage[margin=1in]{geometry}
\usepackage{amsfonts}
\usepackage{latexsym,amssymb,amsfonts,graphicx}
\usepackage{amsmath,dsfont}
\usepackage[linewidth=1pt]{mdframed}
\usepackage{mathrsfs}
\usepackage[normalem]{ulem}
\usepackage{epsfig}
\usepackage{tabularx}
\usepackage{geometry}
\usepackage{verbatim}
\usepackage{enumitem}
\usepackage{framed}
\usepackage{algorithm2e}
\usepackage{float}

\usepackage{url}
\usepackage{bm}
\usepackage{color}
\usepackage{natbib}

\usepackage[unicode=true,colorlinks,citecolor=linkcolor,linkcolor=blue,urlcolor=blue]{hyperref}

\providecommand{\U}[1]{\protect\rule{.1in}{.1in}}

\newtheorem{thm}{Theorem}[section]

\newtheorem{Assumption}{\bf Assumption}

\newtheorem{lem}{Lemma}[section]

\newtheorem{prop}{Proposition}[section]
\newtheorem{rem}{Remark}[section]

\usepackage[bf,SL,BF]{subfigure}

\newenvironment{proof}[1][Proof]{\noindent\textbf{#1.} }{\ \rule{0.5em}{0.5em}}
\numberwithin{equation}{section}
\allowdisplaybreaks[4]

\definecolor{linkcolor}{rgb}{0,0,0.7}

\begin{document}

\title{Optimal portfolio under ratio-type periodic evaluation in stochastic factor models under convex trading constraints}
\author{Wenyuan Wang\thanks{School of Mathematics and Statistics, Fujian Normal University, Fuzhou, 350007, China; and School of Mathematical Sciences, Xiamen University, Xiamen, 361005, China. Email: wwywang@xmu.edu.cn}
\and
Kaixin Yan\thanks{School of Mathematical Sciences, Xiamen University, Xiamen, 361005, China. Email: kaixinyan@stu.xmu.edu.cn}
\and
Xiang Yu\thanks{Department of Applied Mathematics, The Hong Kong Polytechnic University, Kowloon, Hong Kong. E-mail: xiang.yu@polyu.edu.hk}
}

\date{\ }

\maketitle
\vspace{-0.6in}

\begin{abstract}
This paper studies a type of periodic utility maximization problem for portfolio management in incomplete stochastic factor models with convex trading constraints. The portfolio performance is periodically evaluated on the relative ratio of two adjacent wealth levels over an infinite horizon, featuring the dynamic adjustments in portfolio decision according to past achievements. Under power utility, we transform the original infinite horizon optimal control problem into an auxiliary terminal wealth optimization problem under a modified utility function. To cope with the convex trading constraints, we further introduce an auxiliary unconstrained optimization problem in a modified market model and develop the martingale duality approach to establish the existence of the dual minimizer such that the optimal unconstrained wealth process can be obtained using the dual representation. With the help of the duality results in the auxiliary problems, the relationship between the constrained and unconstrained models as well as some fixed point arguments, we derive and verify the optimal constrained portfolio process for the original problem over an infinite horizon.

\end{abstract}

\section{Introduction}
It has been well documented that for portfolio management by institutional managers, the long term portfolio performance often dictates the daily decision making for the fund management. Various long-run criteria for portfolio optimization have been proposed and studied in the literature, leading to several well-known stochastic control and optimizaiton problems over a large or an infinite horizon. The long term optimal growth rate, also named Kelly's criterion, has been popularized thanks to its tractability and simple financial implications. 
Moving beyond utility maximization, the so-called risk-sensitive portfolio management was introduced by \cite{BPK99} and \cite{FlemingS99} to encode diverse risk attitudes. Later, \cite{Pham03} formulated the long-run outperformance criterion as a large deviation probability control problem to incorporate the benchmark tracking for the fund management.  Recently, another new long-run portfolio criterion was proposed in \cite{TZ23}, which suggests to roll over the same utility sequentially for infinite periods. In  \cite{TZ23}, the utility during each period is generated by the difference between the wealth levels at two adjacent evaluation dates. In particular, the S-shaped utility with the same risk aversion parameter was adopted in \cite{TZ23} to accommodate both cases that the current wealth may outperform or underperform the benchmark level from the preceding evaluation date. This type of periodic evaluation can partially match with the practical exercises in the annual appraisal review in the fund industry. 
However, a key implication of the S-shaped utility is that while falling below the prior benchmark triggers risk-seeking, the optimal solution under loss aversion may prescribe abandoning risky investments altogether if continued exposure risks driving wealth towards bankruptcy, due to the severe disutility of ruin.

In the present paper, we are interested in generalizing the periodic evaluation in \cite{TZ23} by considering its variation when the relative performance is generated by the ratio of wealth processes at two evaluation dates. One advantage of the ratio-type relative performance is that the opportunity of investment cessation is ruled out from the admissible portfolio processes as the resulting optimal wealth process always stays above zero, yielding the optimal portfolio strategy truly supported over the infinite horizon. Moreover, in contrast to \cite{TZ23} that focus on the Black-Scholes model, we study the ratio-type periodic evaluation problem in incomplete market models with price dynamics driven by unhedgeable stochastic factors and general convex constraints on portfolio strategies.

On one hand, stochastic factor models have been extensively employed in optimal investment to integrate the asset predictability of the return and the stochastic volatility. For example, some early studies on the predictability of stock returns using the stochastic factor can be found in \cite{Fama77}, \cite{Ferson93}, \cite{Bre97}, \cite{Bre98}, \cite{Camp99},  \cite{Watchter} among others. The optimal investment under the stochastic volatility and unhedgeable risk or the combination of the stochastic returns and stochastic volatility was investigated in \cite{French87}, \cite{Kim96},  \cite{Scruggs}, \cite{FPS2000}, \cite{ZT01}, \cite{Pham02}, \cite{FH2003}, \cite{Chacko05}, \cite{Kraft05}, \cite{CH05}, \cite{Liu}, \cite{F13}, \cite{Hata18}, \cite{F20}, \cite{DL20}, to name a few.

On the other hand, convex trading constraints have also attracted a lot of attention in academic studies in portfolio management problems, which have been used as regulatory ways to stabilize stock prices in the financial market and as effective methods to minify the portfolio risk in fund management. Two typical examples are the short-selling constraint and borrowing constraint, which have been extensively studied in the literature. See, for instance, some empirical studies on the short-selling constraint and the economic influence in the financial market among \cite{AMP93}, \cite{CHS02}, \cite{ABCC04}, \cite{GS06}, \cite{BP13}, \cite{BJ13}, \cite{GM15}, and et al. See also some relevant studies on portfolio optimization with no short-selling such as the utility maximization (\cite{XS92(a)}, \cite{XS92(b)}, \cite{SH94}); the mean-variance portfolio optimization (\cite{LZ02}); the arbitrage theory in general market models (\cite{PS14}). For portfolio management problems under borrowing constraints, we refer to \cite{FlemingZ91}, \cite{Grossman92}, \cite{Te00}, \cite{Dyb11}, \cite{Lee19}, \cite{Jeon24} and references therein. In the main result of the paper, we consider the general convex trading constraints similar to \cite{CK92} and \cite{Esco23}, and aim to develop a general convex duality approach.  

Inspired by \cite{TZ23}, we first reformulate the infinite horizon optimization problem into an auxiliary terminal wealth optimization problem based on dynamic programming principle, leading to two subsequent questions: (1) The existence and the characterization of the optimal portfolio for the auxiliary terminal wealth utility maximization problem; (2) The fixed point argument to characterize the original value function and the verification proof of the constructed optimal portfolio for the original problem over the infinite horizon. The mixture of convex trading constraints and the stochastic factor model render both questions significantly more challenging. For technical convenience, we focus on the power utility for the periodic evaluation in this study.
Firstly, to address the auxiliary terminal wealth optimization problem under a modified utility function (see the induced utility in \eqref{2.8} and \eqref{problem3}), we need to develop the convex duality approach in the incomplete market model in the face of infinitely many dual processes (see \eqref{dual.pro}) due to unhedgeable stochastic factor process and convex trading constraints. We must overcome new challenges in defining a properly dual space and establishing the existence of the dual minimizer, since our modified utility \eqref{2.8} violates standard assumptions in the literature. Specifically:
(1) The simultaneous incorporation of stochastic factors and convex trading constraints necessitates a non-trivial modification of the dual space to align with the primal problem’s structure;
(2) When the risk aversion parameter $\alpha$ is negative, the Arrow-Pratt measure of relative risk aversion fails to be less than 1. To bypass this, we construct a fictitious market to ensure well-posedness of the dual optimization problem.
Secondly, due to the presence of the stochastic factor, the fixed point in  \eqref{A*.def} depends on the variable $y$ instead of the constant fixed point in \cite{TZ23}. More importantly, the verification proof for the constructed portfolio using the result from the auxiliary problem also becomes more technical as the we first need to establish the connection between the constrained optimization problem and the unconstrained optimization problem in a modified model and then leverage some duality arguments to verify the optimality condition.

The present paper contributes to the methodology in resolving the aforementioned issues. To address the existence of the optimal solution in the auxiliary terminal wealth optimization problem fulfilling the convex constraints, in the first step, we introduce and study the auxiliary optimization problem formulated as a combination of the stochastic factor model in \cite{CH05} and the general convex constraints in \cite{CK92}. That is, we introduce a fictitious market model without portfolio constraints described by \eqref{new-fic-model} such that the dual processes enjoy the local martingale characterization in \eqref{Z.process} depending on two parameter processes $\nu\in\mathcal{D}$ and $\eta\in\mathcal{H}$ defined by \eqref{DHspace}, referring as the constraint parameter process and the market completion parameter process respectively. We fix the constraint parameter process $\nu\in\mathcal{D}$ and first study the unconstrained primal optimization problem \eqref{primal:prob:new} under the modified utility function given in \eqref{2.8}. To this end, we establish the existence of the dual minimizer to the auxiliary dual problem in \eqref{dual.pro} over the market completion parameter processes $\eta\in\mathcal{H}$ and the Lagrange multipliers. In particular, due to the modified utility function, some new technical proofs are developed in the present paper for the case when the risk aversion parameter $\alpha<0$ (see Proposition \ref{existence.2}). The core innovation lies in constructing a fictitious market to circumvent the failure of the Arrow-Pratt measure of relative risk aversion being less than 1. In the second step, we first provide a sufficient condition on the constraint parameter process $\nu\in\mathcal{D}$ such that the unconstrained optimal portfolio process in the modified model coincide with the optimal solution to the constrained terminal wealth optimization problem \eqref{A*.def} in the original market model; see Proposition \ref{prop4.4}. Building upon this condition, we can further consider the dual problem \eqref{dual.pro.nu} associated to the constrained terminal wealth optimization problem \eqref{A*.def} formulated as a minimization problem over both parameter processes $\nu\in\mathcal{D}$ and $\eta\in\mathcal{H}$. By showing the existence of the dual optimizer pair $(\nu^*, \eta^*)$ and choosing the appropriate Lagrange multiplier, we can verify that the optimal wealth process induced by the dual representation satisfies the convex constraints and is therefore the optimal solution to the constrained problem \eqref{A*.def}; see Proposition \ref{primaldualcons}. To tackle the second question regarding an optimal portfolio for the original periodic evaluation problem, we show the existence of the unique fixed point by the contraction mapping in the space of bounded and continuous nonnegative functions; see Proposition \ref{fixed.power}. Based on the characterization of the dual minimizer in the auxiliary problem and the fixed point result, we are able to verify the optimality of the concatenated wealth process period by period (see Theorem \ref{thm3.1}) using some convergence arguments involving the dual processes and the duality relationship.    

The remainder of this paper is organized as follows. Section \ref{sec-formulation} introduces the stochastic factor model and formulates the ratio-type periodic evaluation problem under power utility and convex trading constraints over an infinite horizon. Based on dynamic programming principle, Section \ref{sec:auxiliary} reformulates the periodic evaluation problem as a fixed point problem where the operator becomes an auxiliary constrained terminal wealth optimization problem with a modified utility. To tackle this auxiliary problem, Section \ref{sec:parameter} introduces a parameterized market model and an unconstrained optimization problem, for which the dual problem is thoroughly studied. In Section \ref{sec:duality}, by passing from the unconstrained optimization problem to the constrained optimization problem, the convex duality theorem for the auxiliary constrained terminal wealth optimization problem is developed and the characterization of the optimal portfolio fulfilling the constraints is derived. Based on the convex duality results in the auxiliary problem, Section \ref{sec:fix} completes the proof of the fixed point and establishes the main result on the constructed optimal portfolio in a periodic manner for the original problem over an infinite horizon.

\section{Market Model and Problem Formulation}\label{sec-formulation}
Let $(\Omega,\mathcal{F},\{\mathcal{F}\}_{t\geq0},\mathbb{P})$ be a standard filtered probability space supporting a $n+1$-dimensional Brownian motion $W=(W^{\mathsf{T}}_{1},W_{2})^{\mathsf{T}}$ with $W^{\mathsf{T}}_{1}=(W^1_{1t},...,W^n_{1t})
_{t\geq0}\in\mathbb{R}^n$ and $W_{2}=(W_{2t})_{t\geq0}\in\mathbb{R}$. We consider a market $\mathcal{M}$ with $n$ risky assets, one risk-free asset and some external economic factor. {\color{black}  The price of the risk-free asset is given by $B_t=\exp\big(\int_0^tr(Y_s)ds\big)$,\,$t\geq0$, and the price process of the $i$-th risky asset $S^i=(S^i_t)_{t\geq0}$ is governed by 
\begin{eqnarray}\label{SDE2.1}
\frac{dS^i_t}{S^i_t}=\mu_i(Y_t) dt+\sum_{j=1}^n\sigma_{ij}(Y_t) dW^{{j}}_{1t},\quad t\geq 0,\quad S^i_0:=1,\,i=1,...,n,
\end{eqnarray}
where $r(\cdot)\in C^1({\mathbb{R}},\mathbb{R})$ is the interest rate, $ \mu(\cdot)=(\mu_1(\cdot),...,\mu_n(\cdot))^{\mathsf{T}}\in C^1({\mathbb{R}}, \mathbb{R}^n)$ is the mean return, and $\sigma(\cdot)\in C^1({\mathbb{R}},\mathbb{R}^{n\times n})$
is the volatility matrix. Here, we consider the incomplete market model where the interest rate, drift and volatility depend on an external stochastic factor process $Y_t$ described by
\begin{eqnarray}\label{SDE2.2}
    dY_t = b(Y_t)dt+\beta(Y_t)(q dW_{1t}+\sqrt{1-\|q\|^2} dW_{2t}),\quad t\geq0,
\end{eqnarray}
where $Y_0=y\in\mathbb{R},$ $q=(q_1,...,q_n)\in\mathbb{R}^n$ with $\|q\|\leq1$, $C^1({\mathbb{R}}, \mathbb{R})\ni\beta(\cdot)\neq0$, and, $b(\cdot)\in C^1({\mathbb{R}}, \mathbb{R})$ is the drift of the process $Y_t$. To guarantee that a unique solution to \eqref{SDE2.1} and \eqref{SDE2.2} exists respectively, we assume that $\mu(\cdot),\,\sigma(\cdot),\,b(\cdot),$ and $\beta(\cdot)$ satisfy the global Lipschitz and linear growth conditions
\begin{eqnarray}
\|f(x)-f(y)\|\leq C|x-y|,\quad x,y\in\mathbb{R},\nonumber\\
\|f(y)\|^2\leq C^2(1+y^2),\quad y\in\mathbb{R},\nonumber
\end{eqnarray}
where $C$ is a positive constant and $f$ represents $\mu,\sigma,b$ and $\beta$.}
Furthermore, we assume the strong non-degeneracy condition holds that
$$a^{\mathsf{T}}\sigma(y)\sigma(y)^{\mathsf{T}}a\geq \kappa_0\|a\|^2,\quad (y,a)\in\mathbb{R}\times\mathbb{R}^n,$$ for some $\kappa_0>0.$ 
Let us denote the Sharpe ratio
\begin{eqnarray}\label{xi}
\theta(y):=\sigma(y)^{-1}{(\mu(y)-r(y)\mathbf{1})},\quad y\in\mathbb{R},
\end{eqnarray}
where $\mathbf{1}=(1,...,1)^{\mathsf{T}}$. We assume that $\|\theta(\cdot)\|^2\leq M_0$ for some positive constant $M_0$, a condition that is satisfied, for instance, when both $\|\sigma^{-1}(y)\|$ and $\|\mu(y)-r(y)\mathbf{1}\|$ are bounded, or when $\|\sigma^{-1}(y)\|$ is asymptotically inversely proportional to a linear function of $y$ and $\|\mu(y)-r(y)\mathbf{1}\|$ has linear growth in $y$.

A trading strategy $\pi=(\pi^1_t,...,\pi^n_t)^{\mathsf{T}}_{t\geq0}$ is $\mathcal{F}_t$-progressively measurable with $\pi^i_t$ representing the proportion of the wealth invested in the $i$-th risky asset at time $t$. Under the given trading strategy $\pi$, we denote by $X^{\pi}=(X^{\pi}_t)_{t\geq0}$ the self-financing wealth process with the initial capital $X^{\pi}_0=x>0$ that satisfies the dynamics that 
\begin{eqnarray}
X_t^{\pi}=x+\sum_{i=1}^n\int_0^t{\pi^i_u}X_u^{\pi}\frac{dS^i_u}{S^i_u}+\int_0^tX^{\pi}_u(1-\mathbf{1}^{\mathsf{T}}\pi_u)\frac{dB_u}{B_u},\quad t\in[0,\infty).\nonumber
\end{eqnarray}

In the present paper, we are particularly interested in the portfolio constraint that $\pi$ takes values in a convex set $K$.

The dynamics of the controlled wealth $X^{\pi}$ (denoted by $X$ for short) can be rewritten as
\begin{eqnarray}\label{dX}
dX_t
\hspace{-0.3cm}&=&\hspace{-0.3cm}
[r(Y_t)+\pi_t^{\mathsf{T}}(\mu(Y_t)-r(Y_t)\mathbf{1})]X_tdt+X_t\pi_t^{\mathsf{T}}\sigma(Y_t) dW_{1t},\quad t\geq0.
\end{eqnarray}

Motivated by the recent study of \cite{TZ23}, we investigate an infinite-horizon utility maximization problem for wealth $(X_t)_{t \geq 0}$, where performance evaluation occurs periodically at a sequence of prescribed dates $(T_i)_{i \geq 0}$ with $T_0 := 0$ and $T_i \rightarrow \infty$ as $i \rightarrow \infty$. For simplicity, we assume equally spaced intervals $T_i = i\tau$ for $i \geq 0$ with some $\tau > 0$, implying that the portfolio is evaluated every $\tau$ units of time (e.g., monthly or annually).
In contrast to \cite{TZ23}, where portfolio performance is measured by the difference $X_{T_i} - \gamma X_{T_{i-1}}$ between consecutive epochs, we consider a relative wealth performance criterion based on a periodic ratio-type evaluation. Specifically, we define
\[
P_i := \frac{X_{T_i}}{(e^{{\kappa} \tau} X_{T_{i-1}})^{\gamma}},\quad i \geq 1,
\]
where $\gamma \in [0, 1]$ is a relative performance parameter and $\kappa > 0$ represents the benchmark growth rate used in the periodic evaluation. In applications, $\kappa$ may be chosen as the risk-free rate, a contractual hurdle rate, or an exogenously specified target growth rate.
This formulation captures the economic trade-off between relative and absolute wealth concerns: as $\gamma \rightarrow 1$, the agent prioritizes performance relative to the previous wealth benchmark $X_{T_{i-1}}$; conversely, as $\gamma \rightarrow 0$, the agent focuses primarily on the absolute value of current wealth. 
One important implication of this ratio-type structure is that it inherently imposes a positive wealth constraint, ensuring that the optimal wealth process remains strictly positive at all evaluation epochs over the infinite horizon. By contrast, the difference-type evaluation and $S$-shaped utility in \cite{TZ23} can effectively captures both risk aversion over gains and risk-seeking behavior over losses up to the bankruptcy time. In addition, our ratio-type performance measure induces a ``strong baseline amplification effect". Holding $X_{T_i}$ fixed, $P_i$ varies nonlinearly with the lagged wealth $X_{T_{i-1}}$ via the negative exponent $X_{T_{i-1}}^{-\gamma}$. Specifically,
\[
\frac{\partial P_i}{\partial X_{T_i}} = (e^{\kappa\tau} X_{T_{i-1}})^{-\gamma} \quad \text{and} \quad \frac{\partial P_i}{\partial X_{T_{i-1}}} = -\gamma \frac{P_i}{X_{T_{i-1}}}.
\]
Consequently, marginal incentives with respect to $X_{T_i}$ (and sensitivity to $X_{T_{i-1}}$) become significantly amplified when $X_{T_{i-1}}$ is small. Economically, this implies that in ``catch-up" states (i.e., periods of low $X_{T_{i-1}}$), the contemporaneous performance metric is mechanically inflated, potentially triggering strong state-dependent risk-taking incentives driven by the denominator effect.
For comparison, the difference-type evaluation in \cite{TZ23} depends on the previous-period wealth only in a  linear fashion. Thus, while both criteria exhibit path dependence, the ratio-type specification introduces a qualitatively stronger nonlinear baseline effect.

Mathematically speaking, our objective is to solve the optimal portfolio problem over an infinite horizon by employing a periodic evaluation of the relative performance defined by              
\begin{eqnarray}\label{utilitymax}
\sup_{X\in\mathcal{U}_0(x,y)}\mathbb{E}\left[\sum_{i=1}^{\infty}e^{-\rho T_i}U\left({P_{i}},Y_{T_i}\right)\right],
\end{eqnarray}
where $\rho>0$ is the agent’s subjective discount factor, and $\mathcal{U}_0(x,y)$ is the set of the wealth processes under admissible portfolio defined by 
\begin{eqnarray}\label{admissible_set}
\mathcal{U}_0(x,y)
\hspace{-0.3cm}&:=&\hspace{-0.3cm}
\bigg\{X^{x,y,\pi}:\pi
\text{ is a $(\mathcal{F}_t)_{t\geq 0}$-predictable and locally square-integrabale process}
\nonumber\\
\hspace{-0.3cm}&&\hspace{-1cm}
\text{such that }
\pi\in K,\text{ and }X^{x,y,\pi}_t=x+\sum_{i=1}^n\int_0^t{\pi^i_u}X^{x,y,\pi}_u\frac{dS^i_u}{S^i_u}+\int_0^tX^{x,y,\pi}_u(1-\mathbf{1}^{\mathsf{T}}\pi_u)\frac{dB_u}{B_u}>0
\nonumber\\
\hspace{-0.3cm}&&\hspace{-1cm} \text{for } t\geq 0,
\color{black}\text{ and } 
\sum_{i=1}^{\infty}e^{-\rho T_i}\mathbb{E}\left[\left(U\left(\frac{X^{x,y,\pi}_{T_i}}{\big(e^{{\kappa}\tau}X^{x,y,\pi}_{T_{i-1}}\big)^{\gamma}},Y_{T_i}\right)\right)_{-}\right]<\infty
\bigg\},\quad x\in\mathbb{R}_+,
\end{eqnarray}
with $a_{-}=\max\{-a,0\}$ and $K$ being the closed, convex set, and $Y_0=y\in\mathbb{R}$. We also note that the last condition is of the integrability type, which is needed to ensure the well-posedness of the periodic portfolio evaluation and optimization problem in Section \ref{sec:fix}.

In the present paper, we only focus on  the power utility function $U(x,y)=\frac{1}{\alpha}x^{\alpha}h(y)$ with $\alpha\in(-\infty,0)\cup(0,1)$  with $\mathbb{R}\ni y\mapsto h(y)\in\mathbb{R}_+$ being a continuous function defined on the factor level $y$.  Similar to \cite{ZT01}, it is assumed that $m\leq h(y) \leq 1$, $y\in\mathbb{R}$, for some constant $m\in(0,1)$.

\section{Auxiliary Constrained Terminal Wealth Optimization Problem with Modified Utility}\label{sec:auxiliary}

Let us denote the value function of the periodic evaluation problem by
\begin{eqnarray}\label{problem}
V(x,y)
\hspace{-0.3cm}&:=&\hspace{-0.3cm}
\sup_{X\in\mathcal{U}_0(x,y)}\mathbb{E}\left[\frac{1}{\alpha}\sum_{i=1}^{\infty}e^{-\rho T_i}\left(\frac{X_{T_i}}{\left(e^{{\kappa} T_1}X_{T_{i-1}}\right)^{\gamma}}\right)^{\alpha}h({Y_{T_i}})\right],
\end{eqnarray}
where $\mathcal{U}_0(x,y)$ is given by \eqref{admissible_set}.
By the Markov property of the stochastic factor model, one can easily derive the following dynamic programming principle that
\begin{eqnarray}\label{ddp}
V(x,y)=\sup_{X\in\mathcal{U}_0(x,y)}\mathbb{E}\left[\frac{1}{\alpha}e^{-\rho T_1}\left(\frac{X_{T_1}}{e^{\gamma{\kappa} T_1}x^{\gamma}}\right)^{\alpha}h({Y_{T_1}})+e^{-\rho T_1}V(X_{T_1},Y_{T_1})\right].
\end{eqnarray}

For the well-posedness of the problem, the following standing assumption is imposed throughout the paper.
\begin{Assumption}
\label{ass1}
The interest rate function $r(\cdot)\in C^1(\mathbb{R}, \mathbb{R})$ is bounded from above by  $\overline{r}>0$ and bounded below by $\underline{r}>0$.
The model parameters satisfy $\rho>\zeta(\alpha(1-\gamma)){\vee0}$ with $\zeta$ being defined as $\zeta(x):=\overline{r}x+x M_0/2(1-x)$.
\end{Assumption}

The following result gives the upper and lower bounds of the value function $V$.

\begin{prop}\label{pro:bound:V}
It holds that
\begin{align}
\frac{me^{(\underline{r}\alpha-\rho-\kappa\gamma\alpha)\tau}}{\alpha(1-e^{-(\rho-\underline{r}\alpha(1-\gamma))\tau})}x^{\alpha(1-\gamma)}&\leq V(x,y)\leq \frac{e^{(\zeta(\alpha)-\rho-\kappa\gamma\alpha)\tau}}{\alpha(1-e^{(\zeta(\alpha(1-\gamma))-\rho)\tau})}x^{\alpha(1-\gamma)},\quad \alpha\in(0,1),\nonumber\\
\frac{e^{(\underline{r}\alpha-\rho-\kappa\gamma\alpha)\tau}}{\alpha(1-e^{-(\rho-\underline{r}\alpha(1-\gamma))\tau})}x^{\alpha(1-\gamma)}&\leq V(x,y)\leq \frac{me^{(\zeta(\alpha)-\rho-\kappa\gamma\alpha)\tau}}{\alpha(1-e^{(\zeta(\alpha(1-\gamma))-\rho)\tau})}x^{\alpha(1-\gamma)},\quad \alpha\in(-\infty,0),
\end{align}
where $m\in(0,1)$ is the lower bound of the function $h(\cdot)$.
\end{prop}
\begin{proof}
We first assume that $\alpha\in(0,1)$. 
The lower bound of $V(x,y)$ can be derived by noting that $X=(xe^{\int_0^tr(Y_s)ds})_{t\geq 0}$ is an admissible portfolio in $\mathcal{U}_0(x,y)$, and
\begin{eqnarray}
\hspace{-0.3cm}&&\hspace{-0.3cm}
\sup_{X\in\mathcal{U}_{0}(x,y)}\mathbb{E}\left[\sum_{i=1}^{n}e^{-\rho T_i}\frac{1}{\alpha}\bigg(\frac{X_{T_i}}{(e^{{\kappa}\tau}X_{T_{i-1}})^{\gamma}}\bigg)^{\alpha}h(Y_{T_i})\right]
\nonumber\\
\hspace{-0.3cm}&=&\hspace{-0.3cm}
\sup_{X\in\mathcal{U}_{0}(x,y)}\sum_{i=1}^{n}e^{-\rho T_{i}}e^{-{\kappa}\gamma\alpha\tau}\mathbb{E}\left[\mathbb{E}\left[\frac{1}{\alpha}\bigg(\frac{X_{T_i}}{X_{T_{i-1}}}\bigg)^{\alpha}X_{T_{i-1}}^{\alpha(1-\gamma)}{h(Y_{T_{i}})}\bigg|\mathcal{F}_{T_{i-1}}\right]\right]
\nonumber\\
\hspace{-0.3cm}&\geq&\hspace{-0.3cm}
\sum_{i=1}^{n}e^{-\rho T_{i}}e^{-{\kappa}\gamma\alpha\tau}m\frac{1}{\alpha}e^{\underline{r}\alpha\tau}e^{\alpha(1-\gamma)\underline{r}T_{i-1}}x^{\alpha(1-\gamma)},
\end{eqnarray}
where in the last inequality we have used the fact that the lower bound of $h(\cdot)$ is $m$. The desired lower bound in Proposition \ref{pro:bound:V} follows after sending $n\rightarrow\infty$ on both sides.

We next turn to derive the upper bound of $V$. 
For any $X\in\mathcal{U}_0(x,y)$ and $n\geq1$, using the upper bound of $h(\cdot)$, one has
\begin{eqnarray}
\label{3.11.v0}
\hspace{-0.3cm}&&\hspace{-0.3cm}
\mathbb{E}\left[\sum_{i=1}^{n}e^{-\rho T_i}\frac{1}{\alpha}\bigg(\frac{X_{T_i}}{(e^{{\kappa}\tau}X_{T_{i-1}})^{\gamma}}\bigg)^{\alpha}h(Y_{T_i})\right]
\nonumber\\
\hspace{-0.3cm}&=&\hspace{-0.3cm}
\sum_{i=1}^{n}e^{-\rho T_{i}}e^{-{\kappa}\gamma\alpha\tau}\mathbb{E}\left[\mathbb{E}\left[\frac{1}{\alpha}\bigg(\frac{X_{T_i}}{X_{T_{i-1}}}\bigg)^{\alpha}X_{T_{i-1}}^{\alpha(1-\gamma)}{h(Y_{T_{i}})}\bigg|\mathcal{F}_{T_{i-1}}\right]\right]
\nonumber\\
\hspace{-0.3cm}&\leq&\hspace{-0.3cm}
\sum_{i=1}^{n}e^{-\rho T_{i}}e^{-{\kappa}\gamma\alpha\tau}\mathbb{E}\left[\mathbb{E}\left[\frac{1}{\alpha}\bigg(\frac{X_{T_i}}{X_{T_{i-1}}}\bigg)^{\alpha}\bigg|\mathcal{F}_{T_{i-1}}\right]X_{T_{i-1}}^{\alpha(1-\gamma)}\right]
\nonumber\\
\hspace{-0.3cm}&\leq&\hspace{-0.3cm}
\sum_{i=1}^{n}e^{-\rho T_{i}}e^{-{\kappa}\gamma\alpha\tau}
\mathbb{E}\left[\left(\sup_{X\in\mathcal{U}_{T_{i-1}}(1,Y_{T_{i-1}})}\mathbb{E}\left[\frac{1}{\alpha}X_{T_i}^{\alpha}|\mathcal{F}_{T_{i-1}}\right]\right)X_{T_{i-1}}^{\alpha(1-\gamma)}\right],
\end{eqnarray}
where $\mathcal{U}_{T_{i-1}}(x,y)$ is defined by
\begin{eqnarray}\label{U_T}
\mathcal{U}_{T_{i-1}}(x,y)
\hspace{-0.3cm}&:=&\hspace{-0.3cm}
\bigg\{X^{\pi}:\pi
\text{ is $(\mathcal{F}_t)_{t\geq 0}$-progressively measurable and locally square-integrable such that}
\nonumber\\
\hspace{-0.3cm}&&\hspace{-2cm}
\pi\in K,\text{ and }X^{\pi}_t=X^{\pi}_{T_{i-1}}+\sum_{k=1}^n\int_{T_{i-1}}^t{\pi^k_u}X^{\pi}_u\frac{dS^k_u}{S^k_u}+\int_{T_{i-1}}^tX^{\pi}_u(1-\mathbf{1}^{\mathsf{T}}\pi_u)\frac{dB_u}{B_u}>0 \text{ for } t\geq T_{i-1},
\nonumber\\
\hspace{-0.3cm}&&\hspace{-2cm}
\color{black}\text{ and } 
\sum_{k=i}^{\infty}e^{-\rho T_i}\mathbb{E}\left[\left(U\left(\frac{X^{x,y,\pi}_{T_k}}{\big(e^{{\kappa}\tau}X^{x,y,\pi}_{T_{k-1}}\big)^{\gamma}},Y_{T_k}\right)\right)_{-}\right]<\infty\,\text{with }\,X^{\pi}_{T_{i-1}}=x,Y_{T_{i-1}}=y
\bigg\},
\end{eqnarray}
with $(x,y)\in\mathbb{R}_+\times\mathbb{R}$.
Note that
\begin{eqnarray}\label{inv.pro}
\sup_{X\in\mathcal{U}_{T_{i-1}}(1,Y_{T_{i-1}})}\mathbb{E}\left[\frac{1}{\alpha}X_{T_i}^{\alpha}\bigg|\mathcal{F}_{T_{i-1}}\right]
\hspace{-0.3cm}&\leq&\hspace{-0.3cm}
\sup_{X\in\mathcal{A}_{T_{i-1}}(1,Y_{T_{i-1}})}\mathbb{E}\left[\frac{1}{\alpha}X_{T_i}^{\alpha}\bigg|\mathcal{F}_{T_{i-1}}\right],
\end{eqnarray}
where $\mathcal{A}_{T_{i-1}}(1,Y_{T_{i-1}})$ is defined the same as $\mathcal{U}_{T_{i-1}}(1,Y_{T_{i-1}})$ with $K$ replaced by $\mathbb{R}^n$.
This replacement implies that the right-hand side optimization problem degenerates to the unconstrained case considered in \cite{CH05}, where only stochastic factors are present without convex trading restrictions.
Hence, by \eqref{inv.pro} and (4.18) in \cite{CH05}, 
we have 
\begin{eqnarray}
\sup_{X\in\mathcal{U}_{T_{i-1}}(1,Y_{T_{i-1}})}\mathbb{E}\left[\frac{1}{\alpha}X_{T_i}^{\alpha}\bigg|\mathcal{F}_{T_{i-1}}\right]\leq
\sup_{X\in\mathcal{A}_{T_{i-1}}(1,Y_{T_{i-1}})}\mathbb{E}\left[\frac{1}{\alpha}X_{T_i}^{\alpha}\bigg|\mathcal{F}_{T_{i-1}}\right]
\leq 
\frac{1}{\alpha}e^{\zeta(\alpha)\tau}.
\end{eqnarray}
Similarly, it holds that
\begin{eqnarray}
\label{3.13.v0}
\sup_{X\in\mathcal{U}_0(x,y)}\mathbb{E}\left[\frac{1}{\alpha(1-\gamma)}X_{T_{i-1}}^{\alpha(1-\gamma)}\right]\leq \frac{x^{\alpha(1-\gamma)}}{\alpha(1-\gamma)}e^{\zeta(\alpha(1-\gamma))(i-1)\tau}.
\end{eqnarray}
Combining \eqref{3.11.v0}-\eqref{3.13.v0}, we obtain that
\begin{eqnarray}
\hspace{-0.3cm}&&\hspace{-0.3cm}
\mathbb{E}\left[\frac{1}{\alpha}\sum_{i=1}^{n}e^{-\rho T_i}\bigg(\frac{X_{T_i}}{(e^{{\kappa}\tau}X_{T_{i-1}})^{\gamma}}\bigg)^{\alpha}h(Y_{T_i})\right]
\nonumber\\
\hspace{-0.3cm}&\leq&\hspace{-0.3cm}
\sum_{i=1}^{n}e^{-\rho T_{i}}e^{-{\kappa}\gamma\alpha\tau}
\frac{1}{\alpha}\mathbb{E}\left[\left(\sup_{X\in\mathcal{U}_{T_{i-1}}(1,Y_{T_{i-1}})}\mathbb{E}\left[X_{T_i}^{\alpha}|\mathcal{F}_{T_{i-1}}\right]\right)X_{T_{i-1}}^{\alpha(1-\gamma)}\right]
\nonumber\\
\hspace{-0.3cm}&\leq&\hspace{-0.3cm}
\sum_{i=1}^{n}e^{-\rho T_{i}}e^{-{\kappa}\gamma\alpha\tau}e^{\zeta(\alpha)\tau}\sup_{X\in\mathcal{U}_0(x,y)}\mathbb{E}\left[\frac{1}{\alpha}X_{T_{i-1}}^{\alpha(1-\gamma)}\right]
\nonumber\\
\hspace{-0.3cm}&\leq&\hspace{-0.3cm}
\sum_{i=1}^{n}e^{-\rho T_{i}}e^{-{\kappa}\gamma\alpha\tau}\frac{1}{\alpha}e^{\zeta(\alpha)\tau}e^{\zeta(\alpha(1-\gamma))(i-1)\tau}x^{\alpha(1-\gamma)}.
\end{eqnarray}
The desired upper bound in Proposition \ref{pro:bound:V} follows after sending $n\rightarrow\infty$ on both sides.
The proof for the case $\alpha\in(-\infty,0)$ can be easily modified, and it is hence omitted.  
\end{proof}

In view of the scaling property of the utility function $U(x,y)=x^{\alpha}U(1,y)$ and the fact that the value function $V(x,y)$ is bounded with respect to $y$, we heuristically conjecture that our value function takes the form of $V(x,y)={\frac{1}{\alpha}e^{-{\kappa}\tau\gamma\alpha}}A^*(y)x^{\alpha(1-\gamma)}$ for some continuous, bounded, and non-negative function $A^*(\cdot)$. Substituting this form of $V$ into \eqref{ddp} and then dividing both sides by ${\frac{1}{\alpha }e^{-{\kappa}\tau\gamma\alpha}x^{\alpha(1-\gamma)}}$, one obtains
\begin{eqnarray}\label{A*.def}
A^*(y)
\hspace{-0.3cm}&=&\hspace{-0.3cm}
\alpha\sup_{X\in\mathcal{U}_0(x,y)}\mathbb{E}\left[e^{-\rho T_1}\frac{1}{\alpha}\left(\frac{X_{T_1}}{x}\right)^{\alpha}h({Y_{\tau}})+e^{-\rho T_1}\frac{1}{\alpha}A^*(Y_{T_1})\left(\frac{X_{T_1}}{x}\right)^{\alpha(1-\gamma)}\right]
\nonumber\\
\hspace{-0.3cm}&=&\hspace{-0.3cm}
\alpha\sup_{X\in\mathcal{U}_0(1,y)}\mathbb{E}\left[e^{-\rho\tau}\frac{1}{\alpha}X_{\tau}^{\alpha}h({Y_{\tau}})+e^{-\rho \tau}\frac{1}{\alpha}A^*(Y_{\tau})X_{\tau}^{\alpha(1-\gamma)}\right].
\end{eqnarray}
Hence, under power utility, the characterization of $V$ now simplifies to the characterization of the unknown, non-negative, continuous, and bounded function $A^*(\cdot)$. More importantly, the unknown function $A^*(\cdot)$ will be proven as the the fixed point of a contraction operator defined on the function space ${C}_b^{+}(\mathbb{R})$ consisting of all continuous, bounded, and non-negative functions on $\mathbb{R}$. To this end, we will first establish the existence of the optimizer to the auxiliary one-period terminal wealth optimization problem \eqref{A*.def} for a fixed $A^*$, and then show the existence of the unique fixed point $A^*$ to the operator.

Let us then consider the modified utility function
$\mathbb{R}_+\times\mathbb{R}\ni (x,y)\mapsto h_A(x,y)\in \mathbb{R}$ given by
\begin{eqnarray}\label{2.8}
h_A(x,y)
:=\frac{1}{\alpha}x^{\alpha}h(y)+{\frac{1}{\alpha}}A(y)x^{\alpha(1-\gamma)},\quad (x,y)\in \mathbb{R}_+\times\mathbb{R},
\end{eqnarray}
where $A(\cdot)\in C^+_b(\mathbb{R})$ is viewed as a parameter of the function $h_A$.  As a preparation for the main result, we first derive some preliminary properties of the function $h_A$, which play pivotal roles in later proofs.

\begin{lem}
\label{lem2.1}
The function $h_A(x,y)$ defined by \eqref{2.8} is strictly increasing and strictly concave in $x\in\mathbb{R}_+$, and it holds that $\frac{\partial}{\partial x}h_A(0+,y)=\infty$ and $\frac{\partial}{\partial x}h_A(\infty,y)=0$ for any $y\in\mathbb{R}$. 
In addition, there exist some constants $\vartheta\in(0,1)$ and $\varrho\in(1,\infty)$ such that 
\begin{eqnarray}\label{h_a'>h_a'}
\vartheta \frac{\partial}{\partial x}h_A(x,y)\geq \frac{\partial}{\partial x}h_A(\varrho x,y),\quad (x,y)\in\mathbb{R}_+\times\mathbb{R}.
\end{eqnarray} 
Furthermore, when $\alpha\in(0,1)$, there exist some constants $\kappa_1\in(0,\infty)$ and $\rho_1\in(0,1)$ such that
\begin{eqnarray}\label{0<h}
0< h_A(x,y)\leq \kappa_1(1+x^{\rho_1}),\quad (x,y)\in\mathbb{R}_+\times\mathbb{R};
\end{eqnarray}
and when $\alpha\in(-\infty,0)$, there exist some constants $\kappa_2\in(-\infty,0)$ and $\rho_2\in(-\infty,0)$ such that
\begin{eqnarray}\label{0>h}
0>h_A(x,y)\geq \kappa_2(1+x^{\rho_2}),\quad (x,y)\in\mathbb{R}_+\times\mathbb{R}.
\end{eqnarray}
\end{lem}
\begin{proof}
The results follow from elementary calculus. Recall that the function $h(\cdot)\in[m,1]$ with $m\in(0,1)$ and $A(\cdot)\in {C}_b^{+}(\mathbb{R})$. Differentiating twice the both sides of \eqref{2.8} gives
\begin{eqnarray}
\label{2.28.v0}
\frac{\partial}{\partial x}h_A(x,y)= x^{\alpha-1}{h(y)}+A(y)(1-\gamma)x^{\alpha(1-\gamma)-1},\quad (x,y)\in\mathbb{R}_+\times\mathbb{R},
\end{eqnarray}
and
\begin{eqnarray}
\frac{\partial^2}{\partial x^2}h_A(x,y)=(\alpha-1) x^{\alpha-2}{h(y)}+A(y)(1-\gamma)(\alpha-1-\alpha\gamma)x^{\alpha(1-\gamma)-2},\quad (x,y)\in\mathbb{R}_+\times\mathbb{R}.
\end{eqnarray}
Due to the facts of $\gamma\in[0,1]$ and {$\alpha\in(-\infty,0)\cup(0,1)$}, we have $\frac{\partial}{\partial x}h_A(x,y)>0$ and $\frac{\partial^2}{\partial x^2}h_A(x,y)<0$ for any $(x,y)\in\mathbb{R}_+\times\mathbb{R}$, implying that the function $h_A(x,y)$ is strictly concave and strictly increasing  with respect to $x$ on $\mathbb{R}_+$. Furthermore, by \eqref{2.28.v0}, one can easily get $\frac{\partial}{\partial x}h_A(0+,y)=\infty$ and $\frac{\partial}{\partial x}h_A(\infty,y)=0$ for any $y\in\mathbb{R}$. To prove the second claim, for any constant $\varrho\in(1,\infty)$, one can take $\vartheta=\varrho^{\alpha-1}\vee \varrho^{\alpha(1-\gamma)-1}\in(0,1)$. Then 
\begin{align}
    \vartheta \frac{\partial}{\partial x}h_A(x,y)&=\vartheta x^{\alpha-1}h(y)+A(y)(1-\gamma)\vartheta x^{\alpha(1-\gamma)-1}
    \nonumber\\
    &\geq (\varrho x)^{\alpha-1}h(y)+A(y)(1-\gamma)(\varrho x)^{\alpha(1-\gamma)-1}
    \nonumber\\
    &= \frac{\partial}{\partial x}h_A(\varrho x,y),\quad (x,y)\in\mathbb{R}_+\times\mathbb{R}.\nonumber
\end{align}
Finally, when $\alpha\in(0,1)$, taking $\kappa_1=\frac{2}{\alpha}\max\{1,\sup_{y\in\mathbb{R}}A(y)\}\in(0,
\infty)$ and $\rho_1=\alpha\in(0,1)$ yields \eqref{0<h}; when $\alpha\in(-\infty,0)$, taking $\kappa_2=\frac{2}{\alpha}\max\{1,\sup_{y\in\mathbb{R}}A(y)\}\in(-\infty,0)$ and $\rho_2=\alpha\in(-\infty,0)$ yields \eqref{0>h}.
The proof is then complete.
\end{proof}

\section{Auxiliary Parameterized Model and
Unconstrained Problem}\label{sec:parameter}

In the presence of both stochastic factor and portfolio constraints, we combine some techniques in \cite{CH05} and \cite{CK92} to employ the convex duality approach and formulate a fictitious market model with an auxiliary unconstrained optimization problem and the associated dual problem. 
Let us first outline the plan of this section: Proposition \ref{prop2.1} establishes the primal-dual equivalence for the auxiliary unconstrained problem \eqref{primal:prob:new}, equivalently \eqref{problem3}, in the fictitious market with a fixed constraint parameter \(\nu\), provided that the dual optimizer of \eqref{dual.pro} exists and satisfies the budget identity. Propositions \ref{existence.1} and \ref{existence.2} then verify the required existence of the dual optimizer for \eqref{dual.pro} in the two cases \(\alpha\in(0,1)\) and \(\alpha<0\), respectively. Finally, Lemma \ref{lemma4.1} translates the dual optimizer of \eqref{dual.pro} into a financed wealth process in the fictitious market. Together, these results solve the one-period auxiliary problem \eqref{A*.def} after \(\nu\) is fixed, preparing for the optimization over \(\nu\) in Section \ref{sec:duality}.

We now consider a modified market model using the original market model together with two parameter processes. Let us first introduce some notations. Denote by $\mathcal{H}_0$ the set of progressively measurable processes $(\nu_t)_{t\geq0}$ with values in $\mathbb{R}^n$ such that $\mathbb{E}[\int_0^{\tau}\|\nu_s\|^2ds]<\infty$. 
We introduce the classes of processes
\begin{align}
    \mathcal{D}&:=\left\{\nu\in\mathcal{H}_0:\,\mathbb{E}\left[\int_0^{\tau}\delta(\nu_s)ds\right]<\infty\text{ and }
    \mathbb{E}\left[\exp\left(\frac{1}{2}\int_0^t\|\theta^{\nu}(Y_s)\|^2ds\right)\right]<\infty\text{ for all }t\in[0,\tau]
    \right\},\nonumber\\
    \mathcal{H}&:=\left\{\eta\in\mathcal{H}_0:
    \mathbb{E}\left[\exp\left(\frac{1}{2}\int_0^t\|\eta_s\|^2ds\right)\right]<\infty\text{ for all }t\in[0,\tau]
    \right\},\label{DHspace}
\end{align}
where $$\delta(x)\equiv \delta(x|K):=\sup_{\pi\in K}(-\pi^{\mathsf{T}}x),\quad x\in \mathbb{R}^{n},$$
is the support function of the nonempty, closed, convex set $-K$ in $\mathbb{R}^n$, and, 
\begin{eqnarray}
\theta^{\nu}(Y_t):={\sigma^{-1}(Y_t)}(\mu(Y_t)-r(Y_t)\mathbf{1}+\nu_t)=\theta(Y_t)+\sigma^{-1}(Y_t)\nu_t,\nonumber
\end{eqnarray}
with $\nu=(\nu_t)_{t\in[0,\tau]}\in\mathcal{D}$.
Denote the effective domain of $\delta$ by $$\Tilde{K}:=\{x\in\mathbb{R}^n;\delta(x|K)<\infty\},$$
which is a convex cone (called the barrier cone of $-K$). Furthermore, we follow \cite{CK92} to assume that the function $\delta(\cdot|K)$ is continuous on $\Tilde{K}$ and bounded below on $\mathbb{R}^n$ that $\delta(x|K)\geq \delta_0$ for all $x\in\mathbb{R}^n$ and some $\delta_0\in\mathbb{R}$. In addition, the following subadditivity property holds
\begin{eqnarray}\label{subad.property}
\delta(x+y)\leq \delta(x)+\delta(y),\quad x,y\in\mathbb{R}^n.
\end{eqnarray}

For any given $\nu\in\mathcal{D}$, we consider a new financial market $\mathcal{M}_{\nu}$ with one bond and $n$ stocks
\begin{align}
dB^{\nu}_t&=B^{\nu}_t(r(Y_t)+\delta(\nu_t))dt, \nonumber\\
dS^{i,\nu}_t&=S^{i,\nu}_t\left[(\mu_i(Y_t)+\nu^i_t+\delta(\nu_t))dt+\sum_{j=1}^{n}\sigma_{ij}(Y_t)dW^j_{1t}\right].\label{new-fic-model}
\end{align}
Furthermore, for given $\nu\in\mathcal{D},\eta\in\mathcal{H}$ we denote
\begin{align}
Z_t^{\nu,\eta}&:=\exp\left(-\int_0^t[\theta^{\nu}(Y_s)^{\mathsf{T}}dW_{1s}-\eta_sdW_{2s}]-\frac{1}{2}\int_0^t[\|\theta^{\nu}(Y_s)\|^2+\|\eta_s\|^2]ds\right),\label{Z.process}\nonumber\\
W_{1t}^{\nu}&:=W_{1t}+\int_0^{t}\theta^{\nu}(Y_s)ds\quad\text{and}\quad W_{2t}^{\eta}:=W_{2t}-\int_0^{t}\eta_sds.
\end{align}

We can interpret the parameter process $\nu \in \mathcal{D}$ as the constraint parameter that transforms the original constrained market into a fictitious unconstrained market. Specifically, when $\nu$ satisfies some conditions, the unconstrained optimal solution in the fictitious market model \eqref{new-fic-model} coincides with the optimal solution in the original market model that satisfies the convex trading constraints $K$. 
We also interpret the parameter process $\eta\in\mathcal{H}$ as the market completion parameter because the dual problem can be formulated as a minimization problem over $\eta\in\mathcal{H}$, and the dual optimizer will provide the optimal wealth for the auxiliary constrained terminal wealth optimization problem in stochastic factor models via the duality relationship. In addition, the dynamics of the two processes $(Y_t)_{t\geq 0}$ and $(Z_{t}^{\nu,\eta})_{t\geq 0}$ can be rewritten as
\begin{align}
    dY_t&=\left(b(Y_t)-\beta(Y_t)(q\theta^{\nu}(Y_t)-\sqrt{1-\|q\|^2} \eta_t)\right)dt + \beta(Y_t)(q dW_{1t}^{\nu}+\sqrt{1-\|q\|^2}dW_{2t}^{\eta}),\quad t\geq 0,\nonumber\\
dZ_t^{\nu,\eta}&=Z_t^{\nu,\eta}\left(\left(\|\theta^{\nu}(Y_t)\|^2+\|\eta_{t}\|^2\right)dt-\theta^{\nu}(Y_t)^{\mathsf{T}}dW_{1t}^{\nu}+\eta_tdW_{2t}^{\eta}\right),\quad t\geq 0.
\label{z.process}
\end{align}
The wealth process $X^{\nu}\equiv X^{x,y,\pi,\nu}$, under a given portfolio $\pi$ in $\mathcal{M}_{\nu}$, satisfies
\begin{eqnarray}\label{eq:Xnu}
dX^{\nu}_t 
\hspace{-0.3cm}&=&\hspace{-0.3cm} 
\left[r(Y_t)+\delta(\nu_t)+\pi_t^{\mathsf{T}}(\mu(Y_t)+\nu_t-r(Y_t)\mathbf{1})\right]X^{\nu}_tdt+X^{\nu}_t\pi^{\mathsf{T}}_t\sigma(Y_t)dW_{1t}.
\end{eqnarray}
One can also verify that 
\begin{eqnarray}\label{dX/B}
d\left[\frac{X^{\nu}_t}{B^{\nu}_t}\right]=\frac{X^{\nu}_t}{B^{\nu}_t}\pi^{\mathsf{T}}_t\sigma(Y_t)dW^{\nu}_{1t}, \quad t\geq 0.
\end{eqnarray}
Additionally, an application of It\^o's formula to the product of processes $Z^{\nu,\eta}$ and $X^{\nu}/B^{\nu}$ yields that
\begin{eqnarray}\label{supermaringale}
\frac{X^{\nu}_tZ_t^{\nu,\eta}}{B^{\nu}_t}=x+\int_0^t\frac{X^{\nu}_sZ_s^{\nu,\eta}}{B^{\nu}_s}\left[(\sigma(Y_s)^{\mathsf{T}}\pi_s-\theta^{\nu}(Y_s))^{\mathsf{T}}dW_{1s}+\eta_sdW_{2s}\right], \quad t\geq 0.
\end{eqnarray}
If $(X^{\nu}_{t})_{t\geq 0}$ is a wealth process under an admissible portfolio control, by \eqref{supermaringale}, the process $X^{\nu}Z^{\nu,\eta}/B^{\nu}$ is a non-negative $\mathbb{P}$-local martingale, and hence is a $\mathbb{P}$-supermartingale.

{\color{black}
We consider the unconstrained optimization problem in $\mathcal{M}_{\nu}$ given by
\begin{eqnarray}\label{primal:prob:new}
V_{\nu}(y;A):=\sup_{X^{{\nu}}\in \mathcal{U}^{\nu}_0(1,y)}\mathbb{E}[h_A(X^{{\nu}}_{\tau},Y_{\tau})],
\end{eqnarray}
where 
\begin{eqnarray}
\mathcal{U}^{\nu}_0(x,y)
\hspace{-0.3cm}&:=&\hspace{-0.3cm}
\bigg\{X^{\nu}:\pi
\text{ is $(\mathcal{F}_t)_{t\geq 0}$-predictable and locally square-integrable such that}
\nonumber\\
\hspace{-0.3cm}&&\hspace{-0.3cm}
X^{\nu}_t=x+\sum_{i=1}^n\int_0^t{\pi^i_u}X^{\nu}_u\frac{dS^{i,\nu}_u}{S^{i,\nu}_u}+\int_0^tX^{\nu}_u(1-\mathbf{1}^{\mathsf{T}}\pi_u)\frac{dB^{\nu}_u}{B^{\nu}_u}>0 \text{ for } t\geq 0,\text{ and}
\nonumber\\
\hspace{-0.3cm}&&\hspace{-0.3cm}
\left.\sum_{i=1}^{\infty}e^{-\rho T_i}\mathbb{E}\left[\left(U\left(\frac{X^{\nu}_{T_i}}{\big(e^{\rho\tau}X^{\nu}_{T_{i-1}}\big)^{\gamma}},Y_{T_i}\right)\right)_{-}\right]<\infty \text{ with }Y_0=y
\right\},\, (x,y)\in\mathbb{R}_+\times\mathbb{R}.
\end{eqnarray}

Using similar methods as in the proof of Lemma 2.2 of \cite{CH05}, we can now state the problem \eqref{primal:prob:new} equivalently as
\begin{eqnarray}\label{problem3}
\sup_{X^{\nu}\in\mathcal{U}_0^{\nu}(1,y)}\mathbb{E}[h_A(X^{\nu}_{\tau},Y_{\tau})]
\hspace{-0.3cm}&=&\hspace{-0.3cm}
\sup_{X^{{\nu}}\in\Tilde{\mathcal{U}}^{\nu}_{0,\tau}(1,y)}\mathbb{E}\left[h_A(X^{{\nu}},Y_{\tau})\right],
\end{eqnarray}
where
\begin{eqnarray}\label{tilde.U.power}
\Tilde{\mathcal{U}}^{\nu}_{s,t}(X^{\nu}_s,Y_s):=\left\{X\in\mathcal{F}^+_{t}:\sup_{\eta\in\mathcal{H}}\mathbb{E}\left[\frac{Z^{\nu,\eta}_t/B^{\nu}_{t}}{Z^{\nu,\eta}_s/{B^{\nu}_s}}X\Big|\mathcal{F}_s\right]\leq X^{\nu}_s\right\},\quad 0\leq s\leq t<\infty,
\end{eqnarray}
with $(Z^{\nu,\eta}_t)_{t\geq 0}$ defined by \eqref{z.process}.
}

We next formulate the dual problem associated with the unconstrained optimization problem \eqref{problem3}. For a given function $f$, let us denote its Legendre-Fenchel transform by
\begin{eqnarray}\label{Phi*}
\Phi_f(y):=\sup_{x\geq0}(f(x)-yx),\quad y\in\mathbb{R}_+.\nonumber
\end{eqnarray}
Provided that $f$ is continuous and concave with $f^{\prime}(\infty)=0$, the maximizer attaining the supremum always exists (although not necessarily unique), which is denoted by 
\begin{eqnarray}\label{def.x*}
x_f^*(y):=\arg\max_{x\geq0}(f(x)-yx).
\end{eqnarray} It follows that $\Phi_f(y)=f(x_f^*(y))-yx_f^*(y)$.

By Lemma \ref{lem2.1}, for fixed $y\in\mathbb{R}$, one can define the inverse function of $\frac{\partial}{\partial x}h_A(x,y)$ by $\mathbb{R}_+\ni x\mapsto I(x,y)\in\mathbb{R}_+$, which is a strictly decreasing function.
Let 
$\Phi_{h_A}(u,y)=\sup_{x\geq0}\{h_A(x,y)-xu\}$ be the Legendre-Fenchel transform of the concave function $h_A(x,y)$. It is well known that
\begin{eqnarray}\label{phi(y)=h}
h_A(x^*_{h_A}(u,y),y)-x^*_{h_A}(u,y)u=h_A(I(u,y),y)-I(u,y)u,
\quad u\in\mathbb{R}_+,
\end{eqnarray}
as well as when $\alpha\in(0,1)\, (\alpha\in(-\infty,0),\,\text{\text{resp.}})$
\begin{eqnarray}
\label{phi(0)}
\Phi_{h_A}(0,y)
\hspace{-0.3cm}&:=&\hspace{-0.3cm}
\lim_{u\rightarrow0+}\Phi_{h_A}(u,y)=h_A(\infty,y)=\infty\, (0,\,\text{\text{resp.}}),\\
\label{phi(infty)}
\Phi_{h_A}(\infty,y)
\hspace{-0.3cm}&:=&\hspace{-0.3cm}
\lim_{u\rightarrow\infty}\Phi_{h_A}(u,y)=h_A(0,y)=0\, (-\infty,\,\text{\text{resp.}}),
\end{eqnarray}
and 
\begin{eqnarray}\label{Phi'}
\frac{\partial}{\partial u}\Phi_{h_A}(u,y)=-x^*_{h_A}(u,y),\quad \frac{\partial^2}{\partial u^2}\Phi_{h_A}(u,y)=-\frac{\partial}{\partial u}x_{h_A}^{*}(u,y),\quad (u,y)\in\mathbb{R}_+\times\mathbb{R}.
\end{eqnarray}
In particular, the function $\Phi_{h_A}(\cdot,y)$ is a strictly decreasing, strictly convex and twice differentiable function with respect to the first argument.

{\color{black}For a fixed $y\in\mathbb{R}$, the associated dual problem to the unconstrained problem on the right hand side of \eqref{problem3} is defined as
\begin{align}
\label{dual.pro}
\Tilde{V}(\nu;1,y)
&:=
\inf_{\eta\in\mathcal{H},\lambda>0}\left\{\sup_{X\in\mathcal{F}_{\tau}^+}\mathbb{E}\left[\frac{1}{\alpha}X^{\alpha}{h(Y_{\tau})}+\frac{1}{\alpha}A(Y_{\tau})X^{\alpha(1-\gamma)}-\lambda\frac{XZ_{\tau}^{\nu,\eta}}{B^{\nu}_{\tau}}\right]+\lambda\right\},
\nonumber\\
&=\inf_{\eta\in\mathcal{H},\lambda>0}\mathbb{E}\left[\Phi_{h_A}\left(\lambda\frac{Z^{\nu,\eta}_{\tau}}{B^{\nu}_{\tau}},Y_{\tau}\right)+\lambda\right]
\nonumber\\
&=:\inf_{\eta\in\mathcal{H},\lambda>0}L_{\nu}(\eta,\lambda),
\end{align}
where $\mathcal{F}_{\tau}^{+}$ is the set of non-negative $\mathcal{F}_{\tau}$-measurable random variables.

One can easily verify the weak duality between the unconstrained optimization problem \eqref{problem3} and the dual problem \eqref{dual.pro} that
\begin{eqnarray}\label{relationship}
\inf_{\eta\in\mathcal{H},\lambda>0}L_{\nu}(\eta,\lambda)
\hspace{-0.3cm}&\geq&\hspace{-0.3cm}
\sup_{X^{{\nu}}\in\Tilde{\mathcal{U}}^{\nu}_{0,\tau}(1,y)}\mathbb{E}\left[\frac{1}{\alpha}(X^{{\nu}})^{\alpha}h(Y_{\tau})+\frac{1}{\alpha}A(Y_{\tau})(X^{{\nu}})^{\alpha(1-\gamma)}\right].
\end{eqnarray}
We say that there is no duality gap when the equality holds.

In the next result, provided the existence of an optimal solution to the dual problem \eqref{dual.pro}, we can establish the duality relationship between the optimal solutions to the dual problem \eqref{dual.pro} and the  unconstrained optimization problem \eqref{problem3}. The proof of Proposition \ref{prop2.1} is given in Appendix \ref{app:proof-prop21}.

\begin{prop}\label{prop2.1}
Fix $\nu\in\mathcal{D}$ and $\alpha\in(-\infty,0)\cup(0,1)$.
Let $(\eta^{*},\lambda^{*})\in \mathcal{H}\times\mathbb{R}_+$ and
$X^{\nu,*}:=x_{h_A}^*(\lambda^{*}\frac{Z^{\nu,\eta^*}_{\tau}}{B^{\nu}_{\tau}},Y_{\tau})$.
If
\begin{eqnarray}\label{E=1}
X^{\nu,*}\in\Tilde{\mathcal{U}}^{\nu}_{0,\tau}(1,y)\quad\text{and}\quad\mathbb{E}\left[X^{\nu,*}\frac{Z^{\nu,\eta^*}_{\tau}}{B^{\nu}_{\tau}}\right]=1,
\end{eqnarray}
then, $X^{\nu,*}$ is the optimal solution to the primal problem \eqref{problem3}; and, $(\eta^{*},\lambda^*)$ is the optimal solution to the dual problem \eqref{dual.pro}. In particular, there is no duality gap.
Conversely, if $(\eta^*,\lambda^*)$ is the optimal solution to the dual problem \eqref{dual.pro}, then \eqref{E=1} holds  and $X^{\nu,*}$ is the optimal solution to the unconstrained problem \eqref{problem3}.
\end{prop}}

To prove the existence of the optimal solution to dual problem \eqref{dual.pro}, let us introduce the Dol\'{e}ans-Dade exponentials 
\begin{align}
\label{DD.exp}
\mathcal{E}_t^{W_1}(\theta^{\nu})&:=\exp\left\{\int_0^t\theta^{\nu}(Y_s)dW_{1s}-\frac{1}{2}\int_0^t\|\theta^{\nu}(Y_s)\|^2ds\right\},\quad t\in[0,\tau],\nonumber\\
\mathcal{E}_t^{W_2}(\eta)&:=\exp\left\{\int_0^t\eta_sdW_{2s}-\frac{1}{2}\int_0^t\eta^2_sds\right\},\quad t\in[0,\tau].
\end{align}
The following Propositions \ref{existence.1} and \ref{existence.2} address the existence  of the optimal solution to the dual problem \eqref{dual.pro} for two separate cases $\alpha\in(0,1)$ and $\alpha\in(-\infty,0)$, respectively.
\begin{prop}
\label{existence.1}
For $\alpha\in(0,1)$, there exists an optimal solution to the dual problem \eqref{dual.pro}.
\end{prop}
\begin{proof}
Denote by $\mathcal{F}^{Y}_{\tau}$ the smallest $\sigma$-field generated by $(Y_{t})_{0\leq t\leq \tau}$.
We first verify that, for fixed $\lambda\in\mathbb{R}_+$, $L_{\nu}(\cdot,\lambda)$ defined by \eqref{dual.pro} is a convex functional on $\mathcal{H}$.
Note that the Arrow-Pratt measure of relative risk aversion satisfies
\begin{align}\label{APmrra}
    -\frac{x\frac{\partial^2}{\partial x^2}h_A(x,y)}{\frac{\partial}{\partial x}h_A(x,y)}
    &=-\frac{(\alpha-1)x^{\alpha-1}h(y)+(1-\gamma)(\alpha(1-\gamma)-1)x^{\alpha(1-\gamma)-1}A(y)}{x^{\alpha-1}h(y)+(1-\gamma)x^{\alpha(1-\gamma)-1}A(y)}\nonumber\\
    &\leq1, \quad (x,y)\in\mathbb{R}_+\times\mathbb{R}.
\end{align}
Then, by Lemma 12.6 of \cite{KL91}, one gets that the function $z\mapsto \Phi_{h_A}(e^{z},y)$ is convex on $\mathbb{R}$ for any fixed $y\in\mathbb{R}$. This, together with the convexity of the Euclidean norm in $\mathbb{R}$ and the decreasing property of $\Phi_{h_A}$ as well as Jensen's inequality, implies that
{\color{black}
\begin{eqnarray}
\hspace{-0.3cm}&&\hspace{-0.3cm}
L_{\nu}(\omega_1\eta_1+\omega_2\eta_2,\lambda)
\nonumber\\
\hspace{-0.3cm}&\leq&\hspace{-0.3cm}
\mathbb{E}\left[\Phi_{h_A}\left(\frac{\lambda}{B^{\nu}_{\tau}}\mathcal{E}^{W_{1}}_{\tau}(-\theta^{\nu})(\mathcal{E}^{W_{2}}_{\tau}(\eta_1))^{\omega_1}(\mathcal{E}^{W_{2}}_{\tau}(\eta_2))^{\omega_2},Y_{\tau}\right)\right]+\lambda
\nonumber\\
\hspace{-0.3cm}&\leq&\hspace{-0.3cm}
\mathbb{E}\left[\omega_1\mathbb{E}\left[\left.\Phi_{h_A}\left(\frac{\lambda}{B^{\nu}_{\tau}}\mathcal{E}^{W_1}_{\tau}(-\theta^{\nu})\mathcal{E}^{W_2}_{\tau}(\eta_1),Y_{\tau}\right)\right|\mathcal{F}^{Y}_{\tau}\right]\right.
\nonumber\\
\hspace{-0.3cm}&&\hspace{0.3cm}
\left.+\omega_2\mathbb{E}\left[\left.\Phi_{h_A}\left(\frac{\lambda}{B^{\nu}_{\tau}}\mathcal{E}^{W_1}_{\tau}(-\theta^{\nu})\mathcal{E}^{W_2}_{\tau}(\eta_2),Y_{\tau}\right)\right|\mathcal{F}^{Y}_{\tau}\right]\right]+\lambda
\nonumber\\
\hspace{-0.3cm}&=&\hspace{-0.3cm}
\omega_1 L_{\nu}(\eta_1,\lambda)+\omega_2 L_{\nu}(\eta_2,\lambda),
\end{eqnarray}
}for any $\eta_1,\eta_2\in\mathcal{H}$ and $\omega_1,\omega_2\geq0$ with $\omega_1+\omega_2=1$. Hence, for fixed $\lambda\in\mathbb{R}_+$, $L_{\nu}(\cdot,\lambda)$ is a convex functional on $\mathcal{H}$.
Furthermore, one can employ similar arguments in the proof of Theorem 12.3 in \cite{KL91} to show the existence of optimal solution to dual problem \eqref{dual.pro} when $\alpha\in(0,1)$ because the remaining assumptions of that theorem are fulfilled.
\end{proof}

The next result tackles the more challenging case when $\alpha\in(-\infty,0)$. In this case, the arguments in \cite{KL91} (see, Lemma 12.6 and Theorem 12.3 of \cite{KL91}) are no longer applicable because \eqref{APmrra} does not hold. In fact, $z\mapsto \Phi_{h_A}(e^{z},y)$ is concave on $\mathbb{R}$. We note that \cite{La11} investigated this special case and established the existence of the optimal solution to the dual problem when the utility function is of form $U(x)=\frac{1}{\alpha}x^{\alpha}$ for $\alpha<0$, by connecting the dual problem to a fictitious market through a change-of-measure argument. The use of this change-of-measure argument is possible because their dual functional allows for the separation of the variable that needs to be optimized. In our context, however, due to the form of modified utility function $h_A(x,y)=\frac{1}{\alpha}x^{\alpha}h(y)+\frac{1}{\alpha}x^{\alpha(1-\gamma)}A(y)$, it becomes infeasible for our dual functional to separate the variables. Nevertheless, we show below that it is not necessary to change the measure, and we can still finish the task by establishing the relationship to another artificial optimization problem.

\begin{prop}
\label{existence.2}
For $\alpha\in(-\infty,0)$, there exists an optimal solution to the dual problem \eqref{dual.pro}.
\end{prop}
{\color{black}
\begin{proof}
We consider a fictitious risky asset $\Tilde{S}=(\Tilde{S}_t)_{t\geq0}$ given by
$$d\Tilde{S}_t=\Tilde{S}_tdW_{2t},\quad t\in[0,\tau],$$
with $\Tilde{S}_0=1,$ and the artificial risk-free asset with interest rate 0. A trading strategy $\eta=(\eta_t)_{t\geq0}$ is a predictable process representing the admissible portfolio fraction invested in the risky asset $\Tilde{S}$ at time $t$. The resulting wealth process $\Tilde{X}$ satisfies $$d\Tilde{X}_t=\Tilde{X}_t\eta_tdW_{2t},\quad t\in[0,\tau],$$
with $\Tilde{X}_0=1.$ Recall that the Dol\'{e}ans-Dade exponential $\mathcal{E}_t^W(\cdot)$ is defined by \eqref{DD.exp}. We note that $\mathcal{E}^{W_2}_t(\eta)=\Tilde{X}_t$, $t\in[0,\tau]$. Furthermore, we define the set of wealth processes by
\begin{eqnarray}
\mathcal{V}(1,y)\hspace{-0.3cm}&:=&\hspace{-0.3cm}
\left\{\Tilde{X}:\Tilde{X}_t = 1+\int_0^{t}\Tilde{X}_s\eta_sdW_{2s}>0\text{ for } t\in[0,\tau],\, \eta\text{ is predictable, locally square-integrable}\right.
\nonumber\\
\hspace{-0.3cm}&&\hspace{0cm}
\left.\text{and satisfies }\mathbb{E}\left[\exp\left(\int_0^{\tau}\frac{\|\eta_s\|^2}{2}ds\right)\right]<\infty\text{ and }\mathbb{E}\left[\left(-\Phi_{h_A}\left(\frac{\lambda}{B^{\nu}_{\tau}}\mathcal{E}^{W_1}_{\tau}(-\theta^{\nu})\Tilde{X}_{\tau},Y_{\tau}\right)\right)_-\right]<\infty\right\},\nonumber
\end{eqnarray}
with $Y_0=y\in\mathbb{R}$.
When $\alpha\in(-\infty,0)$, it follows from \eqref{phi(0)} and \eqref{phi(infty)} that $\Phi_{h_A}(x,y)\leq0$ for any $(x,y)\in\mathbb{R}_+\times\mathbb{R}$, and hence
$$\mathbb{E}\left[\left(-\Phi_{h_A}\left(\frac{\lambda}{B^{\nu}_{\tau}}\mathcal{E}^{W_1}_{\tau}(-\theta^{\nu})\Tilde{X}_{\tau},Y_{\tau}\right)\right)_-\right]<\infty$$ is readily satisfied.
Therefore, it can be seen that if $\eta\in\mathcal{H}$  
then the resulting wealth process $\Tilde{X}_t=\mathcal{E}_t^{W_2}(\eta)=1+\int_0^t\Tilde{X}_s\eta_sdW_{2s}\in\mathcal{V}(1,y)$; and vice versa.
It then holds that 
\begin{eqnarray}
\label{connection}
\inf_{\eta\in\mathcal{H}}\mathbb{E}\left[\Phi_{h_A}\left(\lambda\frac{Z^{\nu,\eta}_{\tau}}{B^{\nu}_{\tau}},Y_{\tau}\right)\right]
\hspace{-0.3cm}&=&\hspace{-0.3cm}
-\sup_{\eta\in\mathcal{H}}\mathbb{E}\left[-\Phi_{h_A}\left(\frac{\lambda}{B^{\nu}_{\tau}}\mathcal{E}^{W_1}_{\tau}(-\theta^{\nu})\mathcal{E}^{W_2}_{\tau}(\eta),Y_{\tau}\right)\right]
\nonumber\\
\hspace{-0.3cm}&=&\hspace{-0.3cm}
-\sup_{\Tilde{X}\in\mathcal{V}(1,y)}\mathbb{E}\left[-\Phi_{h_A}\left(\frac{\lambda}{B^{\nu}_{\tau}}\mathcal{E}^{W_1}_{\tau}(-\theta^{\nu})\Tilde{X}_{\tau},Y_{\tau}\right)\right]
\nonumber\\
\hspace{-0.3cm}&=:&\hspace{-0.3cm}
-v(1,y),\quad y\in\mathbb{R}.
\end{eqnarray}
From \eqref{relationship} and the fact that $\Phi_{h_A}(\cdot,\cdot)\leq0$ when $\alpha\in(-\infty,0)$, it follows that
\begin{eqnarray}
\lambda\geq \inf_{\eta\in\mathcal{H}}\mathbb{E}\left[\Phi_{h_A}\left(\lambda\frac{Z^{\nu,\eta}_{\tau}}{B^{\nu}_{\tau}},Y_{\tau}\right)\right]+\lambda \geq \sup_{X^{{\nu}}\in\Tilde{\mathcal{U}}^{\nu}_{0,\tau}(1,y)}\mathbb{E}\left[\frac{1}{\alpha}(X^{{\nu}})^{\alpha}h(Y_{\tau})+\frac{1}{\alpha}A(Y_{\tau})(X^{{\nu}})^{\alpha(1-\gamma)}\right]>-\infty.\nonumber
\end{eqnarray}
Consequently, we arrive at an artificial optimization problem
\begin{eqnarray}\label{ar.primal}
v(1,y)=\sup_{\Tilde{X}\in\mathcal{V}(1,y)}\mathbb{E}\left[-\Phi_{h_A}\left(\frac{\lambda}{B^{\nu}_{\tau}}\mathcal{E}^{W_1}_{\tau}(-\theta^{\nu})\Tilde{X}_{\tau},Y_{\tau}\right)\right]<\infty,\quad y\in\mathbb{R}.
\end{eqnarray}
The dual problem associated to the artificial primal problem \eqref{ar.primal} is defined by
\begin{eqnarray}\label{ar.dual}
\Tilde{v}_{\nu}(z,y)
\hspace{-0.3cm}&:=&\hspace{-0.3cm}
\inf_{\varrho\in\mathcal{H}}\left\{\sup_{\Tilde{X}\in\mathcal{F}^{+}_{\tau}}\mathbb{E}\left[-\Phi_{h_A}\left(\frac{\lambda}{B^{\nu}_{\tau}}\mathcal{E}_{\tau}^{W_1}(-\theta^{\nu})\Tilde{X},Y_{\tau}\right)-z\mathcal{E}_{\tau}^{W_2}(\varrho)\Tilde{X}\right]\right\}
\nonumber\\
\hspace{-0.3cm}&=&\hspace{-0.3cm}
\inf_{\varrho\in\mathcal{H}}\mathbb{E}\left[-\Phi_{h_A}\left(\frac{\partial}{\partial x}h_A\left(z\frac{B^{\nu}_{\tau}\mathcal{E}_{\tau}^{W_2}(\varrho)}{\lambda\mathcal{E}_{\tau}^{W_1}(-\theta^{\nu})},Y_{\tau}\right),Y_{\tau}\right)-\frac{\partial}{\partial x}h_A\left(z\frac{B^{\nu}_{\tau}\mathcal{E}_{\tau}^{W_2}(\varrho)}{\lambda\mathcal{E}_{\tau}^{W_1}(-\theta^{\nu})},Y_{\tau}\right)z\frac{B^{\nu}_{\tau}\mathcal{E}_{\tau}^{W_2}(\varrho)}{\lambda\mathcal{E}_{\tau}^{W_1}(-\theta^{\nu})}\right]
\nonumber\\
\hspace{-0.3cm}&=&\hspace{-0.3cm}
\inf_{\varrho\in\mathcal{H}}\mathbb{E}\left[-h_A\left(z\frac{B^{\nu}_{\tau}}{\lambda\mathcal{E}_{\tau}^{W_1}(-\theta^{\nu})}\mathcal{E}_{\tau}^{W_2}(\varrho),Y_{\tau}\right)\right],\quad (z,y)\in\mathbb{R}_+\times\mathbb{R},
\end{eqnarray}
where, in the last equality, we have used the property (4.5) in \cite{KL91}. It is easy to check that $0<\Tilde{v}_{\nu}(z,y)<\infty,$ for $(z,y)\in\mathbb{R}_+\times\mathbb{R}$.

Thanks to the strict decreasing property and convexity of $\mathbb{R}_{+}\ni x\mapsto-h_{A}(e^{x},y)\in\mathbb{R}_{+}$, it is easy to see that the functional $\mathbb{E}\big[-h_A\big(zB^{\nu}_{\tau}\mathcal{E}_{\tau}^{W_2}(\varrho)/\lambda\mathcal{E}_{\tau}^{W_1}(-\theta^{\nu}),Y_{\tau}\big)\big]$ is convex with respect to $\varrho$ on $\mathcal{H}$. We can then obtain the existence of optimal solution to the dual problem \eqref{ar.dual} by employing a similar argument as that of Proposition \ref{existence.1}. In addition, following the proof of Proposition \ref{prop2.1},  the existence of optimizer to problem \eqref{ar.dual} implies the existence of optimal solution to the artificial primal problem \eqref{ar.primal}. This, together with \eqref{connection}, implies the existence of optimal solution to the dual problem \eqref{dual.pro}.
\end{proof}

The next lemma suggests the relationship between the portfolio strategy $\pi$ and the final wealth $X^*_{\tau}$. Its proof is based on Lemma 2.2 and Theorem 2.3 in \cite{CH05}, and it is hence omitted.

\begin{lem}\label{lemma4.1}
For some $\nu\in\mathcal{D}$, let $(\eta^*,\lambda^*)\in\mathcal{H}\times\mathbb{R}_+$ be the optimal solution to the problem \eqref{dual.pro} and recall the optimal solution to the unconstrained problem \eqref{problem3} is given by  $X^{\nu,*}:=x_{h_A}^*(\lambda^{*}\frac{Z^{\nu,\eta^*}_{\tau}}{B^{\nu}_{\tau}},Y_{\tau})$.
Then, there exists a progressively measurable process $\pi^{\nu}$ with $\mathbb{E}[\int_0^{\tau}\|\pi^{\nu}_t\|^2dt]<\infty$ such that the resulting wealth $X^{1,y,\pi^{\nu},\nu}$ satisfies
\begin{eqnarray}\label{eq:X:process}
dX^{1,y,\pi^{\nu},\nu}_t 
\hspace{-0.3cm}&=&\hspace{-0.3cm} 
\left[r(Y_t)+\delta(\nu_t)+(\pi^{\nu}_t)^{\mathsf{T}}(\mu(Y_t)+\nu_t-r(Y_t)\mathbf{1})\right]X^{1,y,\pi^{\nu},\nu}_tdt\nonumber\\
\hspace{-0.3cm}&&\hspace{-0.3cm} 
+X^{1,y,\pi^{\nu},\nu}_t(\pi_t^{\nu})^{\mathsf{T}}\sigma(Y_t)dW_{1t},
\end{eqnarray}
with
\begin{align}
    \label{eq.4.33.w}
    X^{1,y,\pi^{\nu},\nu}_0=1,\quad \quad X^{1,y,\pi^{\nu},\nu}_\tau=X^{\nu,*}.
\end{align}
In particular, $X^{1,y,\pi^{\nu},\nu}\in \mathcal{U}^{\nu}_0(1,y)$ is an optimal solution to \eqref{primal:prob:new}.
\end{lem}

\section{Duality for Auxiliary Constrained Terminal Wealth Optimization Problem}\label{sec:duality}

Building upon the previous duality results for the unconstrained optimization problem \eqref{primal:prob:new}, we now turn to examine the duality result for the auxiliary constrained terminal wealth optimization problem \eqref{A*.def} in this section. To show how the fixed-\(\nu\) solution from Section \ref{sec:parameter} can be converted into the solution of the auxiliary problem \eqref{A*.def}, we carry out some technical results as follows: Proposition \ref{prop4.4} gives a sufficient complementarity condition \eqref{eq:pro3.5:a}--\eqref{eq:pro3.5:b} under which the optimizer in a fictitious market is also optimal for the original constrained problem \eqref{A*.def}. Proposition \ref{prop4.5} proves the converse feasibility mechanism: the dual budget inequality \eqref{eq:pro4.5:e} forces the financed portfolio to satisfy the convex constraint and the complementary slackness condition. Proposition \ref{primaldualcons} combines these two directions to establish the primal-dual equivalence and absence of duality gap for \eqref{A*.def}, based on the full dual problem \eqref{dual.pro.nu} and the budget condition \eqref{E=1.nu}. Finally, Proposition \ref{prop5.4} gives the existence of the minimizer for \eqref{dual.pro.nu} over both \(\nu\) and \(\eta\), which is the dual input needed in Section \ref{sec:fix} to define the fixed-point operator and construct the periodic optimal portfolio.

We first provide a sufficient condition for the optimal solution of the unconstrained problem \eqref{primal:prob:new} such that it coincides with the optimal solution to the constrained problem \eqref{A*.def}. 
\begin{prop}\label{prop4.4}
    Assume that for some $\upsilon\in\mathcal{D}$, 
$(\eta^*,\lambda^*)\in\mathcal{H}\times\mathbb{R}_+$ is the optimal solution to the dual problem \eqref{dual.pro}, and, the corresponding portfolio process $\pi^{\upsilon}$ of the wealth process $X^{1,y,\pi^{\upsilon},\upsilon}$ given by \eqref{eq:X:process} and \eqref{eq.4.33.w}
satisfies
    \begin{align}
        &\pi^{\upsilon}_{t}\in K, \quad a.s.,\label{eq:pro3.5:a}\\
&\delta(\upsilon_t)+(\pi^{\upsilon}_t)^{\mathsf{T}}\upsilon_t=0,\quad a.s..\label{eq:pro3.5:b}
    \end{align}
    Then, the wealth process $X^{1,y,\pi^{\upsilon},\upsilon}$ given by \eqref{eq:X:process} and \eqref{eq.4.33.w} 
    is optimal for the auxiliary constrained optimization problem \eqref{A*.def} in the original market $\mathcal{M}$. In addition, it holds that
    \begin{eqnarray}\label{eq:pro3.5:conc}
    \mathbb{E}\left[h_{A}(X^{1,y,\pi^{\upsilon},\upsilon}_{\tau},Y_{\tau})\right]\leq \sup_{X^{{\nu}}\in\mathcal{U}_{0}^{\nu}(1,y)}\mathbb{E}[h_A(X^{{\nu}}_{\tau},Y_{\tau})]=\sup_{X^{{\nu}}\in\Tilde{\mathcal{U}}_{0,\tau}^{\nu}(1,y)}\mathbb{E}[h_A(X^{{\nu}},Y_{\tau})],\quad \forall \nu\in\mathcal{D}.
    \end{eqnarray}
\end{prop}
\begin{proof}
In view of \eqref{eq:X:process} with $\nu\equiv\upsilon$, \eqref{eq:pro3.5:a} and \eqref{eq:pro3.5:b}, we have
\begin{eqnarray}\label{eq:pro3.5:c}
dX^{1,y,\pi^{\upsilon},\upsilon}_t=\left[r(Y_t)+(\pi_t^{\upsilon})^{\mathsf{T}}(\mu(Y_t)-r(Y_t)\mathbf{1})\right]X^{1,y,\pi^{\upsilon},\upsilon}_tdt+X^{1,y,\pi^{\upsilon},\upsilon}_t(\pi^{\upsilon}_t)^{\mathsf{T}}\sigma(Y_t)dW_{1t},
\end{eqnarray}
with $X^{1,y,\pi^{\upsilon},\upsilon}_0=1$ and $X^{1,y,\pi^{\upsilon},\upsilon}_{\tau}=x^*_{h_A}(\lambda^*\frac{Z_{\tau}^{\upsilon,\eta^*}}{B^{\upsilon}_{\tau}},Y_{\tau})$.
This, together with \eqref{dX}, yields that $X^{1,y,\pi^{\upsilon},\upsilon}$ is a wealth process corresponding to $\pi^{\upsilon}$ in the original market $\mathcal{M}$. Thus we have $X^{1,y,\pi^{\upsilon},\upsilon}\in\mathcal{U}_0(1,y)$ and $X^{1,y,\pi^{\upsilon},\upsilon}_{\tau}=x^*_{h_A}(\lambda^*\frac{Z_{\tau}^{\upsilon,\eta^*}}{B^{\upsilon}_{\tau}},Y_{\tau})$, which along with Lemma \ref{lemma4.1} yields
\begin{eqnarray}\label{eq:pro3.5:d}
    \sup_{X^{{\upsilon}}\in\mathcal{U}^{\upsilon}_0(1,y)}\mathbb{E}\left[h_A(X^{{\upsilon}}_{\tau},Y_{\tau})\right]=\mathbb{E}\left[h_A(X^{1,y,\pi^{\upsilon},\upsilon}_{\tau},Y_{\tau})\right]\leq \sup_{X\in\mathcal{U}_0(1,y)}\mathbb{E}\left[h_A(X_{\tau},Y_{\tau})\right].
\end{eqnarray}
Note that $\delta(\nu_t)+\pi^{\mathsf{T}}_t\nu_t\geq0$ for all $\nu\in\mathcal{D}$ and admissible $\pi\in K$ due to the definition of the support function $\delta(\cdot)$. Denote $X^{1,y,\pi}$ as the resulting wealth process under $\pi$ in market $\mathcal{M}$. A direct comparison between \eqref{dX} and \eqref{eq:Xnu} yields 
$$X^{1,y,\pi,\nu}_{t}\geq X^{1,y,\pi}_t>0,\quad\forall t\in[0,\tau].$$
Therefore, we have 
$$\sup_{X^{{\nu}}\in\mathcal{U}^{\nu}_0(1,y)}\mathbb{E}[h_A(X^{{\nu}}_{\tau},Y_{\tau})]\geq\sup_{X\in\mathcal{U}_0(1,y)}[h_A(X_{\tau},Y_{\tau})].$$
This implies the opposite inequality of \eqref{eq:pro3.5:d}, verifying the optimality of $X^{1,y,\pi^{\upsilon},\upsilon}$ for the problem \eqref{A*.def}.

On the other hand, fix an arbitrary $\nu\in\mathcal{D}$ and denote $X^{1,y,\pi^{\upsilon},\nu}$ as the resulting wealth process under $\pi^{\upsilon}$ in market $\mathcal{M}_{\nu}$. The equation \eqref{eq:Xnu} becomes
\begin{eqnarray}
dX^{1,y,\pi^{\upsilon},\nu}_{t}
\hspace{-0.3cm}&=&\hspace{-0.3cm} 
\left[r(Y_t)+\delta(\nu_t)+(\pi^{\upsilon}_t)^{\mathsf{T}}(\mu(Y_t)+\nu_t-r(Y_t)\mathbf{1})\right]X^{1,y,\pi^{\upsilon},\nu}_{t}dt+X^{1,y,\pi^{\upsilon},\nu}_{t}(\pi^{\upsilon}_t)^{\mathsf{T}}\sigma(Y_t)dW_{1t}.\nonumber
\end{eqnarray}
The comparison between the above result and \eqref{eq:pro3.5:c} leads to $$X^{1,y,\pi^{\upsilon},\nu}_{t}\geq X^{1,y,\pi^{\upsilon},\upsilon}_t>0,\quad \forall t\in [0,\tau],$$
where we have used the fact that $\delta(\nu_t)+(\pi^{\upsilon}_t)^{\mathsf{T}}\nu_t\geq0$. Hence, we obtain the inequality of \eqref{eq:pro3.5:conc}. This, together with \eqref{problem3}, implies that \eqref{eq:pro3.5:conc} holds. 
\end{proof}

\begin{prop}\label{prop4.5}
Let $X$ be a positive $\mathcal{F}^{+}_{\tau}$-measurable random variable, and suppose that there exists a process $\nu^*\in\mathcal{D}$ and a process $\eta^*\in\mathcal{H}$ such that 
   \begin{eqnarray}\label{eq:pro4.5:e}
    \mathbb{E}\left[X\frac{Z^{\nu,\eta^*}_{\tau}}{B^{\nu}_{\tau}}\right]\leq\mathbb{E}\left[X\frac{Z^{\nu^*,\eta^*}_{\tau}}{B^{\nu^*}_{\tau}}\right]=1,\quad\forall \nu\in\mathcal{D}.
   \end{eqnarray}
   Then, there exists a portfolio process $\pi$ in market $\mathcal{M}_{\nu^*}$ such that the constraints $\delta(\nu^*_t)+\pi^{\mathsf{T}}_t\nu^*_t=0$ and $\pi\in K$ are satisfied a.e., and the wealth process $X^{1,y,\pi}\in\mathcal{U}_0(1,y)$  with $X^{1,y,\pi}_{\tau}=X$.
\end{prop}
The proof of Proposition \ref{prop4.5} is given in Appendix \ref{app:proof-prop45}.

We next consider the dual problem of the primal constrained optimization problem \eqref{A*.def}, which is given by
\begin{eqnarray}\label{dual.pro.nu}
\Tilde{V}(\lambda,y)=\inf_{\nu\in\mathcal{D},\eta\in\mathcal{H}}\mathbb{E}\left[\Phi_{h_A}(\lambda\frac{Z^{\nu,\eta}_{\tau}}{B^{\nu}_{\tau}},Y_{\tau})\right],\quad (\lambda,y)\in\mathbb{R}^+\times\mathbb{R}.
\end{eqnarray}

\begin{prop}\label{primaldualcons}
Let $(\nu^*,\eta^{*},\lambda^{*})\in \mathcal{D}\times\mathcal{H}\times\mathbb{R}_+$ and define
$X^{\nu^*,*}:=x_{h_A}^*(\lambda^{*}\frac{Z^{\nu^*,\eta^*}_{\tau}}{B^{\nu^*}_{\tau}},Y_{\tau})$.
If the following inequality holds
\begin{eqnarray}\label{E=1.nu}
\mathbb{E}\left[X^{\nu^*,*}\frac{Z^{\nu,\eta}_{\tau}}{B^{\nu}_{\tau}}\right]\leq
\mathbb{E}\left[X^{\nu^*,*}\frac{Z^{\nu^*,\eta^*}_{\tau}}{B^{\nu^*}_{\tau}}\right]=1,\quad \text{for any }(\nu,\eta)\in\mathcal{D}\times\mathcal{H},
\end{eqnarray}
then there exists a portfolio process $\pi$ such that the resulting wealth process $X^{1,y,\pi}$ given by \eqref{eq:X:process} and \eqref{eq.4.33.w} with $\nu=\nu^{*}$ 
is the optimal solution to the primal problem \eqref{A*.def}. Moreover, $\lambda^{*}=\arg\min_{\lambda>0}(\Tilde{V}(\lambda,y)+\lambda)$ and $(\nu^*,\eta^{*})$ solves the dual problem \eqref{dual.pro.nu} with $\lambda=\lambda^{*}$. In particular, there is no duality gap.
Conversely, if
$\lambda^{*}=\arg\min_{\lambda>0}(\Tilde{V}(\lambda,y)+\lambda)$ and $(\nu^*,\eta^{*})$ solves the dual problem \eqref{dual.pro.nu} with $\lambda=\lambda^{*}$, then 
\eqref{E=1.nu} holds.
\end{prop}
\begin{proof}
Assume that \eqref{E=1.nu} holds. By Proposition \ref{prop2.1} with $\nu=\nu^*$, it follows that $(\eta^*,\lambda^*)$ solves the dual problem \eqref{dual.pro} with $\nu=\nu^*$. Furthermore, Proposition \ref{prop4.5} guarantees the existence of a portfolio $\pi\in K$ satisfying  $\delta(\nu^*_t)+\pi^{\mathsf{T}}_t\nu^*_{t}=0,\,a.e.$. The corresponding wealth  $X^{1,y,\pi}\in\mathcal{U}_0(1,y)$ satisfies \eqref{eq:X:process} and \eqref{eq.4.33.w} with $\nu=\nu^*$.
Consequently, invoking the first assertion of Proposition \ref{prop4.4}, we conclude that $X^{1,y,\pi}\in\mathcal{U}_0(1,y)$
is the optimal solution to the primal problem \eqref{A*.def}.

Recall that $X^{\nu^*,*}$ is the optimal solution to the primal problem \eqref{problem3} with $\nu=\nu^*$ (see Proposition \ref{prop2.1}).
Therefore, combining \eqref{E<h}, \eqref{E=1.nu}, and the established result that $X^{1,y,\pi}$ solves the primal problem \eqref{A*.def}, we obtain
\begin{eqnarray}\label{dual.gap.nu}
\hspace{-0.3cm}&&\hspace{-0.3cm}
\inf_{\nu\in\mathcal{D},\eta\in\mathcal{H},\lambda>0}L_{\nu}(\eta,\lambda)
\nonumber\\
\hspace{-0.3cm}&=&\hspace{-0.3cm}
\inf_{\nu\in\mathcal{D},\eta\in\mathcal{H},\lambda>0}\left\{\mathbb{E}\left[h_A(x^*_{h_A}(\lambda\frac{Z^{\nu,\eta}_{\tau}}{B^{\nu}_{\tau}},Y_{\tau}),Y_{\tau})
-\lambda x^*_{h_A}(\lambda\frac{Z^{\nu,\eta}_{\tau}}{B^{\nu}_{\tau}},Y_{\tau})\frac{Z^{\nu,\eta}_{\tau}}{B^{\nu}_{\tau}}\right]+\lambda\right\}
\nonumber\\
\hspace{-0.3cm}&\leq&\hspace{-0.3cm}
\mathbb{E}\left[h_A(x^*_{h_A}(\lambda^*\frac{Z^{\nu^*,\eta^*}_{\tau}}{B^{\nu^*}_{\tau}},Y_{\tau}),Y_{\tau})
-\lambda^* x^*_{h_A}(\lambda^*\frac{Z^{\nu^*,\eta^*}_{\tau}}{B^{\nu^*}_{\tau}},Y_{\tau})\frac{Z^{\nu^*,\eta^*}_{\tau}}{B^{\nu^*}_{\tau}}\right]+\lambda^*
\nonumber\\
\hspace{-0.3cm}&=&\hspace{-0.3cm}
\mathbb{E}\left[\frac{1}{\alpha}\left(x^*_{h_A}(\lambda^*\frac{Z^{\nu^*,\eta^*}_{\tau}}{B^{\nu^*}_{\tau}},Y_{\tau})\right)^{\alpha}{h(Y_{\tau})}+{\frac{1}{\alpha}}A(Y_{\tau})\left(x^*_{h_A}(\lambda^*\frac{Z^{\nu^*,\eta^*}_{\tau}}{B^{\nu^*}_{\tau}},Y_{\tau})\right)^{\alpha(1-\gamma)}\right]
\nonumber\\
\hspace{-0.3cm}&=&\hspace{-0.3cm}
\sup_{X^{{\nu^*}}\in\Tilde{\mathcal{U}}^{\nu^*}_{0,\tau}(1,y)}\mathbb{E}\left[\frac{1}{\alpha}(X^{{\nu^*}})^{\alpha}{h(Y_{\tau})}+{\frac{1}{\alpha}}A(Y_{\tau})(X^{{\nu^*}})^{\alpha(1-\gamma)}\right]
\nonumber\\
\hspace{-0.3cm}&= &\hspace{-0.3cm}
\sup_{X\in{\mathcal{U}}_{0}(1,y)}\mathbb{E}\left[\frac{1}{\alpha}X_{\tau}^{\alpha}{h(Y_{\tau})}+{\frac{1}{\alpha}}A(Y_{\tau})X_{\tau}^{\alpha(1-\gamma)}\right].
\end{eqnarray}
Invoking again the optimality of $X^{1,y,\pi}$ for the primal problem \eqref{A*.def}, we have
\begin{eqnarray}
\label{eq.5.25.w}
\inf_{\nu\in\mathcal{D},\eta\in\mathcal{H},\lambda>0}L_{\nu}(\eta,\lambda)
\hspace{-0.3cm}&=&\hspace{-0.3cm}
\inf_{\nu\in\mathcal{D},\eta\in\mathcal{H},\lambda>0}\left\{\sup_{X\in\mathcal{F}_{\tau}^+}\left\{\mathbb{E}\left[h_A(X,Y_{\tau})-\lambda\frac{Z_{\tau}^{\nu,\eta}}{B^{\nu}_{\tau}}X\right]\right\}+\lambda\right\}
\nonumber\\
\hspace{-0.3cm}&\geq&\hspace{-0.3cm}
\inf_{\nu\in\mathcal{D}}\sup_{X^{{\nu}}\in\Tilde{\mathcal{U}}^{\nu}_{0,\tau}(1,y)}\mathbb{E}\left[h_A(X^{{\nu}},Y_{\tau})\right]
\nonumber\\
\hspace{-0.3cm}&\geq&\hspace{-0.3cm}
\inf_{\nu\in\mathcal{D}}\mathbb{E}\left[h_A(X^{\nu,*},Y_{\tau})\right]\quad (\text{hint: \eqref{E=1.nu}\,$\Rightarrow\, X^{\nu,*}\in \Tilde{\mathcal{U}}^{\nu}_{0,\tau}(1,y)$ for any $\nu\in\mathcal{D}$})
\nonumber\\
\hspace{-0.3cm}&=&\hspace{-0.3cm} \inf_{\nu\in\mathcal{D}}\mathbb{E}[h_A(X^{1,y,\pi}_{\tau},Y_{\tau})]=\mathbb{E}[h_A(X^{1,y,\pi}_{\tau},Y_{\tau})]
\nonumber\\
\hspace{-0.3cm}&=&\hspace{-0.3cm}
\sup_{X\in{\mathcal{U}}_{0}(1,y)}\mathbb{E}\left[h_A(X_{\tau},Y_{\tau})\right].
\end{eqnarray}
Combining \eqref{dual.gap.nu} and \eqref{eq.5.25.w} yields
\begin{align}
\left\{
\begin{array}{ll}
\inf\limits_{\nu\in\mathcal{D},\eta\in\mathcal{H},\lambda>0}L_{\nu}(\eta,\lambda)=\sup\limits_{X\in{\mathcal{U}}_{0}(1,y)}\mathbb{E}\left[h_A(X_{\tau},Y_{\tau})\right],
\\
L_{\nu^*}(\eta^*,\lambda^*)\leq L_{\nu^*}(\eta^*,\lambda),\quad \lambda>0,
\\
L_{\nu^*}(\eta^*,\lambda^*)\leq L_{\nu}(\eta,\lambda^*),\quad \nu\in\mathcal{D},\,\eta\in\mathcal{H},
\\
L_{\nu^*}(\eta^*,\lambda^*)=\inf\limits_{\nu\in\mathcal{D},\eta\in\mathcal{H}}L_{\nu}(\eta,\lambda^*)=\Tilde{V}(\lambda^*,y)+\lambda^*,
\\
L_{\nu^*}(\eta^*,\lambda^*)=
\inf\limits_{\nu\in\mathcal{D},\eta\in\mathcal{H},\lambda>0}L_{\nu}(\eta,\lambda)=\inf\limits_{\lambda>0}\inf\limits_{\nu\in\mathcal{D},\eta\in\mathcal{H}}L_{\nu}(\eta,\lambda)=\inf\limits_{\lambda>0}(\Tilde{V}(\lambda,y)+\lambda). &
\end{array}
\right.
\nonumber
\end{align}
This implies that there is no duality gap, $\lambda^*=\arg\min_{\lambda>0}L_{\nu^*}(\eta^*,\lambda)=\arg\min_{\lambda>0}(\Tilde{V}(\lambda,y)+\lambda)$, and $(\nu^*,\eta^*)$ solves the dual problem \eqref{dual.pro.nu} with $\lambda=\lambda^*$. Thus, the first claim is proved.

Assume that $\lambda^{*}=\arg\min_{\lambda>0}(\Tilde{V}(\lambda,y)+\lambda)$ and $(\nu^*,\eta^{*})$ solves the dual problem \eqref{dual.pro.nu} with $\lambda=\lambda^{*}$. Then we have
\begin{align}
L_{\nu^{*}}(\eta^{*},\lambda^{*})\geq& \inf\limits_{\eta\in\mathcal{H},\lambda>0}L_{\nu^{*}}(\eta,\lambda)
\nonumber\\
\geq & \inf\limits_{\nu\in\mathcal{D},\eta\in\mathcal{H},\lambda>0}L_{\nu}(\eta,\lambda)
\nonumber\\
=&
\inf\limits_{\lambda>0}\inf\limits_{\nu\in\mathcal{D},\eta\in\mathcal{H}}L_{\nu}(\eta,\lambda)
\nonumber\\
=&\inf\limits_{\lambda>0}(\Tilde{V}(\lambda,y)+\lambda)
\nonumber\\
=&\Tilde{V}(\lambda^{*},y)+\lambda^{*}
\nonumber\\
=&\inf_{\nu\in\mathcal{D},\eta\in\mathcal{H}}\mathbb{E}\left[\Phi_{h_A}(\lambda^{*}\frac{Z^{\nu,\eta}_{\tau}}{B^{\nu}_{\tau}},Y_{\tau})\right]+\lambda^{*}
\nonumber\\
=&\mathbb{E}\left[\Phi_{h_A}(\lambda^{*}\frac{Z^{\nu^{*},\eta^{*}}_{\tau}}{B^{\nu^{*}}_{\tau}},Y_{\tau})\right]+\lambda^{*}
\nonumber\\
=&L_{\nu^{*}}(\eta^{*},\lambda^{*}).\nonumber
\end{align}
This implies that $(\eta^{*},\lambda^{*})$ is the optimal solution of the problem \eqref{dual.pro} with $\nu=\nu^*$. By Lemma \ref{lemma4.1} with $\nu=\nu^*$, there exists a progressively measurable process $\pi^{\nu^*}$ with $\mathbb{E}[\int_0^{\tau}\|\pi^{\nu^*}_t\|^2dt]<\infty$ such that the resulting wealth $X^{1,y,\pi^{\nu^*},\nu^*}$ satisfies \eqref{eq:X:process} and \eqref{eq.4.33.w} with $\nu=\nu^*$.
We claim that $\pi^{\nu^*}$ satisfies 
\begin{eqnarray}\label{eq:pro4.6:f}
\pi^{\nu^*}_t\in K,\quad\text{and}\quad\delta(\nu_t^*)+(\pi^{\nu^*}_t)^{\mathsf{T}}\nu^*_t=0,\,a.e..
\end{eqnarray}
Recall the notations $(\nu^*)_t^{(\nu,\epsilon,n)}$, $\hat{\delta}^{\nu}(\nu^*_t)$, $L^{\nu}_t$, $N^{\nu}_t$ and $\tau_n$ defined in the proof of Proposition \ref{prop4.5}. 
Using \eqref{eq:pro4.5:h}, the decreasing property of $x^*_{h_A}$ in its first argument, and the fact that $(\nu^*,\eta^{*})$ solves the dual problem \eqref{dual.pro.nu} with $\lambda=\lambda^{*}$, we obtain
\begin{eqnarray}\label{eq:pro4.6:b}
0\leq \Lambda^{(n)}_{\epsilon}
\hspace{-0.3cm}&:=&\hspace{-0.3cm}
\frac{1}{\epsilon}\left[\Phi_{h_A}(\lambda^*\frac{Z^{(\nu^*)^{(\nu,\epsilon,n)},\eta^*}_{\tau}}{B^{(\nu^*)^{(\nu,\epsilon,n)}}_{\tau}},Y_{\tau})-\Phi_{h_A}(\lambda^*\frac{Z^{\nu^*,\eta^*}_{\tau}}{B^{\nu^*}_{\tau}},Y_{\tau})\right]
\nonumber\\
\hspace{-0.3cm}&=&\hspace{-0.3cm}
\lambda^*\frac{\partial}{\partial x}\Phi_{h_A}(\lambda^*\left[w\frac{Z^{\nu^*,\eta^*}_{\tau}}{B^{\nu^*}_{\tau}}+(1-w)\frac{Z^{(\nu^*)^{(\nu,\epsilon,n)},\eta^*}_{\tau}}{B^{(\nu^*)^{(\nu,\epsilon,n)}}_{\tau}}\right],Y_{\tau})\frac{Z^{\nu^*,\eta^*}_{\tau}}{\epsilon B^{\nu^*}_{\tau}}
\quad \quad  (\text{for some $w\in(0,1)$})
\nonumber\\
\hspace{-0.3cm}&&\hspace{-0.3cm}
\times
\left(\frac{{Z^{(\nu^*)^{(\nu,\epsilon,n)},\eta^*}_{\tau}}/{B^{(\nu^*)^{(\nu,\epsilon,n)}}_{\tau}}}{{Z^{\nu^*,\eta^*}_{\tau}}/{B^{\nu^*}_{\tau}}}-1\right)
\nonumber\\
\hspace{-0.3cm}&=&\hspace{-0.3cm}
\lambda^* x^*_{h_A}(\lambda^*\frac{Z^{\nu^*,\eta^*}_{\tau}}{B^{\nu^*}_{\tau}}\left[w+(1-w)\frac{{Z^{(\nu^*)^{(\nu,\epsilon,n)},\eta^*}_{\tau}}/{B^{(\nu^*)^{(\nu,\epsilon,n)}}_{\tau}}}{{Z^{\nu^*,\eta^*}_{\tau}}/{B^{\nu^*}_{\tau}}}\right],Y_{\tau})\frac{Z^{\nu^*,\eta^*}_{\tau}}{\epsilon B^{\nu^*}_{\tau}}
\nonumber\\
\hspace{-0.3cm}&&\hspace{-0.3cm}
\times\left(1-\frac{{Z^{(\nu^*)^{(\nu,\epsilon,n)},\eta^*}_{\tau}}/{B^{(\nu^*)^{(\nu,\epsilon,n)}}_{\tau}}}{{Z^{\nu^*,\eta^*}_{\tau}}/{B^{\nu^*}_{\tau}}}\right)
\nonumber\\
\hspace{-0.3cm}&\leq&\hspace{-0.3cm}
\lambda^* \frac{Z^{\nu^*,\eta^*}_{\tau}}{ B^{\nu^*}_{\tau}}\sup_{0<\epsilon<1}\frac{1-e^{-3\epsilon n}}{\epsilon}x^*_{h_A}(\lambda^*\frac{Z^{\nu^*,\eta^*}_{\tau}}{B^{\nu^*}_{\tau}}\left[w+(1-w)\frac{{Z^{(\nu^*)^{(\nu,\epsilon,n)},\eta^*}_{\tau}}/{B^{(\nu^*)^{(\nu,\epsilon,n)}}_{\tau}}}{{Z^{\nu^*,\eta^*}_{\tau}}/{B^{\nu^*}_{\tau}}}\right],Y_{\tau})
\nonumber\\
\hspace{-0.3cm}&\leq &\hspace{-0.3cm}
\lambda^* \frac{Z^{\nu^*,\eta^*}_{\tau}}{ B^{\nu^*}_{\tau}}\sup_{0<\epsilon<1}\frac{1-e^{-3\epsilon n}}{\epsilon}x^*_{h_A}(\lambda^*\frac{Z^{\nu^*,\eta^*}_{\tau}}{B^{\nu^*}_{\tau}}\left[w+(1-w)e^{-3\epsilon n}\right],Y_{\tau})
\nonumber\\
\hspace{-0.3cm}&\leq &\hspace{-0.3cm}
\lambda^* \frac{Z^{\nu^*,\eta^*}_{\tau}}{ B^{\nu^*}_{\tau}}\sup_{0<\epsilon<1}\frac{1-e^{-3\epsilon n}}{\epsilon}x^*_{h_A}(\lambda^*\frac{Z^{\nu^*,\eta^*}_{\tau}}{B^{\nu^*}_{\tau}}e^{-3\epsilon n},Y_{\tau}),\quad 0<\epsilon<1.
\end{eqnarray}
Therefore, Fatou's lemma implies
\begin{eqnarray}\label{eq:pro4.6:c}
0\leq \limsup_{\epsilon\downarrow0}\mathbb{E}\left[\Lambda^{(n)}_{\epsilon}\right]\leq \mathbb{E}\left[\limsup_{\epsilon\downarrow0}\Lambda^{(n)}_{\epsilon}\right].
\end{eqnarray}
On the other hand, by \eqref{eq:pro4.5:h} and \eqref{eq:pro4.6:b}, we also have 
\begin{eqnarray}\label{eq:pro4.6:d}
\Lambda^{(n)}_{\epsilon}
\hspace{-0.3cm}&\leq&\hspace{-0.3cm}
\lambda^* x^*_{h_A}(\lambda^*\frac{Z^{\nu^*,\eta^*}_{\tau}}{B^{\nu^*}_{\tau}}\left[w+(1-w)\frac{{Z^{(\nu^*)^{(\nu,\epsilon,n)},\eta^*}_{\tau}}/{B^{(\nu^*)^{(\nu,\epsilon,n)}}_{\tau}}}{{Z^{\nu^*,\eta^*}_{\tau}}/{B^{\nu^*}_{\tau}}}\right],Y_{\tau})\frac{Z^{\nu^*,\eta^*}_{\tau}}{B^{\nu^*}_{\tau}}
\frac{1-\frac{Z^{(\nu^*)^{(\nu,\epsilon,n)},\eta}_{\tau}/B^{(\nu^*)^{(\nu,\epsilon,n)}}_{\tau}}{Z^{\nu,\eta}_{\tau}/B^{\nu}_{\tau}}}{\epsilon}
\nonumber\\
\hspace{-0.3cm}&\leq&\hspace{-0.3cm}
\lambda^* x^*_{h_A}(\lambda^*\frac{Z^{\nu^*,\eta^*}_{\tau}}{B^{\nu^*}_{\tau}}\left[w+(1-w)\frac{{Z^{(\nu^*)^{(\nu,\epsilon,n)},\eta^*}_{\tau}}/{B^{(\nu^*)^{(\nu,\epsilon,n)}}_{\tau}}}{{Z^{\nu^*,\eta^*}_{\tau}}/{B^{\nu^*}_{\tau}}}\right],Y_{\tau})\frac{Z^{\nu^*,\eta^*}_{\tau}}{ B^{\nu^*}_{\tau}}\Tilde{\Lambda}^{(n,\epsilon)}_{\tau},
\end{eqnarray}
where
$$\Tilde{\Lambda}^{(n,\epsilon)}_{t}:=\frac{1}{\epsilon}\left(1-\exp\left(-\epsilon(L^{\nu}_{t\wedge\tau_n}+N^{\nu}_{t\wedge\tau_n})-\frac{\epsilon^2}{2}\int_0^{t\wedge\tau_n}\|\sigma^{-1}(Y_s)(\nu_s-\nu^*_s)\|^2ds\right)\right).$$
It is straightforward to verify that $\Tilde{\Lambda}^{(n,\epsilon)}_{t}\rightarrow L^{\nu}_{t\wedge\tau_n}+N^{\nu}_{t\wedge\tau_n}$ and $
\frac{Z^{(\nu^*)^{(\nu,\epsilon,n)},\eta^*}_{t}/B^{(\nu^*)^{(\nu,\epsilon,n)}}_{t}}{Z^{\nu^*,\eta^*}_{t}/B^{\nu^*}_{t}}\rightarrow 1$ a.s. as $\epsilon\downarrow0$.
Hence, by the continuity of $x^*_{h_A}$ in its first argument, we have
\begin{eqnarray}\label{eq:pro4.6:e}
\hspace{-0.3cm}&&\hspace{-0.3cm}
\lim_{\epsilon\downarrow0}\lambda^* x^*_{h_A}(\lambda^*\frac{Z^{\nu^*,\eta^*}_{\tau}}{B^{\nu^*}_{\tau}}\left[w+(1-w)\frac{{Z^{(\nu^*)^{(\nu,\epsilon,n)},\eta^*}_{\tau}}/{B^{(\nu^*)^{(\nu,\epsilon,n)}}_{\tau}}}{{Z^{\nu^*,\eta^*}_{\tau}}/{B^{\nu^*}_{\tau}}}\right],Y_{\tau})\frac{Z^{\nu^*,\eta^*}_{\tau}}{ B^{\nu^*}_{\tau}}\Tilde{\Lambda}^{(n,\epsilon)}_{\tau}
\nonumber\\
\hspace{-0.3cm}&=&\hspace{-0.3cm}
\lambda^* x^*_{h_A}(\lambda^*\frac{Z^{\nu^*,\eta^*}_{\tau}}{B^{\nu^*}_{\tau}},Y_{\tau})\frac{Z^{\nu^*,\eta^*}_{\tau}}{ B^{\nu^*}_{\tau}}(L^{\nu}_{\tau_n}+N^{\nu}_{\tau_n}),\quad a.s..
\end{eqnarray}
Combining \eqref{eq:pro4.6:c}-\eqref{eq:pro4.6:e}, we deduce that
$$0\leq \limsup_{\epsilon\downarrow0}\mathbb{E}\left[\Lambda^{(n)}_{\epsilon}\right]\leq \mathbb{E}\left[\limsup_{\epsilon\downarrow0}\Lambda^{(n)}_{\epsilon}\right]\leq\mathbb{E}\left[\lambda^* x^*_{h_A}(\lambda^*\frac{Z^{\nu^*,\eta^*}_{\tau}}{B^{\nu^*}_{\tau}},Y_{\tau})\frac{Z^{\nu^*,\eta^*}_{\tau}}{ B^{\nu^*}_{\tau}}(L^{\nu}_{\tau_n}+N^{\nu}_{\tau_n})\right].$$
Repeating the same arguments following equation \eqref{eq:pro4.5:b} in Proposition \ref{prop4.5} with $X$ replaced by $x^*_{h_A}(\lambda^*\frac{Z^{\nu^*,\eta^*}_{\tau}}{B^{\nu^*}_{\tau}})$ and $x(\nu)$ replaced by $\mathbb{E}\left[\Phi_{h_A}(\lambda^*\frac{Z^{\nu,\eta^*}_{\tau}}{B^{\nu}_{\tau}},Y_{\tau})\right]$, establishes \eqref{eq:pro4.6:f}. 
By \eqref{supermaringale}, the process $X^{1,y,\pi^{\nu^*},\nu}Z^{\nu,\eta}/B^{\nu}$ is a $\mathbb{P}$-supermartingale, implying
\begin{align}
\label{5.31.w}
\mathbb{E}\left[X^{1,y,\pi^{\nu^*},\nu}_{\tau}Z^{\nu,\eta}_{\tau}/B^{\nu}_{\tau}\right]\leq 1.
\end{align}
Recalling the dynamics 
\begin{align}
\left\{
\begin{array}{ll}
dX^{1,y,\pi^{\nu^*},\nu}_t 
=\left[r(Y_t)+(\pi^{\nu^*}_t)^{\mathsf{T}}(\mu(Y_t)-r(Y_t)\mathbf{1})+\left(\delta(\nu_t)+(\pi^{\nu^*}_t)^{\mathsf{T}}\nu_t\right)\right]X^{1,y,\pi^{\nu^*},\nu}_tdt
\\
\quad \quad \quad \quad \quad \,\,\,\,\,\,
+X^{1,y,\pi^{\nu^*},\nu}_t(\pi_t^{\nu^*})^{\mathsf{T}}\sigma(Y_t)dW_{1t},
\\
X^{1,y,\pi^{\nu^*},\nu}_0=1,&
\end{array}
\right.
\nonumber
\end{align}
and
\begin{align}
\left\{
\begin{array}{ll}
dX^{1,y,\pi^{\nu^*},\nu^*}_t =\left[r(Y_t)+(\pi^{\nu^*}_t)^{\mathsf{T}}(\mu(Y_t)-r(Y_t)\mathbf{1})\right]X^{1,y,\pi^{\nu^*},\nu^*}_tdt
\\
\quad \quad \quad \quad \quad \,\,\,\,\,\,
+X^{1,y,\pi^{\nu^*},\nu^*}_t(\pi_t^{\nu^*})^{\mathsf{T}}\sigma(Y_t)dW_{1t},
\\
X^{1,y,\pi^{\nu^*},\nu^*}_0=1,\quad X^{1,y,\pi^{\nu^*},\nu^*}_{\tau}=X^{*},&
\end{array}
\right.
\nonumber
\end{align}
and using the fact that $\delta(\nu_t)+(\pi^{\nu^*}_t)^{\mathsf{T}}\nu_t\geq 0$, we obtain
\begin{align}
\label{5.32.w}
    X^{\nu^*,*}=x^*_{h_A}(\lambda^*\frac{Z^{\nu^*,\eta^*}_{\tau}}{B^{\nu^*}_{\tau}},Y_{\tau})=X^{1,y,\pi^{\nu^*},\nu^*}_{\tau}\leq X^{1,y,\pi^{\nu^*},\nu}_{\tau}.
\end{align}
Combining \eqref{5.31.w} and \eqref{5.32.w} yields
\begin{align}
\label{5.33.w}
\mathbb{E}\left[X^{\nu^*,*}Z^{\nu,\eta}_{\tau}/B^{\nu}_{\tau}\right]\leq 1.
\end{align}
Finally, applying the same argument as in the proof of Proposition \ref{prop2.1}, we have $\mathbb{E}\left[X^{\nu^*,*}\frac{Z^{\nu^*,\eta^*}_{\tau}}{B^{\nu^*}_{\tau}}\right]=1$. Together with \eqref{5.33.w}, this establishes \eqref{E=1.nu}.
\end{proof}

The next result states the existence of the solution to the dual problem \eqref{dual.pro.nu}. Its proof is analogous to that of Proposition \ref{existence.1} and \ref{existence.2}, and it is hence omitted.

\begin{prop}\label{prop5.4}
For any $(\lambda,y)\in\mathbb{R}_+\times\mathbb{R}$, there exists $(\nu^*,\eta^*)\in\mathcal{D}\times\mathcal{H}$ that solves the dual problem \eqref{dual.pro.nu}, that is
$$\Tilde{V}(\lambda,y)=\mathbb{E}\left[\Phi_{h_A}(\lambda\frac{Z^{\nu^*,\eta^*}_{\tau}}{B^{\nu^*}_{\tau}},Y_{\tau})\right].$$
\end{prop}

}

\section{Main Results in Periodic Evaluation Problem}\label{sec:fix}

Finally, in this section, we address the existence of the optimal portfolio to the original periodic evaluation problem \eqref{problem} over an infinite horizon. As discussed in Section \ref{sec:auxiliary}, it amounts to first prove the existence of the unique fixed point to the operator in \eqref{A*.def} and then construct and verify the optimal portfolio over the infinite horizon using the optimal solution to the one-period constrained terminal wealth optimization problem.

We first show that there indeed exists a unique $A^*(\cdot)$ solving equation \eqref{A*.def} and we give the upper and lower bounds for the fixed-point $A^*(\cdot)$.

\begin{prop}\label{fixed.power}
Define a functional $\Psi: {C}_b^{+}(\mathbb{R})\mapsto {C}_b^{+}(\mathbb{R})$ as
$$\Psi(A):=\Psi(A;y):=\alpha\sup_{X\in\mathcal{U}_0(1,y)}\mathbb{E}\left[e^{-\rho\tau}\frac{1}{\alpha}X_{\tau}^{\alpha}h({Y_{\tau}})+e^{-\rho \tau}\frac{1}{\alpha}A(Y_{\tau})X_{\tau}^{\alpha(1-\gamma)}\right].$$
Then $\Psi$ is a contraction on the metric space $({C}_b^{+}(\mathbb{R}),d)$ with the metric $d$ defined by $d(x,y):=\sup_{t\in\mathbb{R}}|x(t)-y(t)|,$  $x,y\in {C}_b^{+}(\mathbb{R})$, and $\Psi$ admits a unique fixed-point $A^*(\cdot)$ such that $A^*(y)=\Psi(A^*;y)$, for all $y\in\mathbb{R}$.
Moreover, the unique fixed-point $A^*(\cdot)$ of $\Psi$ satisfies, when $\alpha\in(0,1)$,
    $$\frac{me^{(\underline{r}\alpha-\rho)\tau}}{(1-e^{-(\rho-\underline{r}\alpha(1-\gamma))\tau})}\leq A^*(y)\leq \frac{e^{(\zeta(\alpha)-\rho)\tau}}{(1-e^{-(\rho-\zeta(\alpha(1-\gamma)))\tau})},\quad y\in\mathbb{R};$$
    and when $\alpha\in(-\infty,0)$,
    $$\frac{m e^{(\zeta(\alpha)-\rho)\tau}}{(1-e^{-(\rho-\zeta(\alpha(1-\gamma)))\tau})}\leq A^*(y)\leq  
 \frac{e^{(\underline{r}\alpha-\rho)\tau}}{(1-e^{-(\rho-\underline{r}\alpha(1-\gamma))\tau})},\quad y\in\mathbb{R}.$$
\end{prop}
\begin{proof}
Because the arguments for cases $\alpha\in(0,1)$ and $\alpha\in(-\infty,0)$ are similar, we shall only present the proof when $\alpha\in(0,1)$.
It is not difficult to verify that the metric space $(C^+_b(\mathbb{R}),d)$ is complete. Consider any $A_1,A_2\in C^+_b(\mathbb{R})$. By Propositions \ref{primaldualcons} and \ref{prop5.4}, there exists an optimizer $X^*_1\in{\mathcal{U}}_{0}(1,y)$ such that 
\begin{eqnarray}
\Psi(A_1) 
=\mathbb{E}\left[e^{-\rho\tau}(X^*_{1\tau})^{\alpha}h(Y_{\tau})+A_1(Y_{\tau})e^{-\rho\tau}(X^*_{1\tau})^{\alpha(1-\gamma)}\right],\nonumber
\end{eqnarray}
and, by the definition of metric $d$, there exists a point $y^*\in\mathbb{R}$ such that $d(\Psi(A_1),\Psi(A_2))=|\Psi(A_1;y^*)-\Psi(A_2;y^*)|$. Without loss of generality, we assume that $\Psi(A_1;y^*)\geq\Psi(A_2;y^*)$. Then, it holds that
\begin{eqnarray}
d(\Psi(A_1),\Psi(A_2))
\hspace{-0.3cm}&=&\hspace{-0.3cm}
\left[\mathbb{E}\left[e^{-\rho\tau}(X^*_{1\tau})^{\alpha}h(Y_{\tau})+A_1(Y_{\tau})e^{-\rho\tau}(X^*_{1\tau})^{\alpha(1-\gamma)}\right]\right.
\nonumber\\
\hspace{-0.3cm}&&\hspace{-0.3cm}
\left.-\alpha\sup_{X\in{\mathcal{U}}_{0}(1,y)}\mathbb{E}\left[e^{-\rho\tau}\frac{1}{\alpha}X^{\alpha}_{\tau}h(Y_{\tau})+A_2(Y_{\tau})e^{-\rho\tau}\frac{1}{\alpha}X_{\tau}^{\alpha(1-\gamma)}\right]\right]\Bigg|_{y=y^*}
\nonumber\\
\hspace{-0.3cm}&\leq&\hspace{-0.3cm}
\left[\mathbb{E}\left[e^{-\rho\tau}(X^*_{1\tau})^{\alpha}h(Y_{\tau})+A_1(Y_{\tau})e^{-\rho\tau}(X^*_{1\tau})^{\alpha(1-\gamma)}\right]\right.
\nonumber\\
\hspace{-0.3cm}&&\hspace{-0.3cm}
\left.-\mathbb{E}\left[e^{-\rho\tau}(X_{1\tau}^*)^{\alpha}h(Y_{\tau})+A_2(Y_{\tau})e^{-\rho\tau}(X^*_{1\tau})^{\alpha(1-\gamma)}\right]\right]\Big|_{y=y^*}
\nonumber\\
\hspace{-0.3cm}&\leq&\hspace{-0.3cm}
e^{-\rho\tau}\sup_{X\in{\mathcal{U}}_{0}(1,y)}\mathbb{E}\left[X_{\tau}^{\alpha(1-\gamma)}\right]\Big|_{y=y^*}\times d(A_1,A_2)
\nonumber\\
\hspace{-0.3cm}&\leq&\hspace{-0.3cm}
e^{-(\rho-\zeta(\alpha(1-\gamma)))\tau}d(A_1,A_2),\nonumber
\end{eqnarray}
where in the last inequality, we have used the fact that
$\sup_{X\in{\mathcal{U}}_{0}(1,y)}\mathbb{E}\left[X_{\tau}^{\alpha(1-\gamma)}\right]\leq e^{\zeta(\alpha(1-\gamma))\tau}$ for $y\in\mathbb{R}.$ Due to the standing Assumption \ref{ass1}, the existence and uniqueness of a fixed-point $A^*\in C^+_b(\mathbb{R})$ for $\Psi$
follow from the Banach contraction theorem.

Next, define $\underline{A}:=\inf_{y\in\mathbb{R}}A(y)$ and $\overline{A}:=\sup_{y\in\mathbb{R}}A(y)$ for any $A(\cdot)\in C_b^+(\mathbb{R})$.
Noting that $X=(X_t)_{t\geq0}=\left(e^{\int_0^{t}r(Y_s)ds}\right)_{t\geq0}$ is admissible to the problem \eqref{A*.def}, we readily obtain that 
\begin{eqnarray}
A^*(y)=\Psi(A^*;y)\geq me^{(\underline{r}\alpha-\rho)\tau}+\underline{A}^*e^{-(\rho-\underline{r}\alpha(1-\gamma))\tau},\quad y\in\mathbb{R}.\nonumber
\end{eqnarray}
Due to the arbitrariness of $y$, it holds that
$$\underline{A}^*\geq me^{(\underline{r}\alpha-\rho)\tau}+\underline{A}^*e^{-(\rho-\underline{r}\alpha(1-\gamma))\tau},$$
which yields the lower bound of $A^*$. For the upper bound, using \eqref{inv.pro}, we have
\begin{eqnarray}
A^*(y)=\Psi(A^*;y)
\hspace{-0.3cm}&=&\hspace{-0.3cm}
\alpha\sup_{X\in{\mathcal{U}}_{0}(1,y)}\mathbb{E}\left[\frac{1}{\alpha}e^{-\rho\tau}X_{\tau}^{\alpha}{h(Y_{\tau})}+\frac{1}{\alpha}{A^*(Y_{\tau})}e^{-\rho\tau}X_{\tau}^{\alpha(1-\gamma)}\right]
\nonumber\\
\hspace{-0.3cm}&\leq&\hspace{-0.3cm}
e^{-\rho\tau}\sup_{X\in{\mathcal{U}}_{0}(1,y)}\mathbb{E}\left[X_{\tau}^{\alpha}\right]+\overline{A}^*e^{-\rho\tau}\sup_{X\in{\mathcal{U}}_{0}(1,y)}\mathbb{E}\left[X_{\tau}^{\alpha(1-\gamma)}\right]
\nonumber\\
\hspace{-0.3cm}&\leq&\hspace{-0.3cm}
e^{(\zeta(\alpha)-\rho)\tau}+\overline{A}^*e^{-(\rho-\zeta(\alpha(1-\gamma)))\tau},\quad y\in\mathbb{R}.\nonumber
\end{eqnarray}
Again, thanks to the arbitrariness of $y$, we have that $$\overline{A}^*\leq e^{(\zeta(\alpha)-\rho)\tau}+\overline{A}^*e^{-(\rho-\zeta(\alpha(1-\gamma)))\tau}.$$ This, together with  $\rho>\zeta(\alpha(1-\gamma))\vee0$ from Assumption \ref{ass1}, gives the upper bound.
\end{proof}

With the previous preparations, our next goal is to construct and verify the optimal solution (i.e., the optimal trading strategy as well as the optimal value function) to the original periodic evaluation problem \eqref{problem}. 
 
{\color{black}
Recall that $T_n=n\tau$ for $n\geq 0$. For later use, we introduce the following notations
\begin{align}
    \mathcal{D}_n:=\left\{\nu;(\nu_{T_{n-1}+t})_{t\geq0}\in\mathcal{D}\right\}, \quad
    \mathcal{H}_n:=\left\{\eta;(\eta_{T_{n-1}+t})_{t\geq0}\in\mathcal{H}\right\},\quad n\geq 1,\nonumber
\end{align}
with $\mathcal{D}_1=\mathcal{D}$ and $\mathcal{H}_1=\mathcal{H}$.
For $(\nu_n,\eta_n)_{n\geq 1}$ with $(\nu_n,\eta_n)\in\mathcal{D}_n\times\mathcal{H}_n$ for any $n\geq 1$, define two processes  
\begin{eqnarray}\label{def.Z.eta}
\Tilde{\nu}_t
:=
\sum_{n=1}^{\infty}\nu_{nt}\mathbf{1}_{\{t\in[T_{n-1},T_n)\}}\quad\text{and}\quad \Tilde{\eta}_t:=\sum_{n=1}^{\infty}\eta_{nt}\mathbf{1}_{\{t\in[T_{n-1},T_n)\}},\quad t\geq0.
\end{eqnarray}
}

We are now ready to state and prove the main result of this paper.

\begin{thm}\label{thm3.1}
Define
\begin{eqnarray}
\label{6.2.w}
\Tilde{V}^{(n)}(\lambda,Y_{T_n})
\hspace{-0.3cm}&:=&\hspace{-0.3cm}
\inf_{\nu_n\in\mathcal{D}_n,\eta_n\in\mathcal{H}_n}\mathbb{E}\left[\Phi_{h_{A^*}}(\lambda\frac{Z_{T_n}^{\Tilde{\nu},\Tilde{\eta}}/B^{\Tilde{\nu}}_{T_n}}{Z^{\Tilde{\nu},\Tilde{\eta}}_{T_{n-1}}/B^{\Tilde{\nu}}_{T_{n-1}}},Y_{T_n})\bigg|\mathcal{F}_{T_n}\right],\quad n\geq 0,
\end{eqnarray}
and
$\lambda^*_n:=\arg\min_{\lambda>0}(\Tilde{V}^{(n)}(\lambda,Y_{T_n})+\lambda)$. Furthermore, let $(\nu^*_n,\eta_n^*)\in\mathcal{D}_n\times\mathcal{H}_n$ be the optimal solution to the dual problem \eqref{6.2.w} with $\lambda=\lambda^*_n$
and 
\begin{eqnarray}\label{X^*}
X^*_{T_n}:=X^*_{T_{n-1}}x^*_{h_{A^*}}\left(\lambda_n^*\frac{Z^{\Tilde{\nu}^*,\Tilde{\eta}^*}_{T_n}/B^{\Tilde{\nu}^*}_{T_n}}{Z^{\Tilde{\nu}^*,\Tilde{\eta}^*}_{T_{n-1}}/B^{\Tilde{\nu}^*}_{T_{n-1}}},Y_{T_{n}}\right),\quad n\geq 1,
\end{eqnarray}
with $X^*_{T_0}=x$, where the function $x^*_{h_{A^*}}(\cdot,\cdot)$ is defined by \eqref{phi(y)=h} with $A^*(\cdot)\in C^+_b(\mathbb{R})$ being the unique fixed-point of the functional $\Psi$, and 
the processes $\Tilde{\nu}^*,\,\Tilde{\eta}^*$ are given by
$$\Tilde{\nu}^*_t
:=
\sum_{n=1}^{\infty}\nu^*_{nt}\mathbf{1}_{\{t\in[T_{n-1},T_n)\}}\quad\text{and}\quad \Tilde{\eta}^*_t:=\sum_{n=1}^{\infty}\eta^*_{nt}\mathbf{1}_{\{t\in[T_{n-1},T_n)\}},\quad t\geq0.$$
Then, the value function of the problem \eqref{problem} is given by
\begin{eqnarray}\label{eq:V}
V(x,y)={\frac{1}{\alpha}}e^{-{\kappa}\tau\gamma\alpha}A^*(y)x^{\alpha(1-\gamma)},\quad(x,y)\in\mathbb{R}_+\times\mathbb{R}.
\end{eqnarray}
There exists a portfolio process $\pi^*$ such that the resulting wealth process $X^{\pi^*}$ solves the problem \eqref{problem} and satisfies $X^{\pi^*}_{T_n}=X^*_{T_n}$ for all $n\geq1$. 

\end{thm}
\begin{proof}
Using parallel arguments as those of Proposition \ref{prop5.4}, one can prove the solvability of \eqref{6.2.w} for each $n\geq 1$.
Define a series of functionals $\Psi_n:C^+_b(\mathbb{R})\mapsto C^+_b(\mathbb{R})$ as 
\begin{align}
    \label{psi.control.prob}
\Psi_n(A)&:=\Psi_n(A;y)
\nonumber\\
&:=\alpha\sup_{X\in\mathcal{U}_{T_n}(1,y)}\mathbb{E}\left[e^{-\rho \tau}\frac{1}{\alpha}X_{T_{n+1}}^{\alpha}h({Y_{T_{n+1}}})+e^{-\rho \tau}\frac{1}{\alpha}A(Y_{T_{n+1}})X_{T_{n+1}}^{\alpha(1-\gamma)}\bigg|\mathcal{F}_{T_n}\right],\quad n\geq 0,
\end{align}
where $\mathcal{U}_{T_n}(1,Y_{T_n})$ is given by \eqref{U_T}. 
For any $X\in\mathcal{U}_{T_n}(1,y)$, by \eqref{U_T}, we have $Y_{T_n}=y$, and there exists an admissible portfolio strategy $\hat{\pi}$ such that $X_{t}=X^{\hat{\pi}}_t$ for any $t\geq T_n$, and $X^{\hat{\pi}}_{T_n}=x$. Let $X^{x,y,\hat{\pi}}$ be as given by \eqref{admissible_set} with $\pi$ replaced as $\hat{\pi}$. Then, by the strong Markov property of $X^{\hat{\pi}}$ and $Y$, we have
\begin{align}
\label{6.6.w}
    &\mathbb{E}\left[e^{-\rho \tau}\frac{1}{\alpha}X_{T_{n+1}}^{\alpha}h({Y_{T_{n+1}}})+e^{-\rho \tau}\frac{1}{\alpha}A(Y_{T_{n+1}})X_{T_{n+1}}^{\alpha(1-\gamma)}\bigg|\mathcal{F}_{T_n}\right]
    \nonumber\\
    =&\mathbb{E}\left[e^{-\rho \tau}\frac{1}{\alpha}\left(X_{T_{n+1}}^{\hat{\pi}}\right)^{\alpha}h({Y_{T_{n+1}}})+e^{-\rho \tau}\frac{1}{\alpha}A(Y_{T_{n+1}})\left(X_{T_{n+1}}^{\hat{\pi}}\right)^{\alpha(1-\gamma)}\bigg|\mathcal{F}_{T_n}\right]
    \nonumber\\
    =&\mathbb{E}\left[e^{-\rho \tau}\frac{1}{\alpha}\left(X_{T_{n+1}}^{\hat{\pi}}\right)^{\alpha}h({Y_{T_{n+1}}})+e^{-\rho \tau}\frac{1}{\alpha}A(Y_{T_{n+1}})\left(X_{T_{n+1}}^{\hat{\pi}}\right)^{\alpha(1-\gamma)}\bigg|X_{T_n}^{\hat{\pi}}=x,Y_{T_n}=y\right]
    \nonumber\\
    =&\mathbb{E}\left[e^{-\rho \tau}\frac{1}{\alpha}\left(X_{T_{1}}^{x,y,\hat{\pi}}\right)^{\alpha}h({Y_{T_{1}}})+e^{-\rho \tau}\frac{1}{\alpha}A(Y_{T_{1}})\left(X_{T_{1}}^{x,y,\hat{\pi}}\right)^{\alpha(1-\gamma)}\right]
    \nonumber\\
    \leq & \Psi_{0}(A)=\Psi(A),\quad X\in\mathcal{U}_{T_n}(1,y).
\end{align}
On the other hand, for any $X\in\mathcal{U}_{0}(1,y)$, by \eqref{admissible_set}, we have $Y_{0}=y$, and there exists an admissible portfolio strategy $\tilde{\pi}$ such that $X_{t}=X^{x,y,\tilde{\pi}}_t$ for any $t\geq 0$, and $X^{x,y,\tilde{\pi}}_{0}=x$. Let $X^{\tilde{\pi}}$ be as given in \eqref{U_T} with $\pi$ replaced as $\tilde{\pi}$. Then, using an argument similar to that of \eqref{6.6.w}, we have
\begin{align}
    \label{6.7.w}
   &\mathbb{E}\left[e^{-\rho \tau}\frac{1}{\alpha}\left(X_{T_{1}}\right)^{\alpha}h({Y_{T_{1}}})+e^{-\rho \tau}\frac{1}{\alpha}A(Y_{T_{1}})\left(X_{T_{1}}\right)^{\alpha(1-\gamma)}\right]
    \nonumber\\
    =&
    \mathbb{E}\left[e^{-\rho \tau}\frac{1}{\alpha}\left(X_{T_{1}}^{x,y,\tilde{\pi}}\right)^{\alpha}h({Y_{T_{1}}})+e^{-\rho \tau}\frac{1}{\alpha}A(Y_{T_{1}})\left(X_{T_{1}}^{x,y,\tilde{\pi}}\right)^{\alpha(1-\gamma)}\right]
    \nonumber\\
    =&\mathbb{E}\left[e^{-\rho \tau}\frac{1}{\alpha}\left(X_{T_{n+1}}^{\tilde{\pi}}\right)^{\alpha}h({Y_{T_{n+1}}})+e^{-\rho \tau}\frac{1}{\alpha}A(Y_{T_{n+1}})\left(X_{T_{n+1}}^{\tilde{\pi}}\right)^{\alpha(1-\gamma)}\bigg|X_{T_n}^{\tilde{\pi}}=x, Y_{T_{n}}=y\right]
    \nonumber\\
    =&\mathbb{E}\left[e^{-\rho \tau}\frac{1}{\alpha}\left(X_{T_{n+1}}^{\tilde{\pi}}\right)^{\alpha}h({Y_{T_{n+1}}})+e^{-\rho \tau}\frac{1}{\alpha}A(Y_{T_{n+1}})\left(X_{T_{n+1}}^{\tilde{\pi}}\right)^{\alpha(1-\gamma)}\bigg|\mathcal{F}_{T_n}\right]\Bigg|_{X_{T_n}^{\tilde{\pi}}=x, Y_{T_{n}}=y}
    \nonumber\\
    \leq & \Psi_{n}(A),\quad X\in\mathcal{U}_{0}(1,y).
\end{align}
Combining \eqref{6.6.w} and \eqref{6.7.w} yields
\begin{align}
\label{6.8.w}
    \Psi_{n}(A)\equiv \Psi(A), \quad n\geq 0,\, A\in C^+_b(\mathbb{R}).
\end{align}
Thanks to the arguments across Sections \ref{sec:auxiliary}-\ref{sec:duality}, the control problem \eqref{psi.control.prob} with $n=0$ is solvable, which together with \eqref{6.8.w} implies that the control problem \eqref{psi.control.prob} is solvable for all $n\geq 0$. Then, by Proposition \ref{fixed.power}, $(\Psi_n)_{n\geq 0}$ shares the same unique fixed point $A^*$.

For any admissible wealth process $X\in\mathcal{U}_0(x,y)$, define a discrete-time stochastic process $D=(D_n)_{n\in\mathbb{N}}$ by
$$D_n:=\sum_{i=1}^n\frac{1}{\alpha}e^{-\rho T_i}\left(\frac{X_{T_i}}{(e^{{\kappa}\tau}X_{T_{i-1}})^{\gamma}}\right)^{\alpha}h(Y_{T_i})+{\frac{1}{\alpha}}e^{-{\kappa}\alpha\gamma\tau}A^*(Y_{T_n})e^{-\rho T_n}X_{T_n}^{\alpha(1-\gamma)},\quad n\geq1,$$
and $$D_0:={\frac{1}{\alpha}}e^{-{\kappa}\alpha\gamma\tau}A^*(y)x^{\alpha(1-\gamma)}.$$
Then, it holds that
\begin{eqnarray}\label{Dn+1}
D_{n+1}
\hspace{-0.3cm}&=&\hspace{-0.3cm}
D_{n}+e^{-\rho T_{n+1}}\frac{1}{\alpha}\bigg(\frac{X_{T_{n+1}}}{(e^{{\kappa}\tau}X_{T_{n}})^{\gamma}}\bigg)^{\alpha}h(Y_{T_{n+1}})-{\frac{1}{\alpha}}e^{-{\kappa}\tau\alpha\gamma}A^*(Y_{T_n})e^{-\rho T_n}X_{T_n}^{\alpha(1-\gamma)}
\nonumber\\
\hspace{-0.3cm}&&\hspace{-0.3cm}
+{\frac{1}{\alpha}}e^{-{\kappa}\tau\alpha\gamma}A^*(Y_{T_{n+1}})e^{-\rho T_{n+1}}X_{T_{n+1}}^{\alpha(1-\gamma)}
\nonumber\\
\hspace{-0.3cm}&=&\hspace{-0.3cm}
D_n+e^{-\rho T_n}e^{-{\kappa}\tau\alpha\gamma}\left[e^{-\rho\tau}\left(\frac{1}{\alpha}\bigg(\frac{X_{T_{n+1}}}{X_{T_{n}}^{\gamma}}\bigg)^{\alpha}h(Y_{T_{n+1}})+{\frac{1}{\alpha}}A^*(Y_{T_{n+1}})X^{\alpha(1-\gamma)}_{T_{n+1}}\right)\right.
\nonumber\\
\hspace{-0.3cm}&&\hspace{-0.3cm}
\left.-{\frac{1}{\alpha}}A^*(Y_{T_n})X^{\alpha(1-\gamma)}_{T_n}\right].
\end{eqnarray}
Taking the conditional expectation on both sides of \eqref{Dn+1}, we get
\begin{eqnarray}\label{3.4}
\mathbb{E}[D_{n+1}|\mathcal{F}_{T_n}]
\hspace{-0.3cm}&=&\hspace{-0.3cm}
D_n+e^{-\rho T_n}e^{-{\kappa}\tau\alpha\gamma}X^{\alpha(1-\gamma)}_{T_n}\left[e^{-\rho\tau}\mathbb{E}\left[\frac{1}{\alpha}\bigg(\frac{X_{T_{n+1}}}{X_{T_n}}\bigg)^{\alpha}h(Y_{T_{n+1}})\right.\right.
\nonumber\\
\hspace{-0.3cm}&&\hspace{3cm}
\left.\left.+{{\frac{1}{\alpha}}A^*(Y_{T_{n+1}})}\bigg(\frac{X_{T_{n+1}}}{X_{T_n}}\bigg)^{\alpha(1-\gamma)}\bigg|\mathcal{F}_{T_n}\right]-{\frac{1}{\alpha}}A^*(Y_{T_n})\right].
\end{eqnarray}
Note that $X\in\mathcal{U}_0(x,y)$ implies $({X_{t}}/{X_{T_n}})_{t\geq T_{n}}\in\mathcal{U}_{T_{n}}(1,y)$. Then
\begin{eqnarray}\label{3.5}
\hspace{-0.3cm}&&\hspace{-0.3cm}
e^{-\rho\tau}\mathbb{E}\left[\frac{1}{\alpha}\bigg(\frac{X_{T_{n+1}}}{X_{T_n}}\bigg)^{\alpha}h(Y_{T_{n+1}})+{{\frac{1}{\alpha}}A^*(Y_{T_{n+1}})}\bigg(\frac{X_{T_{n+1}}}{X_{T_n}}\bigg)^{\alpha(1-\gamma)}\bigg|\mathcal{F}_{T_n}\right]
\nonumber\\
\hspace{-0.3cm}&\leq&\hspace{-0.3cm}
e^{-\rho\tau}\sup_{X\in{\mathcal{U}}_{T_n}(1,Y_{T_n})}\mathbb{E}\left[\frac{1}{\alpha}X_{T_{n+1}}^{\alpha}h(Y_{T_{n+1}})+{\frac{1}{\alpha}}A^*(Y_{T_{n+1}})X_{T_{n+1}}^{\alpha(1-\gamma)}\Big|\mathcal{F}_{T_n}\right]
\nonumber\\
\hspace{-0.3cm}&=&\hspace{-0.3cm}
{\frac{1}{\alpha}}\Psi_n(A^*;Y_{T_n}),\quad n\in\mathbb{N},
\end{eqnarray}
which, together with \eqref{3.4} and the fact of $\Psi_n(A^*;\cdot)=A^*(\cdot)$, implies
\begin{eqnarray}\label{3.6}
\mathbb{E}[D_{n+1}|\mathcal{F}_{T_n}]\leq D_n+{\frac{1}{\alpha}}e^{-\rho T_n}e^{-{\kappa}\tau\alpha\gamma}X_{T_n}^{\alpha(1-\gamma)}\left[\Psi_n(A^*;Y_{T_n})-A^*(Y_{T_n})\right]=D_n.
\end{eqnarray}
Hence, $(D_n)_{n\in\mathbb{N}}$ is a $\{\mathcal{F}_{T_n}\}$-supermartingale, and it follows that
\begin{eqnarray}
\hspace{-0.3cm}&&\hspace{-0.3cm}
{\frac{1}{\alpha}}e^{-{\kappa}\tau\alpha\gamma}A^*(y)x^{\alpha(1-\gamma)}
\nonumber\\
\hspace{-0.3cm}&=&\hspace{-0.3cm}
D_0\geq
\mathbb{E}\left[\sum_{i=1}^ne^{-\rho T_i}\frac{1}{\alpha}\bigg(\frac{X_{T_i}}{(e^{{\kappa}\tau}X_{T_{i-1}})^{\gamma}}\bigg)^{\alpha}h(Y_{T_{i}})+{\frac{1}{\alpha}}e^{-{\kappa}\tau\alpha\gamma}A^*(Y_{T_n})e^{-\rho T_n}X_{T_n}^{\alpha(1-\gamma)}\right],\nonumber
\end{eqnarray}
which yields that
\begin{eqnarray}
\mathbb{E}\left[\sum_{i=1}^ne^{-\rho T_i}\frac{1}{\alpha}\bigg(\frac{X_{T_i}}{(e^{{\kappa}\tau}X_{T_{i-1}})^{\gamma}}\bigg)^{\alpha}h(Y_{T_{i}})\right]
\hspace{-0.3cm}&\leq&\hspace{-0.3cm}
{\frac{1}{\alpha}}e^{-{\kappa}\tau\alpha\gamma}A^*(y)x^{\alpha(1-\gamma)}-{\frac{1}{\alpha}}e^{-\rho T_n}e^{-{\kappa}\tau\alpha\gamma}
\mathbb{E}\left[A^*(Y_{T_n})X^{\alpha(1-\gamma)}_{T_n}\right].\nonumber
\end{eqnarray}
By Assumption \ref{ass1} (i.e., $\rho>\zeta(\alpha(1-\gamma))\vee0$), \eqref{3.13.v0}, the fact $A^*(\cdot)\in C^+_b(\mathbb{R})$ as well as the monotone convergence theorem, we have
\begin{eqnarray}
\mathbb{E}\left[\sum_{i=1}^{\infty}e^{-\rho T_i}\frac{1}{\alpha}\bigg(\frac{X_{T_i}}{(e^{{\kappa}\tau}X_{T_{i-1}})^{\gamma}}\bigg)^{\alpha}h(Y_{T_{i}})\right]\leq {\frac{1}{\alpha}}e^{-{\kappa}\tau\alpha\gamma}A^*(y)x^{\alpha(1-\gamma)},\nonumber
\end{eqnarray}
and hence
\begin{eqnarray}\label{eq:thm6.1.V<}
V(x,y)=\sup_{X\in\mathcal{U}_0(x,y)}\mathbb{E}\left[\sum_{i=1}^{\infty}e^{-\rho T_i}\frac{1}{\alpha}\bigg(\frac{X_{T_i}}{(e^{{\kappa}\tau}X_{T_{i-1}})^{\gamma}}\bigg)^{\alpha}h(Y_{T_{i}})\right]\leq {\frac{1}{\alpha}}e^{-{\kappa}\tau\alpha\gamma}A^*(y)x^{\alpha(1-\gamma)}.
\end{eqnarray}
It remains to show the reverse inequality of \eqref{eq:thm6.1.V<}.
Recall that $\lambda^*_n:=\arg\min_{\lambda>0}(\Tilde{V}^{(n)}(\lambda,Y_{T_n})+\lambda)$ and $(\nu^*_n,\eta_n^*)\in\mathcal{D}_n\times\mathcal{H}_n$ is the optimal solution to the dual problem \eqref{6.2.w} with $\lambda=\lambda^*_n$.
Then, by Proposition \ref{primaldualcons}, there exists a portfolio process $\pi_n^*$ such that the process $X^{\pi^*_n}/X^{\pi^*}_{T_{n}} \in\mathcal{U}_{T_n}(1,Y_{T_n})$ satisfies
\begin{eqnarray}\label{thm6.1.condi1}
\hspace{-0.3cm}&&\hspace{-0.3cm}
e^{-\rho\tau}\mathbb{E}\left[\left.\frac{1}{\alpha}\bigg(\frac{X^{\pi_n^*}_{T_{n+1}}}{X^{\pi^*}_{T_{n}}}\bigg)^{\alpha}h(Y_{T_{n+1}})+{\frac{1}{\alpha}}A^*(Y_{T_{n+1}})\bigg(\frac{X^{\pi_n^*}_{T_{n+1}}}{X^{\pi^*}_{T_{n}}}\bigg)^{\alpha(1-\gamma)}\right|\mathcal{F}_{T_n}\right]
\nonumber\\
\hspace{-0.3cm}&=&\hspace{-0.3cm}
e^{-\rho\tau}\sup_{X\in{\mathcal{U}}_{T_n}(1,Y_{T_n})}\mathbb{E}\left[\frac{1}{\alpha}X^{\alpha}_{T_{n+1}}h(Y_{T_{n+1}})+{\frac{1}{\alpha}}A^*(Y_{T_{n+1}})X_{T_{n+1}}^{\alpha(1-\gamma)}\Big|\mathcal{F}_{T_n}\right]={\frac{1}{\alpha}}\Psi_n(A^*;Y_{T_n}),
\end{eqnarray}
and
\begin{align}
    \label{thm6.1.condi2}
dX^{\pi_n^*}_t=&
\left[r(Y_t)+\delta(\nu^*_{nt})+(\pi^*_{nt})^{\mathsf{T}}(\mu(Y_t)+\nu^*_{nt}-r(Y_t)\mathbf{1})\right]X^{\pi_n^*}_tdt\nonumber\\
&
+X^{\pi^*_n}_t(\pi_{nt}^{*})^{\mathsf{T}}\sigma(Y_t)dW_{1t}
\nonumber\\
=&\left[r(Y_t)+(\pi^*_{nt})^{\mathsf{T}}(\mu(Y_t)-r(Y_t)\mathbf{1})\right]X^{\pi_n^*}_tdt
+X^{\pi^*_n}_t(\pi_{nt}^{*})^{\mathsf{T}}\sigma(Y_t)dW_{1t},\quad t\in[T_n,T_{n+1}),
\end{align}
with
\begin{align}
    \label{thm6.1.condi3} {X^{\pi_n^*}_{T_{n+1}}}/{X^{\pi^*}_{T_n}}=x^*_{h_{A^*}}\left(\lambda^*_{n+1}\frac{Z^{\Tilde{\nu}^*,\Tilde{\eta}^*}_{T_{n+1}}/B^{\Tilde{\nu}^*}_{T_{n+1}}}{Z^{\Tilde{\nu}^*,\Tilde{\eta}^*}_{T_n}/B^{\Tilde{\nu}^*}_{T_n}},Y_{T_{n+1}}\right).
\end{align}
By Proposition \ref{primaldualcons}, we also have
\begin{align}\label{3.11.power}
&\mathbb{E}\left[\frac{Z^{\Tilde{\nu}^*,\Tilde{\eta}^*}_{T_{n+1}}/B^{\Tilde{\nu}^*}_{T_{n+1}}}{Z^{\Tilde{\nu}^*,\Tilde{\eta}^*}_{T_n}/B^{\Tilde{\nu}^*}_{T_n}}x^*_{h_{A^*}}\left(\lambda^*_{n+1}\frac{Z^{\Tilde{\nu}^*,\Tilde{\eta}^*}_{T_{n+1}}/B^{\Tilde{\nu}^*}_{T_{n+1}}}{Z^{\Tilde{\nu}^*,\Tilde{\eta}^*}_{T_n}/B^{\Tilde{\nu}^*}_{T_n}},Y_{T_{n+1}}\right)\Bigg|\mathcal{F}_{T_n}\right]
\nonumber\\
=&\mathbb{E}\left[x^*_{h_{A^*}}\left(\lambda_1^*\frac{Z^{\Tilde{\nu}^*,\Tilde{\eta}^*}_{\tau}}{B^{\Tilde{\nu}^*}_{\tau}},Y_{\tau}\right)\frac{Z^{\Tilde{\nu}^*,\Tilde{\eta}^*}_{\tau}}{B^{\Tilde{\nu}^*}_{\tau}}\Bigg|\mathcal{F}_0\right]=1,\quad n\geq 1.
\end{align}
Define
$$D^*_n:=\sum_{i=1}^ne^{-\rho T_i}\frac{1}{\alpha}\bigg(\frac{X^{\pi_i^*}_{T_i}}{{(e^{{\kappa}\tau}X^{\pi_i^*}_{T_{i-1}})}^{\gamma}}\bigg)^{\alpha}h(Y_{T_{i}})+{\frac{1}{\alpha}}e^{-{\kappa}\alpha\gamma\tau}A^*(Y_{T_n})e^{-\rho T_n}{(X^{\pi_n^*}_{T_n})}^{\alpha(1-\gamma)}, \quad n\geq 1,$$
and
$$D_0^*:={\frac{1}{\alpha}}e^{-{\kappa}\alpha\gamma\tau}A^*(y)x^{\alpha(1-\gamma)}.$$
Using the same arguments for \eqref{3.6}, we conclude that
\begin{eqnarray}
\mathbb{E}[D^*_{n+1}|\mathcal{F}_{T_n}]=D^*_n+{\frac{1}{\alpha}}e^{-\rho T_n}e^{-{\kappa}\alpha\gamma\tau}{(X^{\pi_n^*}_{T_n})}^{\alpha(1-\gamma)}\left[\Psi(A^*;Y_{T_n})-A^*(Y_{T_n})\right]=D^*_n, \quad n\geq 0,\nonumber
\end{eqnarray}
which gives that $D^*_n$ is a $\{\mathcal{F}_{T_n}\}$-martingale. Hence, we have
\begin{eqnarray}
\hspace{-0.3cm}&&\hspace{-0.3cm}
\mathbb{E}\left[\sum_{i=1}^{n}e^{-\rho T_i}\frac{1}{\alpha}\bigg(\frac{X^{\pi_i^*}_{T_i}}{(e^{{\kappa}\tau}X^{\pi_i^*}_{T_{i-1}})^{\gamma}}\bigg)^{\alpha}h(Y_{T_{i}})\right]
\nonumber\\
\hspace{-0.3cm}&=&\hspace{-0.3cm}
{\frac{1}{\alpha}}e^{-{\kappa}\alpha\gamma\tau}A^*(y)x^{\alpha(1-\gamma)}-{\frac{1}{\alpha}}e^{-\rho T_n}e^{-{\kappa}\alpha\gamma\tau}\mathbb{E}\left[A^*(Y_{T_n})(X^{\pi_n^*}_{T_n})^{\alpha(1-\gamma)}\right].\nonumber
\end{eqnarray}
By Assumption \ref{ass1}, \eqref{3.13.v0}, the fact $A^*(\cdot)\in C^+_b(\mathbb{R})$ and monotone convergence theorem, we have
\begin{eqnarray}\label{eq:thm6.1.V=}
\mathbb{E}\left[\sum_{i=1}^{\infty}e^{-\rho T_i}\frac{1}{\alpha}\bigg(\frac{X^{\pi_i^*}_{T_i}}{(e^{{\kappa}\tau}X^{\pi_i^*}_{T_{i-1}})^{\gamma}}\bigg)^{\alpha}h(Y_{T_{i}})\right]={\frac{1}{\alpha}}e^{-{\kappa}\alpha\gamma\tau}A^*(y)x^{\alpha(1-\gamma)}.
\end{eqnarray}
Furthermore, set
$$\pi^*=(\pi^*_t)_{t\geq0}:=\left(\sum_{i=0}^{\infty}\pi^*_{it}\mathbf{1}_{\{t\in[T_{i},T_{i+1})\}}\right)_{t\geq0}.$$
Then, the wealth process $X^{\pi^*}$ under the control $\pi^*$, subject to the boundary conditions $X^{\pi^*}_{0}=x$ and $X^{\pi^*}_{T_{n}-}=X^{\pi^*}_{T_{n}}$ for each $n\geq 1$, is such that
$$X^{\pi^*}=(X^{\pi^*}_t)_{t\geq0}=\left(\sum_{i=0}^{\infty}X^{\pi^*_i}_{t}\mathbf{1}_{\{t\in[T_{i},T_{i+1})\}}\right)_{t\geq0} \in\mathcal{U}_0(x,y).$$
In addition, \eqref{eq:thm6.1.V=} can be rewritten as
\begin{eqnarray}
\mathbb{E}\left[\sum_{i=1}^{\infty}e^{-\rho T_i}\frac{1}{\alpha}\bigg(\frac{X^{\pi^*}_{T_i}}{(e^{{\kappa}\tau}X^{\pi^*}_{T_{i-1}})^{\gamma}}\bigg)^{\alpha}h(Y_{T_{i}})\right]={\frac{1}{\alpha}}e^{-{\kappa}\alpha\gamma\tau}A^*(y)x^{\alpha(1-\gamma)},
\end{eqnarray}
which combined with \eqref{eq:thm6.1.V<} implies \eqref{eq:V}, and $X^{\pi^*}$ solves the control problem \eqref{problem}.
\end{proof}

\begin{rem}
The characterization above is primarily theoretical. For numerical implementation, one feasible route is to initialize a bounded positive function $A_0$ and iterate $A_{m+1}=\Psi(A_m)$. For each fixed $A_m$, evaluating $\Psi(A_m)$ requires solving the auxiliary constrained terminal wealth problem. In principle, this can be done by solving the associated HJB or dual HJB equation for the dual minimization problem. In the latter approach, the minimizing parameter processes $(\nu^*,\eta^*)$ can be obtained numerically in the dual problem. Once the dual optimizers $(\nu^*,\eta^*)$ and the multiplier $\lambda^*$ are computed, the optimal wealth and portfolio can be recovered from the dual representation and the martingale representation. A full development and convergence analysis of such numerical algorithm is beyond the scope of this paper, but it has been used in our concrete examples in Section 7 with satisfactory performance. 
\end{rem}

\begin{rem}
Theorem \ref{thm3.1} is analogous to Theorem 1 of \cite{TZ23} as both results rely on a fixed-point and verification argument. The key difference lies in the targets being verified. In \cite{TZ23}, the value function is $V(x)=A^*x^\alpha$ with a scalar fixed point $A^*$, and the optimal wealth at the evaluation dates is represented by $X^*_{T_i}=X^*_{T_{i-1}}
y^*_{A^*}\left(\lambda^*\frac{Z_{T_i}}{Z_{T_{i-1}}}\right)$. In the present ratio-type model with stochastic factors and convex constraints, the fixed point becomes a function \(A^*(y)\in C_b^+(\mathbb{R})\), and
$
V(x,y)=\frac{1}{\alpha}e^{-\kappa\tau\gamma\alpha}
A^*(y)x^{\alpha(1-\gamma)} .
$
Moreover, the optimal wealth at each evaluation date is
$
X^*_{T_n}=X^*_{T_{n-1}}
x^*_{h_{A^*}}\left(\lambda_n^*
\frac{Z^{\Tilde{\nu}^*,\Tilde{\eta}^*}_{T_n}/B^{\Tilde{\nu}^*}_{T_n}}
{Z^{\Tilde{\nu}^*,\Tilde{\eta}^*}_{T_{n-1}}/B^{\Tilde{\nu}^*}_{T_{n-1}}},
Y_{T_n}\right),
$
where $\nu^*$ is the constraint parameter and $\eta^*$ is the market completion parameter. Thus our proof requires several additional steps to facilitate the fixed-point verification: the auxiliary one-period constrained terminal wealth problem is solved through the fictitious-market duality with $(\nu,\eta)$, the fixed point is established on the function space $C_b^+(\mathbb{R})$, and the period-by-period optimizers are concatenated while preserving the convex constraint. The verification of supermartingale calls for the continuation term
$
\frac{1}{\alpha}e^{-\kappa\tau\gamma\alpha}
A^*(Y_{T_n})e^{-\rho T_n}X_{T_n}^{\alpha(1-\gamma)},
$
rather than $A^*e^{-\rho T_n}X_{T_n}^{\alpha}$ in \cite{TZ23}. This highlights how the
ratio-type criterion, interplayed with the stochastic factor and the convex constraint, substantially changes the arguments for fixed-point and verification.
\end{rem}

\section{Constructive Arguments in Special Cases of $\gamma=1$ and $\gamma=0$}

In this section, we present an illustrative example for the relative performance parameters $\gamma=1$ or $\gamma=0$. 
When $\gamma=1$, the modified utility simplifies to
\begin{eqnarray}
h_A(x,y)=\frac{1}{\alpha}x^{\alpha}h(y)+\frac{1}{\alpha}A(y),\quad (x,y)\in\mathbb{R}_+\times\mathbb{R}.\nonumber
\end{eqnarray}

The auxiliary one-period terminal wealth
optimization problem \eqref{A*.def} is rewritten as
\begin{eqnarray}
\label{7.1}
A^*(y)=\alpha e^{-\rho\tau}\sup_{X\in\mathcal{U}_0(1,y)}\mathbb{E}\left[\frac{1}{\alpha}X_{\tau}^{\alpha}h(Y_{\tau})\right]+\mathbb{E}\left[e^{-\rho \tau}A^*(Y_{\tau})\right].
\end{eqnarray}
Hence, we need to solve the following problem
\begin{eqnarray}\label{ex:prob}
\sup_{X\in\mathcal{U}_0(1,y)}\mathbb{E}\left[\frac{1}{\alpha}X_{\tau}^{\alpha}h(Y_{\tau})\right].
\end{eqnarray}
Note that for each fixed $y$, the Legendre-Fenchel transform of the utility $\frac{1}{\alpha}x^{\alpha}h(y)$ is given by 
$\frac{1-\alpha}{\alpha}z^{\frac{\alpha}{\alpha-1}}h(y)^{\frac{1}{1-\alpha}}$. Consequently, the dual problem of \eqref{ex:prob} is formulated as 
\begin{eqnarray}\label{ex:dual}
\inf_{\nu\in\mathcal{D},\eta\in\mathcal{H},\lambda>0}\left\{\mathbb{E}\left[\frac{1-\alpha}{\alpha}\left(\lambda\frac{Z^{\nu,\eta}_{\tau}}{B^{\nu}_{\tau}}\right)^{\frac{\alpha}{\alpha-1}}h(Y_{\tau})^{\frac{1}{1-\alpha}}\right]+\lambda \right\}.
\end{eqnarray}
For fixed $(\eta,\nu)\in\mathcal{H}\times \mathcal{D}$,  the optimal parameter $\lambda$ is given by
\begin{eqnarray}\label{ex:lambda}
\lambda^*(\eta,\nu)=\left(\mathbb{E}\left[\left(\frac{Z^{\nu,\eta}_{\tau}}{B^{\nu}_{\tau}}\right)^{\frac{\alpha}{\alpha-1}}h(Y_{\tau})^{\frac{1}{1-\alpha}}\right]\right)^{1-\alpha}.
\end{eqnarray}
Substituting this into \eqref{ex:dual} yields
\begin{eqnarray}
\label{ex:dual2}
\inf_{\nu\in\mathcal{D},\eta\in\mathcal{H}}\left\{\frac{1}{\alpha}\left(\mathbb{E}\left[\left(\frac{Z^{\nu,\eta}_{\tau}}{B^{\nu}_{\tau}}\right)^{\frac{\alpha}{\alpha-1}}h(Y_{\tau})^{\frac{1}{1-\alpha}}\right]\right)^{1-\alpha}\right\}.
\end{eqnarray}
Depending on the sign of the parameter $\alpha$, problem \eqref{ex:dual2} is equivalent to 
\begin{eqnarray}
\inf_{\nu\in\mathcal{D},\eta\in\mathcal{H}}\left(\sup_{\nu\in\mathcal{D},\eta\in\mathcal{H}}\right)\,\,\mathbb{E}\left[\left(\frac{Z^{\nu,\eta}_{\tau}}{B^{\nu}_{\tau}}\right)^{\frac{\alpha}{\alpha-1}}h(Y_{\tau})^{\frac{1}{1-\alpha}}\right],\quad\alpha\in(0,1)\,\,(\alpha\in(-\infty,0)).
\end{eqnarray}
Let $M^{\nu,\eta}_t:=\frac{Z^{\nu,\eta}_{t}}{B^{\nu}_{t}}$ for $t\geq 0$. From \eqref{new-fic-model} and \eqref{Z.process}, we have 
\begin{eqnarray}\label{dM2}
dM^{\nu,\eta}_t = M_t^{\nu,\eta} \left[ -(r(Y_t)+ \delta(\nu_t))dt - (\theta^\nu(Y_t))^\top dW_{1t} + \eta_t dW_{2t} \right],\quad t\in[0,\tau],
\end{eqnarray}
with $M^{\nu,\eta}_0=1$.
For the case $\alpha\in(0,1)$, define
\begin{eqnarray}
\label{6.43.w}
w(t,m,y)
\hspace{-0.3cm}&:=&\hspace{-0.3cm}
\inf_{\nu\in\mathcal{D},\eta\in\mathcal{H}}\mathbb{E}\left[(M^{\nu,\eta}_{\tau})^{\frac{\alpha}{\alpha-1}}h(Y_{\tau})^{\frac{1}{1-\alpha}}\Big|M^{\nu,\eta}_{t}=m,Y_t=y\right]
\nonumber\\
\hspace{-0.3cm}&=&\hspace{-0.3cm}
\inf_{\nu\in\mathcal{D},\eta\in\mathcal{H}}\mathbb{E}\left[(M^{\nu,\eta}_{\tau})^{\frac{\alpha}{\alpha-1}}h(Y_{\tau})^{\frac{1}{1-\alpha}}\Big|M^{\nu,\eta}_{t}=1,Y_t=y\right]m^{\frac{\alpha}{\alpha-1}}\nonumber\\
\hspace{-0.3cm}&:=&\hspace{-0.3cm}
m^{\frac{\alpha}{\alpha-1}}v(t,y).
\end{eqnarray}
Then, the Hamilton-Jacobi-Bellman (HJB) equation associated with the control problem \eqref{6.43.w} is
\begin{eqnarray}\label{ex2:HJBw}
0\hspace{-0.3cm}&=&\hspace{-0.3cm}
w_t+b(y)w_y+\frac{1}{2}\beta(y)^2w_{yy}-mr(y)w_m
\nonumber\\
\hspace{-0.3cm}&&\hspace{-0.3cm}
+\inf_{\nu\in\mathcal{D}}\left\{-m\delta(\nu)w_m+\frac{1}{2}m^2\|\theta^{\nu}(y)\|^2w_{mm}-m\beta(y)q\theta^{\nu}(y)w_{my}\right\}
\nonumber\\
\hspace{-0.3cm}&&\hspace{-0.3cm}
+\inf_{\eta\in\mathcal{H}}\left\{\frac{1}{2}m^2\eta^2w_{mm}+m\beta(y)\eta\sqrt{1-\|q\|^2}w_{my}\right\},
\end{eqnarray}
with boundary condition $w(\tau,m,y)=m^{\frac{\alpha}{\alpha-1}}h(y)^{\frac{1}{1-\alpha}}$. 
Substituting $w(t,m,y)=v(t,y)m^{\frac{\alpha}{\alpha-1}}$ into \eqref{ex2:HJBw}, we obtain
\begin{eqnarray}
\label{6.45.w}
0\hspace{-0.3cm}&=&\hspace{-0.3cm}
v_t+b(y)v_y+\frac{1}{2}\beta(y)^2v_{yy}-\frac{\alpha}{\alpha-1}r(y)v
\nonumber\\
\hspace{-0.3cm}&&\hspace{-0.3cm}
+\inf_{\nu\in\tilde{K}}\left\{\frac{\alpha}{1-\alpha}\delta(\nu)v+\frac{\alpha}{2(1-\alpha)^2}\|\theta^{\nu}(y)\|^2v+\frac{\alpha}{1-\alpha}\beta(y)q\theta^{\nu}(y)v_{y}\right\}
\nonumber\\
\hspace{-0.3cm}&&\hspace{-0.3cm}
+\inf_{\eta\in\mathbb{R}}\left\{\frac{\alpha}{2(\alpha-1)^2}\eta^2v+\frac{\alpha}{\alpha-1}\beta(y)\eta\sqrt{1-\|q\|^2}v_{y}\right\},
\end{eqnarray}
with boundary condition $v(\tau,y)=h(y)^{\frac{1}{1-\alpha}}.$
The well-posedness of solutions to the HJB equation \eqref{6.45.w} follows from standard arguments. For brevity, we assume the existence and uniqueness of a classical (or viscosity) solution to \eqref{6.45.w}. Consequently, the optimal constraint parameter process $\nu^*$ and market completion parameter process $\eta^*$ are given by
\begin{align}\label{ex:solution.dual}
    \left\{
    \begin{array}{ll}
    \nu^{*}_{t}=\arg\min\limits_{\nu\in\tilde{K}}\left\{\left(\frac{\alpha}{1-\alpha}\delta(\nu)+\frac{\alpha}{2(1-\alpha)^2}\|\theta^{\nu}(Y_{t})\|^2\right)v(t,Y_{t})+\frac{\alpha}{1-\alpha}\beta(Y_{t})q\theta^{\nu}(Y_{t})v_{y}(t,Y_{t})\right\},&  \\
    \eta^{*}_{t}= 
    \frac{(1-\alpha)\sqrt{1-\|q\|^2}\beta(Y_{t})v_{y}(t,Y_{t})}{ v(t,Y_{t})}.& 
    \end{array}
    \right.
\end{align}
Define the function 
\begin{align}
    \label{ex2:mathX}
\mathcal{X}(t,y,m):=&\left.\mathbb{E}\left[\left.\frac{1}{\lambda^* M^{\nu^*,\eta^*}_{t}}(\lambda^* M^{\nu^*,\eta^*}_{\tau})^{\frac{\alpha}{\alpha-1}}
h(Y_\tau)^{\frac{1}{1-\alpha}}
\right|\mathcal{F}_{t}\right]\right|_{M^{\nu^*,\eta^*}_{t}={m}/{\lambda^{*}},Y_t=y}
\nonumber\\
=& \frac{1}{m}
\mathbb{E}\bigg[\left[\lambda^* M^{\nu^*,\eta^*}_{t}
e^{-\int_t^\tau [(\theta^{\nu^*}(Y_s))^{\mathsf{T}}dW_{1s}-\eta_s^*dW_{2s}]-\int_t^\tau\left[\frac{1}{2}[\|\theta^{\nu^*}(Y_s)\|^2+[\eta_s^*]^2]-r(Y_t)-\delta(\nu_{s}^*)\right]ds}\right]^{\frac{\alpha}{\alpha-1}}
\nonumber\\
&\quad \quad \,\,\,\,\times h(Y_\tau)^{\frac{1}{1-\alpha}}\bigg|
M^{\nu^*,\eta^*}_{t}=\frac{m}{\lambda^{*}},Y_t=y
\bigg]
\nonumber\\
=&m^{\frac{1}{\alpha-1}}\mathbb{E}\left[\left.\left[M^{\nu^*,\eta^*}_{\tau}\right]^{\frac{\alpha}{\alpha-1}}h(Y_\tau)^{\frac{1}{1-\alpha}}\right|M^{\nu^*,\eta^*}_{t}=1,Y_t=y\right]
\nonumber\\
:=&C(t,y) m^{\frac{1}{\alpha-1}}.
\end{align}
Denote the inverse function of $\mathcal{X}(t,y,\cdot)$ by $\mathcal{Y}(t,y,\cdot)$, i.e., $\mathcal{Y}(t,y,\mathcal{X}(t,y,m))=m$. Then
\begin{align}
\label{6.47.ww}
    \mathcal{Y}(t,y,x)=C(t,y)^{1-\alpha}x^{\alpha-1},\quad x>0.
\end{align}
Let $\eta^*$ and $\nu^*$ be given by \eqref{ex:solution.dual}, and let $\lambda^*=\lambda^*(\eta^*,\nu^*)$ with $\lambda^*(\cdot,\cdot)$ defined by \eqref{ex:lambda}. Then $(\nu^{*},\eta^{*},\lambda^{*})$ solves the dual problem \eqref{ex:dual}. Define
$$X^*:=x^*_{h_A}(\lambda^*M^{\nu^*,\eta^*}_{\tau},Y_{\tau})=(\lambda^*M^{\nu^*,\eta^*}_{\tau}/h(Y_\tau))^{\frac{1}{\alpha-1}}.$$ 
By Proposition \ref{primaldualcons},  there exists a portfolio process $\pi^{\nu^*}$ such that the optimal wealth process $X^{\nu^*,*}_t$ for problem \eqref{A*.def} satisfies 
\begin{eqnarray}\label{ex2:X}
 X^{\nu^*,*}_t=\frac{1}{M^{\nu^*,\eta^*}_t}\mathbb{E}\left[X^*M^{\nu^*,\eta^*}_{\tau}|\mathcal{F}_t\right],
\end{eqnarray}
and
\begin{eqnarray}\label{ex2:dX}
dX^{\nu^*,*}_t 
\hspace{-0.3cm}&=&\hspace{-0.3cm} 
\left[r(Y_t)+(\pi_t^{\nu^*})^{\mathsf{T}}(\mu(Y_t)-r(Y_t)\mathbf{1})\right]X^{\nu^*,*}_tdt
+X^{\nu^*,*}_t(\pi_t^{\nu^*})^{\mathsf{T}}\sigma(Y_t) dW_{1t},
\end{eqnarray}
with 
\begin{align}
\label{7.16.w}
    X^{\nu^*,*}_{0}=1,\quad X^{\nu^*,*}_{\tau}=X^*.
\end{align}
By the definitions of $\mathcal{X}$ and $\mathcal{Y}$, we have 
\begin{eqnarray}\label{ex2:M=Y}
\lambda^*M_t^{\nu^*,\eta^*}=\mathcal{Y}(t,Y_t,\mathcal{X}(t,Y_t,\lambda^*M_t^{\nu^*,\eta^*})),
\end{eqnarray}
and 
\begin{eqnarray}\label{ex2:macthX}
\mathcal{X}(t,Y_t,\lambda^*M_t^{\nu^*,\eta^*})
\hspace{-0.3cm}&=&\hspace{-0.3cm}
\frac{1}{\lambda^*M_t^{\nu^*,\eta^*}}\mathbb{E}\left[(\lambda^*M_t^{\nu^*,\eta^*})^{\frac{\alpha}{\alpha-1}}h(Y_\tau)^{\frac{1}{1-\alpha}}\bigg|\mathcal{F}_t\right].
\end{eqnarray}
Combining \eqref{ex2:X} and \eqref{ex2:macthX} gives $X^{\nu^*,*}_t=\mathcal{X}(t,Y_t,\lambda^*M_t^{\nu^*,\eta^*})$, which along with \eqref{ex2:M=Y} implies
\begin{eqnarray}\label{ex2:y=y}
\lambda^*M_t^{\nu^*,\eta^*}=\mathcal{Y}(t,Y_t,X^{\nu^*,*}_t).
\end{eqnarray}
In light of It\^o's formula, \eqref{SDE2.2}, \eqref{ex2:dX} as well as \eqref{ex2:y=y}, we deduce that
\begin{eqnarray}\label{ex2:dM2}
\hspace{-0.3cm}&&\hspace{-0.3cm}
\lambda^*dM_t^{\nu^*,\eta^*}
=
d\mathcal{Y}(t,Y_t,X^{\nu^*,*}_t)
\nonumber\\
\hspace{-0.3cm}&=&\hspace{-0.3cm}
\mathcal{Y}_t(t,Y_t,X^{\nu^*,*}_t)dt+\mathcal{Y}_x(t,Y_t,X^{\nu^*,*}_t)dX^{\nu^*,*}_t+\frac{1}{2}\mathcal{Y}_{xx}(t,Y_t,X^{\nu^*,*}_t)d\langle X^{\nu^*,*},X^{\nu^*,*} \rangle_t
\nonumber\\
\hspace{-0.3cm}&&\hspace{-0.3cm}
+\mathcal{Y}_y(t,Y_t,X^{\nu^*,*}_t)dY_t+\frac{1}{2}\mathcal{Y}_{yy}(t,Y_t,X^{\nu^*,*}_t)d\langle Y,Y \rangle_t+\mathcal{Y}_{xy}(t,Y_t,X^{\nu^*,*}_t)d\langle X^{\nu^*,*},Y\rangle_t
\nonumber\\
\hspace{-0.3cm}&=&\hspace{-0.3cm}
\mathcal{Y}_t(t,Y_t,X^{\nu^*,*}_t)dt+\left[r(Y_t)+(\pi^{\nu^*}_t)^{\mathsf{T}}(\mu(Y_t)-r(Y_t)\mathbf{1})\right]X^{\nu^*,*}_t\mathcal{Y}_x(t,Y_t,X^{\nu^*,*}_t)dt
\nonumber\\
\hspace{-0.3cm}&&\hspace{-0.3cm}
+\frac{1}{2}\mathcal{Y}_{xx}(t,Y_t,X^{\nu^*,*}_t)(X^{\nu^*,*}_t)^2\|(\pi_t^{\nu^*})^{\mathsf{T}}\sigma(Y_t)\|^2dt+\mathcal{Y}_{x}(t,Y_t,X^{\nu^*,*}_t)X^{\nu^*,*}_t(\pi_t^{\nu^*})^{\mathsf{T}}\sigma(Y_t) dW_{1t}
\nonumber\\
\hspace{-0.3cm}&&\hspace{-0.3cm}
+\mathcal{Y}_y(t,Y_t,X^{\nu^*,*}_t)b(Y_t)dt+\mathcal{Y}_y(t,Y_t,X^{\nu^*,*}_t)\beta(Y_t)qdW_{1t}+\mathcal{Y}_y(t,Y_t,X^{\nu^*,*}_t)\beta(Y_t)\sqrt{1-\|q\|^2}dW_{2t}
\nonumber\\
\hspace{-0.3cm}&&\hspace{-0.3cm}
+\frac{1}{2}\mathcal{Y}_{yy}(t,Y_t,X^{\nu^*,*}_t)\beta(Y_t)^2dt+\mathcal{Y}_{xy}(t,Y_t,X^{\nu^*,*}_t)X^{\nu^*,*}_t\beta(Y_t)q\sigma(Y_t)^{\mathsf{T}}\pi^{\nu^*}_tdt.
\end{eqnarray}
On the other hand, by \eqref{dM2} and \eqref{ex2:y=y}, we have
\begin{eqnarray}\label{ex2:dM}
\lambda^*dM_t^{\nu^*,\eta^*}
\hspace{-0.3cm}&=&\hspace{-0.3cm}
\lambda^*M_t^{\nu^*,\eta^*}\left[ -(r(Y_t)+ \delta(\nu^*_t))dt - (\theta^{\nu^*}(Y_t))^\top dW_{1t} + \eta^*_t dW_{2t} \right]
\nonumber\\
\hspace{-0.3cm}&=&\hspace{-0.3cm}
\mathcal{Y}(t,Y_t,X^{\nu^*,*}_t)\left[ -(r(Y_t)+ \delta(\nu^*_t))dt - (\theta^{\nu^*}(Y_t))^\top dW_{1t} + \eta^*_t dW_{2t} \right].
\end{eqnarray}
Comparing \eqref{ex2:dM2} and \eqref{ex2:dM}, we obtain 
\begin{eqnarray}
\label{6.55.w}
\hspace{-0.3cm}&&\hspace{-0.3cm}
\mathcal{Y}_{x}(t,Y_t,X^{\nu^*,*}_t)X^{\nu^*,*}_t(\pi_t^{\nu^*})^{\mathsf{T}}\sigma(Y_t) +\mathcal{Y}_y(t,Y_t,X^{\nu^*,*}_t)\beta(Y_t)q
\nonumber\\
\hspace{-0.3cm}&=&\hspace{-0.3cm}
-\mathcal{Y}(t,Y_t,X^{\nu^*,*}_t)(\theta(Y_t)+\sigma(Y_t)^{-1}\nu_t^*)^{\mathsf{T}},
\end{eqnarray}
and
\begin{eqnarray}
\label{6.56.w}
\mathcal{Y}_y(t,Y_t,X^{\nu^*,*}_t)\beta(Y_t)\sqrt{1-\|q\|^2}=\mathcal{Y}(t,Y_t,X^{\nu^*,*}_t)\eta^*_t.
\end{eqnarray}
From \eqref{6.47.ww}, \eqref{6.55.w}, and \eqref{6.56.w}, it follows that
\begin{eqnarray}\label{ex2:pi}
\pi^{\nu^*}_t
\hspace{-0.3cm}&=&\hspace{-0.3cm}
-(\sigma(Y_t)\sigma(Y_t)^{\mathsf{T}})^{-1}[\mu(Y_t)-r(Y_t)\mathbf{1}+\nu_t^*]\frac{\mathcal{Y}(t,Y_t,X^{\nu^*,*}_t)}{\mathcal{Y}_{x}(t,Y_t,X^{\nu^*,*}_t)X^{\nu^*,*}_t}
\nonumber\\
\hspace{-0.3cm}&&\hspace{-0.3cm}
-(q\sigma(Y_t)^{-1})^{\mathsf{T}}\frac{\mathcal{Y}_{y}(t,Y_t,X^{\nu^*,*}_t)\beta(Y_t)}{\mathcal{Y}_{x}(t,Y_t,X^{\nu^*,*}_t) X^{\nu^*,*}_t}\nonumber\\
\hspace{-0.3cm}&=&\hspace{-0.3cm}
-(\sigma(Y_t)\sigma(Y_t)^{\mathsf{T}})^{-1}[\mu(Y_t)-r(Y_t)\mathbf{1}+\nu_t^*]\frac{\mathcal{Y}(t,Y_t,X^{\nu^*,*}_t)}{\mathcal{Y}_{x}(t,Y_t,X^{\nu^*,*}_t)X^{\nu^*,*}_t}
\nonumber\\
\hspace{-0.3cm}&&\hspace{-0.3cm}
-(q\sigma(Y_t)^{-1})^{\mathsf{T}}\frac{\mathcal{Y}(t,Y_t,X^{\nu^*,*}_t)\eta^*_t}{\mathcal{Y}_{x}(t,Y_t,X^{\nu^*,*}_t) X^{\nu^*,*}_t\sqrt{1-\|q\|^2}}
\nonumber\\
\hspace{-0.3cm}&=&\hspace{-0.3cm}
\frac{1}{1-\alpha}\left[(\sigma(Y_t)\sigma(Y_t)^{\mathsf{T}})^{-1}[\mu(Y_t)-r(Y_t)\mathbf{1}+\nu_t^*]+\frac{(q\sigma(Y_t)^{-1})^{\mathsf{T}}\eta^*_t}{\sqrt{1-\|q\|^2}}\right].
\end{eqnarray}
In the special case when the interest rate $r(\cdot)\equiv r>0$, the mean return $ \mu(\cdot)\equiv \mu$, and the volatility matrix $\sigma(\cdot)\equiv\sigma$ are constant (i.e., no stochastic factors), the bi-variate function $v(t,y)$ degenerates into a uni-variate function of $t$. In this case, the optimal portfolio strategy $\pi^{\nu^*}$ is given by \eqref{ex2:pi}, where
\begin{align}\label{ex2:dualsolution}
\left\{
\begin{array}{ll}
     \nu^*_t\equiv \arg\min_{\nu\in\tilde{K}}\left\{\|\theta+\sigma^{-1}\nu\|^2+2(1-\alpha)\delta(\nu)\right\}, &t\geq 0,  \\
     \eta^{*}_{t}\equiv 0,& t\geq 0.
\end{array}
  \right.
\end{align}

Combining the above arguments and Theorem \ref{thm3.1} yields the next result, whose proof is omitted.
\begin{prop}
\label{prop7.1}
    Suppose that $\gamma=1$. The solution to the auxiliary one-period terminal wealth
optimization problem \eqref{A*.def} \emph{(}or, equivalently, \eqref{7.1}\emph{)} is given by \eqref{ex2:dX} and \eqref{7.16.w}, where the corresponding optimal portfolio strategy is given by \eqref{ex:solution.dual} and \eqref{ex2:pi}. Furthermore, the optimal solution to the primal periodic evaluation problem \eqref{problem} is given by 
\begin{align}
     dX^{\nu^*,*}_t 
=& 
\left[r(Y_t)+(\pi_t^{\nu^*})^{\mathsf{T}}(\mu(Y_t)-r(Y_t)\mathbf{1})\right]X^{\nu^*,*}_tdt
+X^{\nu^*,*}_t(\pi_t^{\nu^*})^{\mathsf{T}}\sigma(Y_t) dW_{1t},\,\, t\in [T_{i},T_{i+1}], \,i\geq 0,\nonumber \\
X^{\nu^*,*}_{0}=&x,\quad X^{\nu^*,*}_{T_{i}}=X^{\nu^*,*}_{T_{i-1}}x^*_{h_{A^*}}\left(\lambda^*\frac{Z^{\nu^*,\eta^*}_{T_i}/B^{\nu^*}_{T_i}}{Z^{\nu^*,\eta^*}_{T_{i-1}}/B^{\nu^*}_{T_{i-1}}},Y_{T_{i}}\right),\quad i\geq 1, \nonumber
\end{align}
where $\nu^*$, $\eta^*$, and $\pi_t^{\nu^*}$ are given by \eqref{ex:solution.dual} and \eqref{ex2:pi}; $A^*$ is the unique fixed point of \eqref{A*.def}; and $\lambda^*=\lambda^*(\eta^*,\nu^*)$ with $\lambda^*(\cdot,\cdot)$ defined by \eqref{ex:lambda}.
\end{prop}

Similarly, for the case $\gamma=0$, we have the next result.

\begin{prop}
\label{prop7.2}
Suppose that $\gamma=0$.    The assertions of Proposition \ref{prop7.1} hold analogously, with the modification that problem \eqref{ex:prob} is substituted with
\begin{eqnarray}
\sup_{X\in\mathcal{U}_0(1,y)}\mathbb{E}\left[\frac{1}{\alpha}X_{\tau}^{\alpha}\left(h(Y_{\tau})+A^{*}(Y_{\tau})\right)\right],\nonumber
\end{eqnarray}
a problem that admits an identical treatment to \eqref{ex:prob}.
\end{prop}

\section{Numerical Illustrations}\label{sec:numerical}

This section present some numerical examples of the optimal portfolio under the ratio-type periodic evaluation. We first keep the same convex trading constraint and compare the stochastic factor model with the benchmark model without stochastic factors. We then keep the same stochastic factor and compare the optimal portfolio with and without convex trading constraints. In addition, we also report sensitivity results with respect to the risk aversion parameter $\alpha$. Finally, we illustrate the comparison of the induced optimal wealth processes between our ratio-type periodic evaluation and the difference-type periodic evaluation in \cite{TZ23}.

\subsection{Numerical setup}\label{subsec:numerical-parameters}

We consider a market model with two risky assets and a stochastic factor process. The stochastic factor $Y$ is interpreted as the business-cycle regimes: negative values represent recession states and positive values represent expansion states. Let  $\alpha=0.5$ and 
the evaluation period  be normalized to one year, i.e., $\tau=1$. The benchmark growth rate $\kappa=0.03$ is interpreted as a moderate annual hurdle rate for ratio-type performance evaluation.  We take the subjective discount rate $\rho=0.3$ and $\gamma=1$. The stochastic factor has mean-reverting drift $b(y)=-0.8y$, constant volatility $\beta(y)=0.25,$  the correlation parameter $q=(0.25,0.10)$.  The factor preference is given by
$$
h(y)=0.75+\frac{0.25}{1+\exp(-y)},
$$
which indicates that higher factor values slightly increase the overall performance measure. Note that $h$ is bounded away from zero and below one, in line with the
assumptions of the paper.
The market coefficients are specified by
\begin{align}
r(y)&=0.03+0.01\tanh(y),\nonumber\\
\mu(y)&=
\begin{pmatrix}
r(y)+0.035+0.025\tanh(y)\\
r(y)+0.030-0.015\tanh(y)
\end{pmatrix},\nonumber\\
\sigma(y)&=
\begin{pmatrix}
0.34-0.04\tanh(y) & 0.06\\
0.06 & 0.28+0.03\tanh(y)\nonumber
\end{pmatrix},
\end{align}
where $\tanh(y):=\frac{e^y-e^{-y}}{e^y+e^{-y}}$.
The quantities
\(\mu_i(y)-r(y)\), \(i=1,2\), represent the excess return of each risky asset over the risk-free asset. In this example,
\[
\mu_1(y)-r(y)=0.035+0.025\tanh(y),\qquad
\mu_2(y)-r(y)=0.030-0.015\tanh(y).
\]
Therefore, Stock 1 can be interpreted as a cyclical asset, as its excess return \(\mu_1(y)-r(y)\) increases when the factor moves to expansion states; while Stock 2 can be regarded as a defensive asset, as its excess return \(\mu_2(y)-r(y)\) is relatively higher during recession states. 

We consider the no-short-selling and
borrowing constraint
\begin{eqnarray}\label{constraint}
K=\left\{\pi\in\mathbb{R}^2_+:\ \pi_1+\pi_2\leq a\right\},
\qquad a=0.6.
\end{eqnarray}
In the unconstrained benchmark, we set
$K=\mathbb{R}^2$.

\subsection{Factor vs. Non-Factor models under trading constraints}
\label{subsec:numerical-factor}
In the first numerical experiment, we compare the portfolio weights in two risky assets within one period (from $t=0$ to $t=1$) in a stochastic-factor model and in  the benchmark model without stochastic factor. In the benchmark model, the market coefficients are fixed at the level $y=0$. Both models are subject to the same constraint \eqref{constraint}.

Figure \ref{fig:numerical-factor-comparison} shows that the benchmark model results in constant constrained portfolio within the period.  In contrast, the optimal portfolio in the model with stochastic factor reacts strongly to the realized stochastic factor path. For the simulated paths, 
note that $\mu_1(Y_t)-r(Y_t)$ is increasing in $Y_t$, and therefore Stock 1 becomes more attractive in expansion states. In Figure \ref{fig:numerical-factor-comparison}, this corresponds roughly to the interval $t\in[0.36,0.58]$, where $Y_t$ is clearly positive and reaches its peak around $t=0.42$. During this time period, the portfolio weight in Stock 1 climbs and the portfolio weight in Stock 2 drops.
Meanwhile, $\mu_2(Y_t)-r(Y_t)$ is decreasing in $Y_t$, and Stock 2 becomes relatively more attractive when $Y_t$ stays in the recession states. It is observed
in the later part of the figure such as $t\in[0.80,1]$,  the optimal allocation is transferred from Stock 1 to the defensive asset Stock 2.
The economic message of Figure \ref{fig:numerical-factor-comparison} is that the stochastic factor substantially changes the relative attractiveness of assets period by period and hence results various patterns of the constrained portfolio under periodic evaluation. Comparing with repeated constant portfolio in all periods, the factor model can produce more diverse portfolio behavior in each period based on the evolution of factor process during that period.

\begin{figure}[H]
\centering
\includegraphics[width=0.65\textwidth]{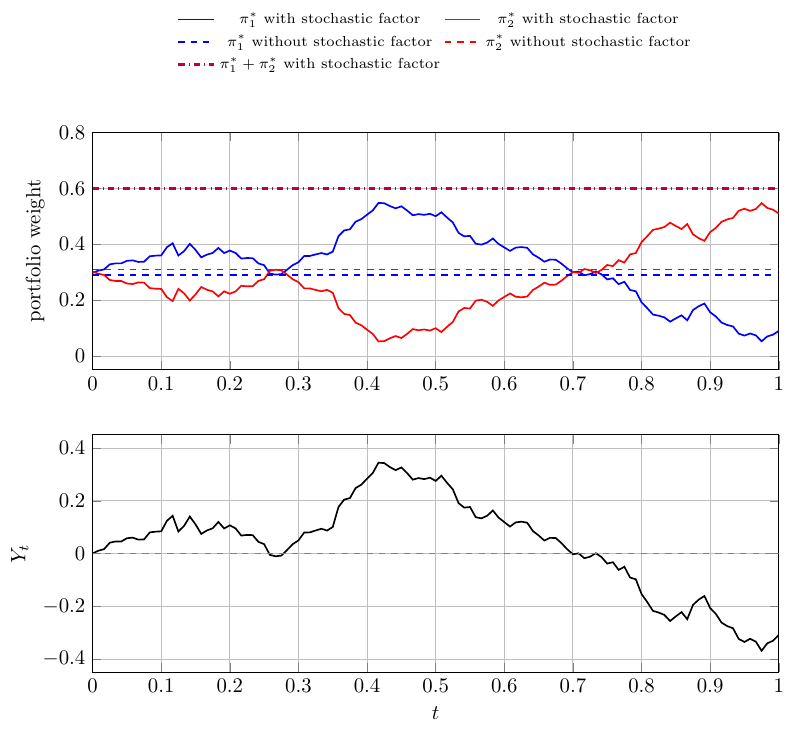}
\caption{Constrained optimal portfolio with and without stochastic factor}
\label{fig:numerical-factor-comparison}
\end{figure}

\subsection{Constrained vs. Unconstrained portfolios with stochastic factors}
\label{subsec:numerical-constraint}

The second example keeps the same stochastic factor dynamics and compares the optimal portfolio with and without the convex trading constraint during one period. 

Figure \ref{fig:numerical-constraint-comparison} shows that the convex constraint affects both the total risky exposure and the allocation weight between two risky assets. Without the constraint, the optimal strategy chooses $\pi_1+\pi_2$ close to 1 along the simulated path, and the total exposure becomes slightly larger than 1 in the recession state. When trading constraint is in force, the total risky exposure is reduced noticeably and the constraint clearly suppresses the total investment in risky assets.

More importantly, the constrained portfolio is not obtained by simply scaling down the unconstrained portfolio. The constraint also changes the relative weights in two assets. For example, during the early stage of the simulated path around $t=0.05$, the unconstrained strategy still assigns a slightly larger weight in Stock 2 than in Stock 1. After the constraint is imposed, however, the constrained strategy assigns a larger weight to Stock 1. This overturn occurs because the total risky allocation becomes scarce under the constraint. When $Y_t$ is mildly positive, the excess return $\mu_1(Y_t)-r(Y_t)$ of Stock 1 has already improved, so Stock 1 becomes more attractive at the margin once the investor cannot freely increase the total risky exposure. The convex constraint therefore forces a reallocation between the risky assets, rather than a brutal reduction in both assets.

The same phenomenon can be observed in the more extreme states. Around $t=0.42$, where the factor is strongly positive, the constrained portfolio puts most of the admissible risky exposure in Stock 1. Around $t=0.98$, where the factor is strongly negative, it shifts most of the admissible risky exposure to Stock 2. Hence, during that evaluation period, the constraint reduces the total risky investment while retains the factor-driven response in investment behavior. Again, these observations indicate that the factor model will result more diverse portfolio patterns period by period, heavily relying on the exogenous performance of factor process during each period.

\begin{figure}[H]
\centering
\includegraphics[width=0.65\textwidth]{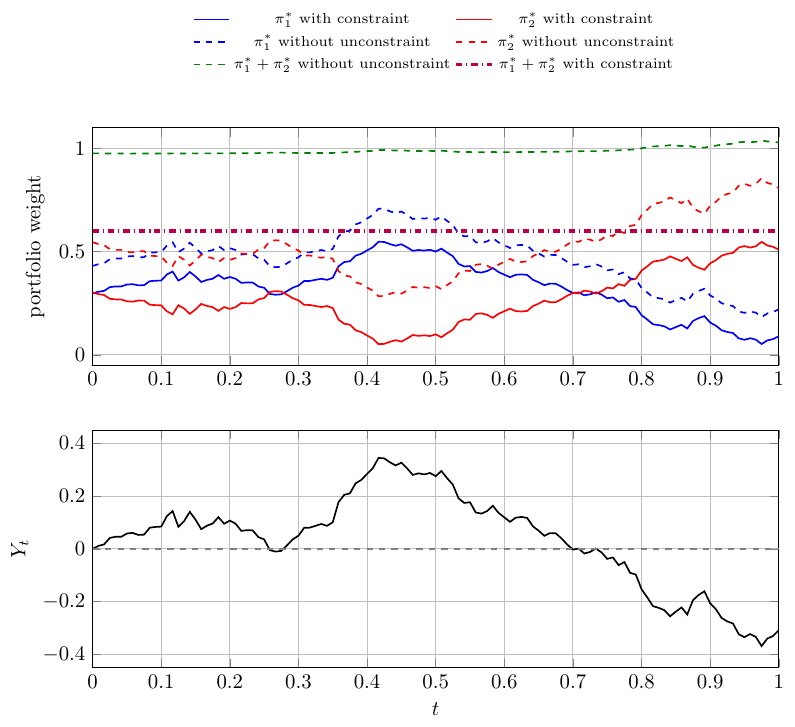}
\caption{Optimal portfolio with and without convex constraints}
\label{fig:numerical-constraint-comparison}
\end{figure}

\subsection{Sensitivity with respect to risk aversion parameter}
\label{subsec:numerical-sensitivity}
We next report the sensitivity results of the optimal portfolio strategy $\pi^*$ on the risk-aversion parameter $\alpha$ in Figure \ref{fig:numerical-alpha-sensitivity}. In the unconstrained case, shown in the right column, increasing $\alpha$ mainly enlarges the risky positions.  In the constrained case, on the other hand, is more subtle, as shown in the left column.  From their local behavior of paths, especially during the recession period around $t=0.8$ to $t=1$, the portfolio weights exhibit non‑monotonic trend in response to changes of $\alpha$. For instance, $\pi_1^*$ does not simply increase as $\alpha$ becomes larger; it may rise from a low value of $\alpha$ to an intermediate value and then fall when $\alpha$ is further increased. The reason is that a larger $\alpha$ increases the overall demand for risky assets, but the convex constraint caps the total risky exposure. Once the constraint becomes binding, a further increase in risk tolerance cannot be expressed by increasing $\pi_1^*+\pi_2^*$. Instead, it changes the weight between two risky assets.  In recession states such as $t\in[0.8,1]$, this reallocation favors the defensive asset Stock 2. Conversely, in expansion states during $t\in[0.4,0.6]$, the reallocation favors the cyclical asset Stock 1.

\begin{figure}[H]
\centering
\includegraphics[width=0.75\textwidth]{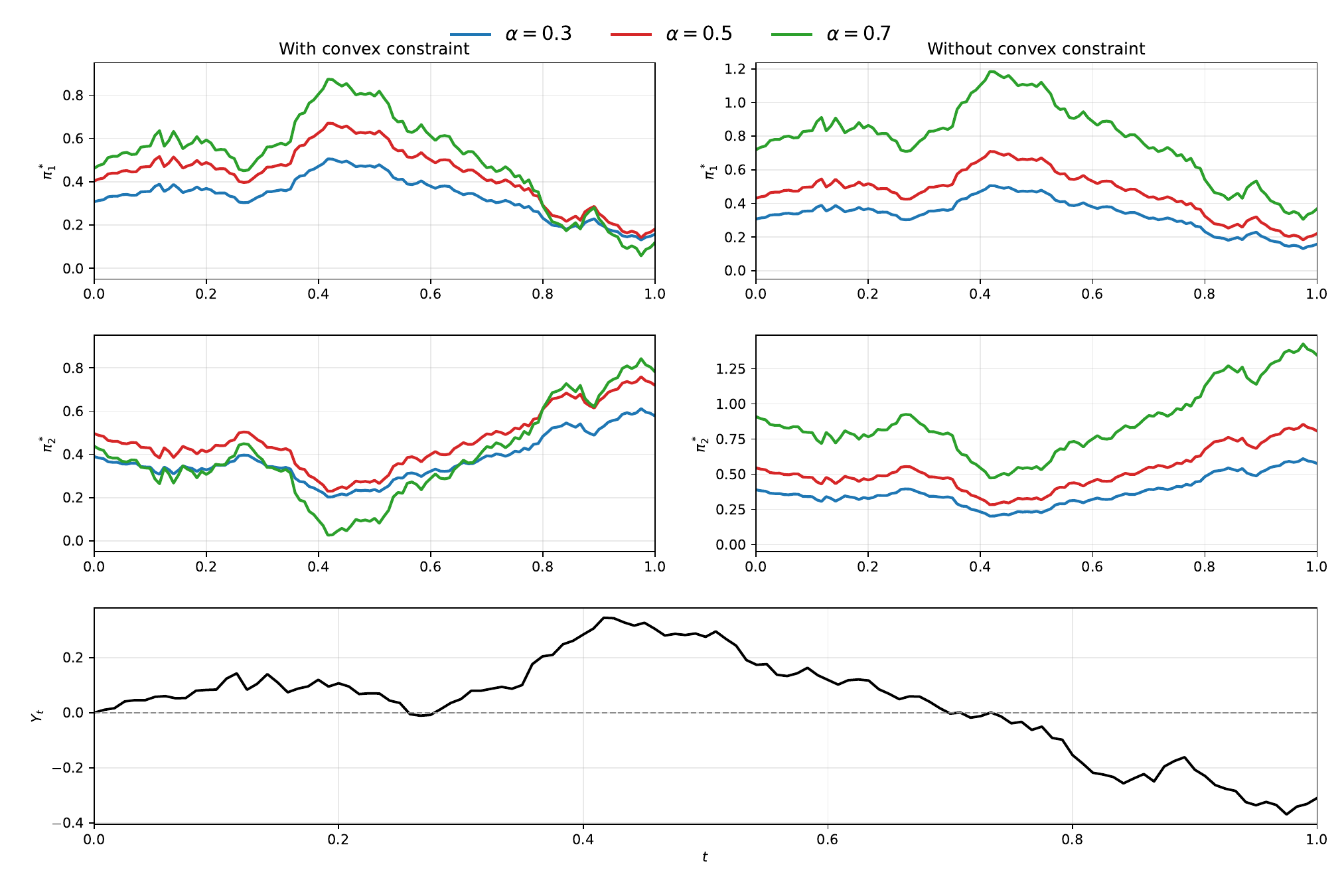}
\caption{Sensitivity of the optimal portfolio $\pi^*$ to $\alpha$ with and without convex constraints}
\label{fig:numerical-alpha-sensitivity}
\end{figure}

\subsection{Optimal wealth under ratio-type vs. difference-type criteria}
\label{subsec:numerical-difference}
 Finally, we conduct a numerical experiment to compare the induced optimal wealth processes under our ratio-type periodic evaluation and under the difference-type periodic evaluation in \cite{TZ23}. We consider the benchmark case $\gamma=1$ and $\kappa=0$. Then the ratio criterion reduces to the gross return over one evaluation period, $P_i=\frac{X_{T_i}}{X_{T_{i-1}}}.$ For the difference-type criterion, the performance variable is
$D_i=X_{T_i}-X_{T_{i-1}}.$ We illustrate the simulated results along the same stock-price path starting from three different initial wealth levels. 

Under the ratio-type criterion, the objective function creates a stronger incentive for the wealth process to catch up when the wealth level is low. The optimal strategy induced by such an objective therefore exhibits a state-dependent risk-taking incentive: the lower the wealth level, the stronger the incentive to take risk in pursuit of wealth accumulation. Figure \ref{fig:numerical-wealth-criteria} illustrates that the state-dependent risk-taking incentive can effectively promote wealth growth when stock market conditions are favorable (e.g., when $S_{t}/S_{0}\geq 1$). However, when market conditions deteriorate (e.g., when $S_{t}/S_{0}< 1$), this same incentive can expose the investor to excessive risk, thereby impeding wealth accumulation. Consequently, the ratio-type criterion is more effective under favorable market conditions, as it encourages the investor to exploit growth opportunities. By contrast, when market conditions are unfavorable, the more conservative strategy induced by the difference-type criterion is better suited to limiting wealth losses, as illustrated in Figure \ref{fig:numerical-wealth-unfavorable}.

\begin{figure}[H]
\centering
\includegraphics[width=0.95\textwidth]{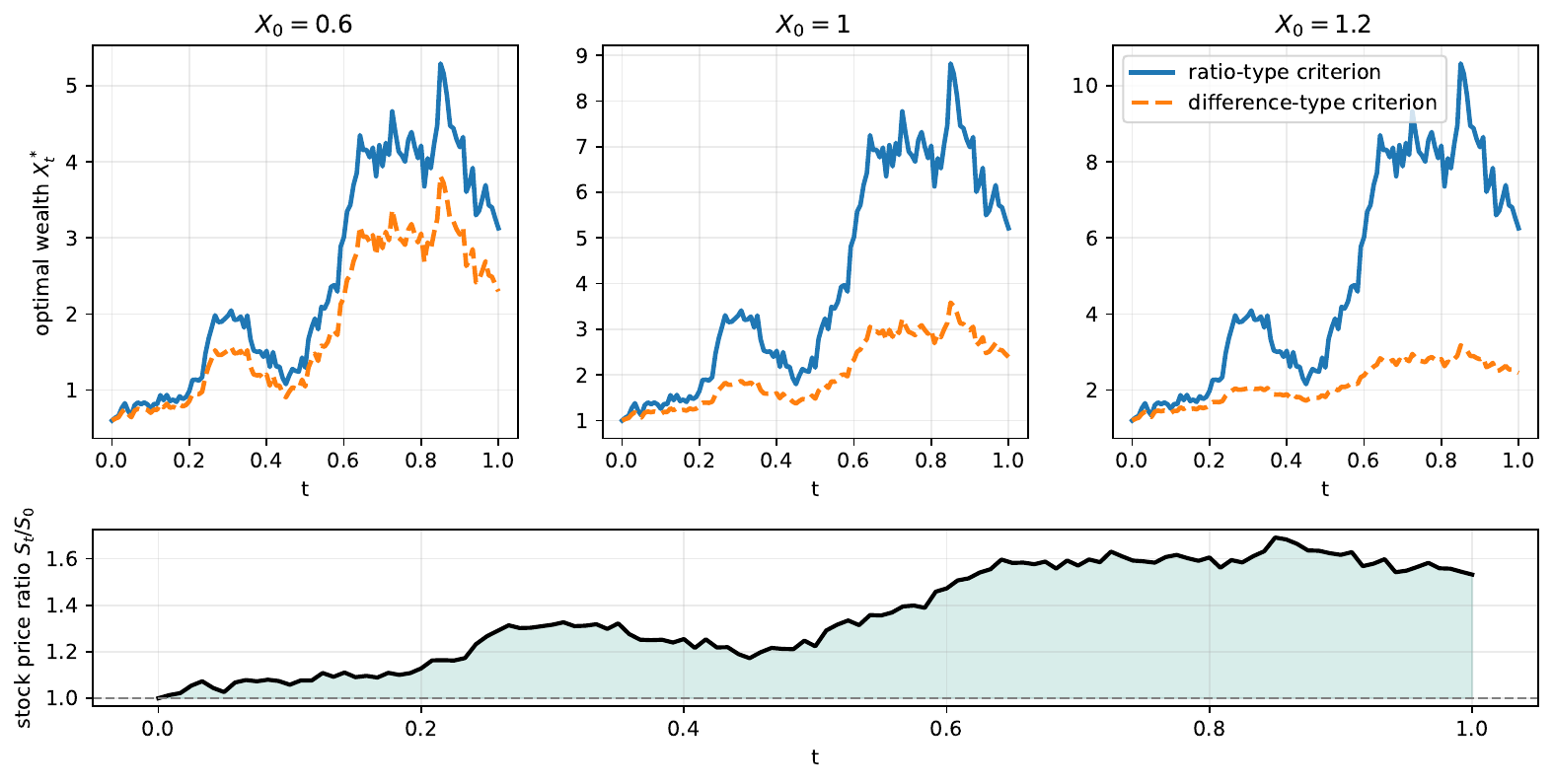}
\caption{Optimal wealth processes under the ratio-type and difference-type
criteria along the same favorable stock-price path}
\label{fig:numerical-wealth-criteria}
\end{figure}

\begin{figure}[H]
\centering
\includegraphics[width=0.95\textwidth]{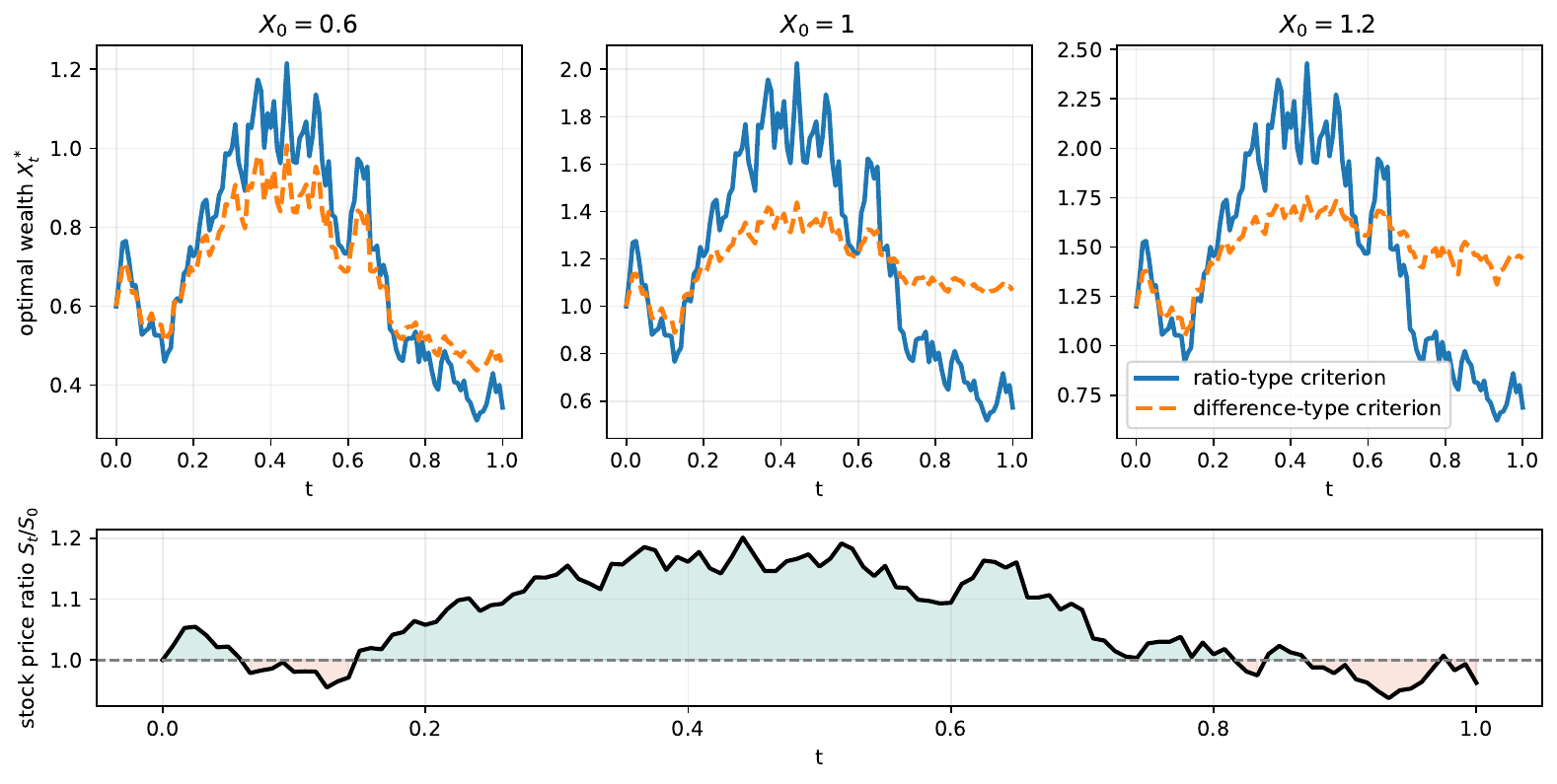}
\caption{Optimal wealth processes under the ratio-type and difference-type criteria along the same favorable/unfavorable stock-price path}
\label{fig:numerical-wealth-unfavorable}
\end{figure}

\ \\
\textbf{Acknowledgements}:  W. Wang is supported by the National Natural Science Foundation of China under no. 12571508, no. 12171405, and no. 11661074, and the Natural Science Foundation of Fujian Province
under no. 2024J01480. 
X. Yu is supported by the Hong Kong RGC General Research Fund (GRF) under grant no. 15304122 and grant no. 15306523 and by the Research Centre for
Quantitative Finance at the Hong Kong Polytechnic University under grant no. P0042708.

\appendix
\section{Technical Proofs}\label{app:technical-proofs}
\subsection{Proof of Proposition \ref{prop2.1}}\label{app:proof-prop21}
\begin{proof}
Let us assume \eqref{E=1} holds.
By Lemma \ref{lem2.1}, we know that, for any fixed $y\in\mathbb{R}$, the function $\mathbb{R}_{+}\ni x\mapsto h_{A}(x,y)$ is continuous and concave with $\frac{\partial}{\partial x}h_{A}(\infty,y)=0$. Hence, for each fixed $y\in\mathbb{R}$, $x_{h_A}^*(u,y)$ is the maximizer of the function $[0,\infty)\ni x\mapsto h_{A}(x,y)-ux$ for any $u\in\mathbb{R}_+$.
Then, for any $X\in\mathcal{F}_{\tau}^+$ and $(\eta,\lambda)\in\mathcal{H}\times\mathbb{R}_+$, using \eqref{2.8} and \eqref{phi(y)=h}, one can get
\begin{eqnarray}
\label{E<h}
h_A(X,Y_{\tau})-\lambda \frac{Z_{\tau}^{\nu,\eta}}{B^{\nu}_{\tau}}X\leq
h_A(x^*_{h_A}(\lambda\frac{Z_{\tau}^{\nu,\eta}}{B^{\nu}_{\tau}},Y_{\tau}),Y_{\tau})-\lambda x^*_{h_A}(\lambda\frac{Z_{\tau}^{\nu,\eta}}{B^{\nu}_{\tau}},Y_{\tau})\frac{Z_{\tau}^{\nu,\eta}}{B^{\nu}_{\tau}}.
\end{eqnarray}
By \eqref{dual.pro}, \eqref{E=1} and \eqref{E<h}, for each $(\eta,\lambda)\in\mathcal{H}\times\mathbb{R}_+$, it holds that
\begin{eqnarray}\label{dual.gap}
\inf_{\eta\in\mathcal{H},\lambda>0}L_{\nu}(\eta,\lambda)
\hspace{-0.3cm}&=&\hspace{-0.3cm}
\inf_{\eta\in\mathcal{H},\lambda>0}\left\{\mathbb{E}\left[h_A(x^*_{h_A}(\lambda\frac{Z^{\nu,\eta}_{\tau}}{B^{\nu}_{\tau}},Y_{\tau}),Y_{\tau})
-\lambda x^*_{h_A}(\lambda\frac{Z^{\nu,\eta}_{\tau}}{B^{\nu}_{\tau}},Y_{\tau})\frac{Z^{\nu,\eta}_{\tau}}{B^{\nu}_{\tau}}\right]+\lambda\right\}
\nonumber\\
\hspace{-0.3cm}&\leq&\hspace{-0.3cm}
\mathbb{E}\left[h_A(x^*_{h_A}(\lambda^*\frac{Z^{\nu,\eta^*}_{\tau}}{B^{\nu}_{\tau}},Y_{\tau}),Y_{\tau})
-\lambda^* x^*_{h_A}(\lambda^*\frac{Z^{\nu,\eta^*}_{\tau}}{B^{\nu}_{\tau}},Y_{\tau})\frac{Z^{\nu,\eta^*}_{\tau}}{B^{\nu}_{\tau}}\right]+\lambda^*
\nonumber\\
\hspace{-0.3cm}&=&\hspace{-0.3cm}
\mathbb{E}\left[\frac{1}{\alpha}\left(x^*_{h_A}(\lambda^*\frac{Z^{\nu,\eta^*}_{\tau}}{B^{\nu}_{\tau}},Y_{\tau})\right)^{\alpha}{h(Y_{\tau})}+{\frac{1}{\alpha}}A(Y_{\tau})\left(x^*_{h_A}(\lambda^*\frac{Z^{\nu,\eta^*}_{\tau}}{B^{\nu}_{\tau}},Y_{\tau})\right)^{\alpha(1-\gamma)}\right]
\nonumber\\
\hspace{-0.3cm}&\leq&\hspace{-0.3cm}
\sup_{X^{{\nu}}\in\Tilde{\mathcal{U}}^{\nu}_{0,\tau}(1,y)}\mathbb{E}\left[\frac{1}{\alpha}(X^{{\nu}})^{\alpha}{h(Y_{\tau})}+{\frac{1}{\alpha}}A(Y_{\tau})(X^{{\nu}})^{\alpha(1-\gamma)}\right],
\end{eqnarray}
which, together with \eqref{relationship}, implies that there is no duality gap. As a result, $X^{\nu,*}$ is the optimal solution to the unconstrained problem \eqref{problem3} and $(\eta^*,\lambda^*)$ is the optimal solution to the dual problem \eqref{dual.pro}. Hence, the first claim holds.

For the second claim, it is not hard to verify that the function $$\mathbb{R}_+\ni\lambda\mapsto\sup_{X\in\mathcal{F}_{\tau}^+}\mathbb{E}\left[h_A(X, Y_{\tau})-\lambda\frac{Z^{\nu,\eta}_{\tau}}{B^{\nu}_{\tau}}X\right]$$ is decreasing and convex. For any $y\in \mathbb{R}$ and $A(\cdot)\in{C}_b^{+}(\mathbb{R})$, we put
\begin{align}
\label{def.ell}
    \ell_{A,y}(x):=&h_{A}(x,y)-x\frac{\partial}{\partial x}h_{A}(x,y)
    \nonumber\\
    =&\frac{1}{\alpha}x^{\alpha}h(y)+{\frac{1}{\alpha}}A(y)x^{\alpha(1-\gamma)}
    -x\left(x^{\alpha-1}{h(y)}+A(y)(1-\gamma)x^{\alpha(1-\gamma)-1}\right)
    \nonumber\\
    =&\left(\frac{1}{\alpha}-1\right)x^{\alpha}h(y)+{\left(\frac{1}{\alpha}-(1-\gamma)\right)}A(y)x^{\alpha(1-\gamma)},\quad x\in\mathbb{R}_+.
\end{align}
By the expression of \eqref{def.ell}, it is obvious that, for any $y\in \mathbb{R}$ and $A(\cdot)\in{C}_b^{+}(\mathbb{R})$, the function $(0,\infty)\ni x\mapsto \ell_{A,y}(x)$ is bounded on any bounded interval and increasing on $(0,\infty)$, $\lim_{x\rightarrow\infty}\ell_{A,y}(x)=\infty\,{(0,\,\text{\text{resp.}})}$, and, $\lim_{x\rightarrow0+}\ell_{A,y}(x)=0 \,{(-\infty,\,\text{resp.})}$ for $\alpha\in(0,1)\, {(\alpha\in (-\infty,0),\,\text{resp.})}$.
Therefore, by Lemma 4.2 in \cite{KL91}, Lemma \ref{lem2.1}, \eqref{E<h}, \eqref{def.ell} and the monotone convergence theorem, it holds that
\begin{align}
\label{2.36.v0}
&\lim_{\lambda\rightarrow0+}\sup_{X\in\mathcal{F}_{\tau}^+}\mathbb{E}\left[h_A(X, Y_{\tau})-\lambda\frac{XZ^{\nu,\eta}_{\tau}}{B^{\nu}_{\tau}}\right]
\nonumber\\
=&
\lim_{\lambda\rightarrow0+}
\mathbb{E}\left[h_A(x^*_{h_A}(\lambda\frac{Z^{\nu,\eta}_{\tau}}{B^{\nu}_{\tau}},Y_{\tau}),Y_{\tau})
-\lambda x^*_{h_A}(\lambda\frac{Z^{\nu,\eta}_{\tau}}{B^{\nu}_{\tau}},Y_{\tau})\frac{Z^{\nu,\eta}_{\tau}}{B^{\nu}_{\tau}}\right]
\nonumber\\
=&
\lim_{\lambda\rightarrow0+}
\mathbb{E}\left[h_A(x^*_{h_A}(\lambda\frac{Z^{\nu,\eta}_{\tau}}{B^{\nu}_{\tau}},Y_{\tau}),Y_{\tau})
-x^*_{h_A}(\lambda\frac{Z^{\nu,\eta}_{\tau}}{B^{\nu}_{\tau}},Y_{\tau})\frac{\partial}{\partial x}h_{A}(x^*_{h_A}(\lambda\frac{Z^{\nu,\eta}_{\tau}}{B^{\nu}_{\tau}},Y_{\tau}),Y_{\tau})\right]
\nonumber\\
=&
\lim_{\lambda\rightarrow0+}
\mathbb{E}\left[\ell_{A,Y_{\tau}}(x^*_{h_A}(\lambda\frac{Z^{\nu,\eta}_{\tau}}{B^{\nu}_{\tau}},Y_{\tau}))\right]
\nonumber\\
=&\infty\,{(0,\,\text{resp.})},
\end{align}
for $\alpha\in(0,1)\,{(\alpha\in(-\infty,0),\,\text{resp.})}$, because we have $x_{h_A}^*(u,y)\uparrow\infty$ as $u\downarrow0$.
Similarly, we have
\begin{align}
\label{2.37.v0}
\lim_{\lambda\rightarrow\infty}\sup_{X\in\mathcal{F}_{\tau}^+}\mathbb{E}\left[h_A(X,Y_{\tau})-\lambda\frac{XZ^{\nu,\eta}_{\tau}}{B^{\nu}_{\tau}}\right]
=
\lim_{\lambda\rightarrow\infty}
\mathbb{E}\left[\ell_{A,Y_{\tau}}(x^*_{h_A}(\lambda\frac{Z^{\nu,\eta}_{\tau}}{B^{\nu}_{\tau}},Y_{\tau}))\right]=0\,{(-\infty,\,\text{resp.})},
\end{align}
for $\alpha\in(0,1)\,{(\alpha\in(-\infty,0),\,\text{resp.})}$, in view of $x_{h_A}^*(u,y)\downarrow0$ as $u\uparrow\infty$.
In addition, note that
\begin{eqnarray}
\frac{\partial}{\partial \lambda}L_{\nu}(\eta,\lambda)
\hspace{-0,3cm}&=&\hspace{-0.3cm}
\frac{\partial}{\partial \lambda}\left[\sup_{X\in\mathcal{F}_{\tau}^+}\mathbb{E}\left[h_A(X, Y_{\tau})-\lambda\frac{XZ^{\nu,\eta}_{\tau}}{B^{\nu}_{\tau}}\right]\right]+1
=
\frac{\partial}{\partial \lambda}\mathbb{E}\left[\ell_{A,Y_{\tau}}(x^*_{h_A}(\lambda\frac{Z^{\nu,\eta}_{\tau}}{B^{\nu}_{\tau}},Y_{\tau}))\right]+1
\nonumber\\
\hspace{-0.3cm}&=&\hspace{-0.3cm}
-\mathbb{E}\left[x^*_{h_A}\left(\lambda\frac{Z^{\nu,\eta}_{\tau}}{B^{\nu}_{\tau}},Y_{\tau}\right)\frac{Z^{\nu,\eta}_{\tau}}{B^{\nu}_{\tau}}\right]+1,\nonumber
\end{eqnarray}
which implies that there exists some $\lambda_0\in(0,\infty)$ such that $\frac{\partial}{\partial \lambda}L_{\nu}(\eta,\lambda)<0$ for $\lambda\in(0,\lambda_0)$ and $\frac{\partial}{\partial \lambda}L_{\nu}(\eta,\lambda)>0$ for $\lambda\in(\lambda_0,\infty)$ with $\lim_{\lambda\rightarrow\infty}\frac{\partial}{\partial \lambda}L_{\nu}(\eta,\lambda)=1$. This, together with \eqref{2.36.v0} and \eqref{2.37.v0}, yields that for $\eta\equiv\eta^*$, $$\lim_{\lambda\rightarrow0+}L_{\nu}(\eta^*,\lambda)=\lim_{\lambda\rightarrow\infty}L_{\nu}(\eta^*,\lambda)=\infty,\text{ for }\alpha\in(0,1),$$
and $$\lim_{\lambda\rightarrow0+}L_{\nu}(\eta^*,\lambda)=0,\,\,\lim_{\lambda\rightarrow\infty}L_{\nu}(\eta^*,\lambda)=\infty,\text{ for }\alpha\in(-\infty,0).$$ 
Hence, $L_{\nu}(\eta^*,\lambda)$ attains its infimum at the solution $\lambda^*\in(0,\infty)$.
Then, we have
\begin{eqnarray}
\hspace{-0.3cm}&&\hspace{-0.3cm}
\inf_{u\in\mathbb{R}_+}\left\{u\lambda+\sup_{X\in\mathcal{F}_{\tau}^+}\mathbb{E}\left[h_A(X,Y_{\tau})-u\lambda\frac{XZ^{\nu,\eta^*}_{\tau}}{B^{\nu}_{\tau}}\right]\right\}
\nonumber\\
\hspace{-0.3cm}&=&\hspace{-0.3cm}
\inf_{v\in\mathbb{R}_+}\left\{v+\sup_{X\in\mathcal{F}_{\tau}^+}\mathbb{E}\left[h_A(X,Y_{\tau})-v\frac{XZ^{\nu,\eta^*}_{\tau}}{B^{\nu}_{\tau}}\right]\right\}
=
L_{\nu}(\eta^*,\lambda^*).\nonumber
\end{eqnarray}
That is, the function 
\begin{align}
g(u)&:=u\lambda^*+\sup_{X\in\mathcal{F}_{\tau}^+}\mathbb{E}\left[h_A(X,Y_{\tau})-u\lambda^*\frac{XZ^{\nu,\eta^*}_{\tau}}{B^{\nu}_{\tau}}\right]
\nonumber\\
&=u\lambda^*+\mathbb{E}\left[
h_A(x^*_{h_A}(u\lambda^*\frac{Z^{\nu,\eta^*}_{\tau}}{B^{\nu}_{\tau}},Y_{\tau}),Y_{\tau})
-u\lambda^* x^*_{h_A}(u\lambda^*\frac{Z^{\nu,\eta^*}_{\tau}}{B^{\nu}_{\tau}},Y_{\tau})\frac{Z^{\nu,\eta^*}_{\tau}}{B^{\nu}_{\tau}}\right]\nonumber
\end{align}
achieves its infimum at $u=1$,
which gives that 
\begin{eqnarray}\label{g'1}
g^{\prime}(1)=\lambda^*-\lambda^*\mathbb{E}\left[x^*_{h_{A}}(\lambda^*\frac{Z^{\nu,\eta^*}_{\tau}}{B^{\nu}_{\tau}},Y_{\tau})\frac{Z^{\nu,\eta^*}_{\tau}}{B^{\nu}_{\tau}}\right]=0. 
\end{eqnarray}

Then, one can find some admissible portfolio $\hat{\pi}^{\nu}$ in the market $\mathcal{M}_{\nu}$ that finances $x^*_{h_{A}}(\lambda^*\frac{Z^{\nu,\eta^*}_{\tau}}{B^{\nu}_{\tau}},Y_{\tau})$ by an application of martingale representation theorem to the martingale $$\frac{X^{1,y,\hat{\pi}^{\nu}}_tZ^{\nu,{\eta}^*}_t}{B^{\nu}_t}=\mathbb{E}\left[x^*_{h_{A}}(\lambda^*\frac{Z^{\nu,\eta^*}_{\tau}}{B^{\nu}_{\tau}},Y_{\tau})\frac{Z^{\nu,\eta^*}_{\tau}}{B^{\nu}_{\tau}}\Bigg|\mathcal{F}_{t}\right],\quad t\in[0,\tau],$$
with $X^{1,y,\hat{\pi}^{\nu}}_0=1$ and $X^{1,y,\hat{\pi}^{\nu}}_\tau=x^*_{h_{A}}(\lambda^*\frac{Z^{\nu,\eta^*}_{\tau}}{B^{\nu}_{\tau}},Y_{\tau})$ (see Section 6 in \cite{KL91} for more details).
Furthermore, from \eqref{supermaringale}, we know $X^{1,y,\hat{\pi}^{\nu}}\frac{Z^{\nu,\eta}}{B^{\nu}}$ is a $\mathbb{P}$-supermartingale for any $\eta\in\mathcal{H}$. Hence, it holds that
$$\mathbb{E}\left[x^*_{h_{A}}(\lambda^*\frac{Z^{\nu,\eta}_{\tau}}{B^{\nu}_{\tau}},Y_{\tau})\frac{Z^{\nu,\eta}_{\tau}}{B^{\nu}_{\tau}}\right]=\mathbb{E}\left[X^{1,y,\hat{\pi}^{\nu}}_{\tau}\frac{Z^{\nu,\eta}_{\tau}}{B^{\nu}_{\tau}}\right]\leq 1,\quad \eta\in\mathcal{H},$$
which, combined with the arbitrary $\eta$, implies $x^*_{h_{A}}(\lambda^*\frac{Z^{\nu,\eta^*}_{\tau}}{B^{\nu}_{\tau}},Y_{\tau})\in\Tilde{\mathcal{U}}^{\nu}_{0,\tau}(1,y)$.
Consequently, together with \eqref{g'1}, we have \eqref{E=1} holds. Using the first claim, one knows that $X^{\nu,*}$ is the optimal solution to the unconstrained problem \eqref{problem3}.
\end{proof}

\subsection{Proof of Proposition \ref{prop4.5}}\label{app:proof-prop45}
\begin{proof}
By Theorem 2.3 in \cite{CH05}, 
there exists a portfolio process $\pi$ such that the resulting wealth process $X^{1,y,\pi,\nu^*}$ in the market $\mathcal{M}_{\nu^*}$ is given by
\begin{eqnarray}\label{eq:pro4.5:a}
X^{1,y,\pi,\nu^*}_t\frac{Z^{\nu^*,\eta^*}_t}{B^{\nu^*}_t}
\hspace{-0.3cm}&=&\hspace{-0.3cm}
\mathbb{E}\left[\left.X\frac{Z^{\nu^*,\eta^*}_\tau}{B^{\nu^*}_\tau}\right|\mathcal{F}_t\right]
\nonumber\\
\hspace{-0.3cm}&=&\hspace{-0.3cm}
1+\int_0^t\frac{X^{1,y,\pi,\nu^*}_sZ_s^{\nu^*,\eta^*}}{B^{\nu^*}_s}\left[(\sigma(Y_s)^{\mathsf{T}}\pi_s-\theta^{\nu^*}(Y_s))^{\mathsf{T}}dW_{1s}+\eta^*_sdW_{2s}\right], 
\end{eqnarray}
and satisfies
\begin{eqnarray}
dX^{1,y,\pi,\nu^*}_t 
\hspace{-0.3cm}&=&\hspace{-0.3cm} 
\left[r(Y_t)+\delta(\nu^*_t)+\pi_t^{\mathsf{T}}(\mu(Y_t)+\nu^*_t-r(Y_t)\mathbf{1})\right]X^{1,y,\pi,\nu^*}_tdt
\nonumber\\
\hspace{-0.3cm}&&\hspace{-0.3cm}
+X^{1,y,\pi,\nu^*}_t\pi^{\mathsf{T}}_t\sigma(Y_t)dW_{1t},
\end{eqnarray}
with $X^{1,y,\pi,\nu^*}_0=1$ and $X^{1,y,\pi,\nu^*}_\tau=X.$
To show $X^{1,y,\pi,\nu^*}\in\mathcal{U}_0(1,y)$, it is sufficient to prove that $\pi$ satisfies 
$$\pi\in K\quad\text{ and }\quad\delta(\nu^*_t)+\pi^{\mathsf{T}}_t\nu^*_t=0, \,\,a.s..$$
For any fixed $\nu\in\mathcal{D}$ and fixed $\epsilon\in(0,1)$, $n\in\mathbb{N}$, let us introduce the following notations 
\begin{align}
(\nu^*)^{(\nu,\epsilon,n)}_{t}&:=\nu^*_t+\epsilon(\nu_t-\nu^*_t)\mathbf{1}_{\{t\leq\tau_n\}},\quad
x(\nu):=\mathbb{E}\left[X\frac{Z^{\nu,\eta^*}_{\tau}}{B^{\nu}_{\tau}}\right],\nonumber\\
\hat{\delta}^{\nu}(\nu^*_t)&:=\begin{cases}
    -\delta(\nu^*_t), &\nu\equiv0, \\
    \delta(\nu_t-\nu^*_t),&\text{otherwise},
\end{cases}
\quad\,\,
L^{\nu}_t:=\int_0^t\hat{\delta}^{\nu}(\nu^*_s)ds,\nonumber\\
N^{\nu}_t&:=\int_0^t\left[\sigma(Y_s)^{-1}(\nu_s-\nu^*_s)\right]^{\mathsf{T}}dW^{\nu^*}_{1s},\nonumber
\end{align} 
where 
\begin{eqnarray}
\tau_n
\hspace{-0.3cm}&:=&\hspace{-0.3cm}
\tau\wedge\inf\left\{t\in[0,\tau];|L^{\nu}_t|\geq n,\text{ or }|N^{\nu}_t|\geq n,\text{ or }\int_0^t\|\sigma(Y_s)^{-1}(\nu_s-\nu^*_s)\|^2ds\geq n,\right.\nonumber\\
\hspace{-0.3cm}&&\hspace{-0.3cm}
\text{or }\int_0^t(X^{1,y,\pi,\nu^*}_s/B^{\nu^*}_s)^2\|\nu_t-\nu^*_t\|^2ds\geq n,
\text{ or }\int^{t}_0\|\theta(Y_s)+\sigma(Y_s)^{-1}\nu^*_s\|^2\geq n,
\text{ or }\int_0^t\|\eta^*_s\|^2ds\geq n,
\nonumber\\
\hspace{-0.3cm}&&\hspace{-0.3cm}
\left.\text{or }\int_0^t(X^{1,y,\pi,\nu^*}_s/B^{\nu^*}_s)^2\|\sigma(Y_s)^{-1}(\nu_s-\nu^*_s)
+(L^{\nu}_s+N^{\nu}_s)\sigma(Y_s)^{\mathsf{T}}\pi_s\|^2ds\geq n\right\},\quad n\in\mathbb{N}.
\end{eqnarray}
We next take $\nu\equiv\nu^*+\varsigma$ for arbitrary fixed $\varsigma\in\mathcal{D}$ and $\nu\equiv0$. For both choices, it holds that $\lim_{n\rightarrow\infty}\tau_n=\tau$ almost surely.
For  $\nu\equiv\nu^*+\varsigma$, using \eqref{subad.property}, one obtains that
\begin{eqnarray}
\delta(\nu^*_s+\epsilon(\nu_s-\nu^*_s))-\delta(\nu^*_s)
\hspace{-0.3cm}&\leq&\hspace{-0.3cm}
\delta(\nu^*_s)+\delta(\epsilon(\nu_s-\nu^*_s))-\delta(\nu^*_s)
\nonumber\\
\hspace{-0.3cm}&=&\hspace{-0.3cm}
\epsilon\delta(\nu_s-\nu^*_s)
=
\epsilon\hat{\delta}^{\nu}(\nu^*_s),
\end{eqnarray}
and, for $\nu\equiv0$, one has
\begin{eqnarray}
\delta(\nu^*_s+\epsilon(\nu_s-\nu^*_s))-\delta(\nu^*_s)
\hspace{-0.3cm}&=&\hspace{-0.3cm}\delta((1-\epsilon)\nu^*_s)-\delta(\nu^*_s)
\nonumber\\
\hspace{-0.3cm}&=&\hspace{-0.3cm}-\epsilon\delta(\nu^*_s)
=
\epsilon\hat{\delta}^{\nu}(\nu^*_s).
\end{eqnarray}

Hence, in either case, we have
\begin{eqnarray}\label{eq:pro4.5:h}
\frac{Z^{(\nu^*)^{(\nu,\epsilon,n)},\eta^*}_{t}/B^{(\nu^*)^{(\nu,\epsilon,n)}}_{t}}{Z^{\nu^*,\eta^*}_{t}/B^{\nu^*}_{t}}
\hspace{-0.3cm}&=&\hspace{-0.3cm}
\exp\left(-\int_0^{t\wedge\tau_n}\left(\delta(\nu^*_s+\epsilon(\nu-\nu^*))-\delta(\nu^*_s)\right)ds-\epsilon N^{\nu}_{t\wedge\tau_n}\right.
\nonumber\\
\hspace{-0.3cm}&&\hspace{1cm}
\left.-\frac{\epsilon^2}{2}\int_0^{t\wedge\tau_n}\|\sigma^{-1}(Y_s)(\nu_s-\nu^*_s)\|^2ds\right)
\nonumber\\
\hspace{-0.3cm}&\geq&\hspace{-0.3cm}
\exp\left(-\epsilon(L^{\nu}_{t\wedge\tau_n}+N^{\nu}_{t\wedge\tau_n})-\frac{\epsilon^2}{2}\int_0^{t\wedge\tau_n}\|\sigma^{-1}(Y_s)(\nu_s-\nu^*_s)\|^2ds\right)
\nonumber\\
\hspace{-0.3cm}&\geq&\hspace{-0.3cm}
e^{-3\epsilon n},
\end{eqnarray}
where in last inequality we have used the definition of the stopping time $\tau_n$ and the fact that $\epsilon\in(0,1)$. Furthermore, we have
\begin{eqnarray}
\frac{x(\nu^*)-x((\nu^*)^{(\nu,\epsilon,n)})}{\epsilon}
\hspace{-0.3cm}&=&\hspace{-0.3cm}
\mathbb{E}\left[\frac{XZ^{\nu^*,\eta^*}_{\tau}}{\epsilon B^{\nu^*}_{\tau}}\left(1-\frac{Z^{(\nu^*)^{(\nu,\epsilon,n)},\eta^*}_{\tau}/B^{(\nu^*)^{(\nu,\epsilon,n)}}_{\tau}}{Z^{\nu^*,\eta^*}_{\tau}/B^{\nu^*}_{\tau}}\right)\right]
\nonumber\\
\hspace{-0.3cm}&\leq&\hspace{-0.3cm}
\sup_{0<\epsilon<1}\frac{1-e^{-3\epsilon n}}{\epsilon}x(\nu^*)
\nonumber\\
\hspace{-0.3cm}&=&\hspace{-0.3cm}
\sup_{0<\epsilon<1}\frac{1-e^{-3\epsilon n}}{\epsilon}<\infty.
\end{eqnarray}
Hence, by Fatou's lemma, we have
\begin{eqnarray}\label{eq:pro4.5:b}
\limsup_{\epsilon\downarrow0}\frac{x(\nu^*)-x((\nu^*)^{(\nu,\epsilon,n)})}{\epsilon}
\hspace{-0.3cm}&\leq&\hspace{-0.3cm}
\mathbb{E}\left[\limsup_{\epsilon\downarrow0}\frac{XZ^{\nu^*,\eta^*}_{\tau}}{\epsilon B^{\nu^*}_{\tau}}\left(1-\frac{Z^{(\nu^*)^{(\nu,\epsilon,n)},\eta^*}_{\tau}/B^{(\nu^*)^{(\nu,\epsilon,n)}}_{\tau}}{Z^{\nu^*,\eta^*}_{\tau}/B^{\nu^*}_{\tau}}\right)\right]
\nonumber\\
\hspace{-0.3cm}&\leq&\hspace{-0.3cm}
\mathbb{E}\left[\frac{XZ^{\nu^*,\eta^*}_{\tau}}{ B^{\nu^*}_{\tau}}\left(L^{\nu}_{\tau_n}+N^{\nu}_{\tau_n}\right)\right].
\end{eqnarray}
By It\^o's lemma and \eqref{dX/B}, we have
\begin{eqnarray}
d\big(\frac{X^{1,y,\pi,\nu^*}_t}{B^{\nu^*}_t}(L_t^{\nu}+N_t^{\nu})\big)
\hspace{-0.3cm}&=&\hspace{-0.3cm}
\frac{X^{1,y,\pi,\nu^*}_t}{B^{\nu^*}_t}\pi^{\mathsf{T}}_t\sigma(Y_t)(L_t^{\nu}+N_t^{\nu})dW^{\nu^*}_{1t}+\frac{X^{1,y,\pi,\nu^*}_t}{B^{\nu^*}_t}(dL_t^{\nu}+dN_t^{\nu})
\nonumber\\
\hspace{-0.3cm}&&\hspace{-0.3cm}
+\frac{X^{1,y,\pi,\nu^*}_t}{B^{\nu^*}_t}\pi^{\mathsf{T}}_t(\nu_t-\nu^*_t)dt,
\end{eqnarray}
which implies that
\begin{eqnarray}
\frac{X^{1,y,\pi,\nu^*}_{\tau_n}}{B^{\nu^*}_{\tau_n}}(L^{\nu}_{\tau_n}+N^{\nu}_{\tau_n})
\hspace{-0.3cm}&=&\hspace{-0.3cm}
\int_0^{\tau_n}\frac{X^{1,y,\pi,\nu^*}_t}{B^{\nu^*}_t}\left(\pi^{\mathsf{T}}_t\sigma(Y_t)(L_t^{\nu}+N_t^{\nu})+\left[\sigma(Y_t)^{-1}(\nu_t-\nu^*_t)\right]^{\mathsf{T}}\right)dW^{\nu^*}_{1t}
\nonumber\\
\hspace{-0.3cm}&&\hspace{-0.3cm}
+\int_0^{\tau_n}\frac{X^{1,y,\pi,\nu^*}_t}{B^{\nu^*}_t}\left(\pi^{\mathsf{T}}_t(\nu_t-\nu^*_t)dt+dL^{\nu}_t\right).\nonumber
\end{eqnarray}
Taking the expectation under the probability measure $d\mathbb{P}_n^{\nu^*,\eta^*}=Z^{\nu^*,\eta^*}_{\tau_n}d\mathbb{P}$ on both sides of above equation, we obtain
\begin{eqnarray}\label{eq:pro4.5:c}
\mathbb{E}\left[\frac{X^{1,y,\pi,\nu^*}_{\tau_n}Z^{\nu^*,\eta^*}_{\tau_n}}{B^{\nu^*}_{\tau_n}}(L^{\nu}_{\tau_n}+N^{\nu}_{\tau_n})\right]=\mathbb{E}\left[\int_0^{\tau_n}\frac{X^{1,y,\pi,\nu^*}_tZ^{\nu^*,\eta^*}_t}{B^{\nu^*}_t}\left(\pi^{\mathsf{T}}_t(\nu_t-\nu^*_t)dt+dL^{\nu}_t\right)\right].
\end{eqnarray}
By \eqref{eq:pro4.5:a}, we have
$$\frac{Z^{\nu^*,\eta^*}_{\tau_n}}{B^{\nu^*}_{\tau_n}}X^{1,y,\pi,\nu^*}_{\tau_n}=\mathbb{E}\left[\left.\frac{Z^{\nu^*,\eta^*}_{\tau}}{B^{\nu^*}_{\tau}}X\right|\mathcal{F}_{\tau_n}\right],$$
which, together with \eqref{eq:pro4.5:c}, implies that
\begin{eqnarray}\label{eq:pro4.5:d}
\mathbb{E}\left[\frac{Z^{\nu^*,\eta^*}_{\tau}}{B^{\nu^*}_{\tau}}X(L^{\nu}_{\tau_n}+N^{\nu}_{\tau_n})\right]
\hspace{-0.3cm}&=&\hspace{-0.3cm}
\mathbb{E}\left[\mathbb{E}\left[\left.\frac{Z^{\nu^*,\eta^*}_{\tau}}{B^{\nu^*}_{\tau}}X\right|\mathcal{F}_{\tau_n}\right](L^{\nu}_{\tau_n}+N^{\nu}_{\tau_n})\right]
\nonumber\\
\hspace{-0.3cm}&=&\hspace{-0.3cm}
\mathbb{E}\left[\frac{X^{1,y,\pi,\nu^*}_{\tau_n}Z^{\nu^*,\eta^*}_{\tau_n}}{B^{\nu^*}_{\tau_n}}(L^{\nu}_{\tau_n}+N^{\nu}_{\tau_n})\right]
\nonumber\\
\hspace{-0.3cm}&=&\hspace{-0.3cm}
\mathbb{E}\left[\int_0^{\tau_n}\frac{X^{1,y,\pi,\nu^*}_tZ^{\nu^*,\eta^*}_t}{B^{\nu^*}_t}\left(\pi^{\mathsf{T}}_t(\nu_t-\nu^*_t)dt+dL^{\nu}_t\right)\right].
\end{eqnarray}
Hence, from \eqref{eq:pro4.5:b} and \eqref{eq:pro4.5:d}, we have
\begin{eqnarray}
\label{eq.5.18.w}
\limsup_{\epsilon\downarrow0}\frac{x(\nu^*)-x((\nu^*)^{(\nu,\epsilon,n)})}{\epsilon}
\hspace{-0.3cm}&\leq&\hspace{-0.3cm}
\mathbb{E}\left[\frac{XZ^{\nu^*,\eta^*}_{\tau}}{ B^{\nu^*}_{\tau}}\left(L^{\nu}_{\tau_n}+N^{\nu}_{\tau_n}\right)\right]
\nonumber\\
\hspace{-0.3cm}&=&\hspace{-0.3cm}
\mathbb{E}\left[\int_0^{\tau_n}\frac{X^{1,y,\pi,\nu^*}_tZ^{\nu^*,\eta^*}_t}{B^{\nu^*}_t}\left(\pi^{\mathsf{T}}_t(\nu_t-\nu^*_t)dt+dL^{\nu}_t\right)\right].
\end{eqnarray}
By the assumption \eqref{eq:pro4.5:e}, it holds that
$$\limsup_{\epsilon\downarrow0}\frac{x(\nu^*)-x((\nu^*)^{(\nu,\epsilon,n)})}{\epsilon}\geq0,$$
which together with \eqref{eq.5.18.w} yields that 
\begin{eqnarray}\label{eq:pro4.5:f}
\mathbb{E}\left[\int_0^{\tau_n}\frac{X^{1,y,\pi,\nu^*}_tZ^{\nu^*,\eta^*}_t}{B^{\nu^*}_t}\left(\pi^{\mathsf{T}}_t(\nu_t-\nu^*_t)dt+dL^{\nu}_t\right)\right]\geq0. 
\end{eqnarray}
Noting that $\nu\equiv\nu^*+\varsigma$, we have
\begin{eqnarray}\label{eq:pro4.5:g}
\mathbb{E}\left[\int_0^{\tau_n}\frac{X^{1,y,\pi,\nu^*}_tZ^{\nu^*,\eta^*}_t}{B^{\nu^*}_t}\left(\pi^{\mathsf{T}}_t\varsigma_tdt+\delta(\varsigma_t)dt\right)\right]\geq0,\quad n\in\mathbb{N},
\end{eqnarray}
and hence 
$$\pi^{\mathsf{T}}_t\varsigma_t+\delta(\varsigma_t)\geq0,\,a.e..$$
Then, using similar arguments of convex analysis as those in the proof of Theorem 9.1 in \cite{CK92}, we have $\pi\in K$.

On the other hand, for $\nu\equiv0$, it follows from \eqref{eq:pro4.5:f} that
$$\mathbb{E}\left[\int_0^{\tau_n}\frac{X^{1,y,\pi,\nu^*}_tZ^{\nu^*,\eta^*}_t}{B^{\nu^*}_t}\left(\pi^{\mathsf{T}}_t\nu^*_t+\delta(\nu^*_t)\right)dt \right]\leq0,\quad n\in\mathbb{N}.$$
Noting that \eqref{eq:pro4.5:g} is valid with $\varsigma$ replaced by $\nu^*$, we obtain $$\delta(\nu^*_t)+\pi^{\mathsf{T}}_t\nu^*_t=0, \quad a.e..$$
The proof is complete.
\end{proof}

\end{document}